\documentclass[12pt]{article}

\usepackage{def-080725}

\newcommand{\maintitle}{
Statistical Inference for Subgraph Frequencies of Exchangeable Hyperedge Models}

\iftoggle{blind}{
    \author{}
}{
\author{
Ayoushman Bhattacharya 
\thanks{Department of Statistics and Data Science, Washington University in St. Louis, USA} , 
Nilanjan Chakraborty 
\thanks{Department of Mathematics and Statistics, Missouri University of Science and Technology, USA}, 
and
Robert Lunde \textsuperscript{*}
}
}
\date{}

\begin{document}
\maketitle
\vspace{-4mm}
\begin{abstract}
In statistical network analysis, models for binary adjacency matrices satisfying vertex exchangeability are commonly used. However, such models may fail to capture key features of the data-generating process when interactions, rather than nodes, are fundamental units. We study statistical inference for subgraph counts under an exchangeable hyperedge model.  We introduce several classes of subgraph statistics for hypergraphs and develop inferential tools for subgraph frequencies that account for edge multiplicity. We show that a subclass of these subgraph statistics is robust to the deletion of low-degree nodes, enabling inference in settings where low-degree nodes are more likely to be missing. We also examine a more traditional notion of subgraph frequency that ignores multiplicity, showing that while inference based on limiting distributions is feasible in some cases, a non-degenerate limiting distribution may not exist in others. Empirically, we assess our methods through simulations and newly collected real-world hypergraph data on academic and movie collaborations, where our inferential tools outperform traditional approaches based on binary adjacency matrices. 
\end{abstract}

\section{Introduction} \label{sec-intro}

Complex relational data emerges in various disciplines including social sciences (\cite{dai2023hypergraph}), engineering (\cite*{GAO2024hypergraph}), and biology (\cite{feng2021hypergraph}), among others. Network data analysis has emerged as a key framework for analyzing this type of data; see, for example, \cite{kolaczyk2014} for an overview.  In network data analysis, it is common to convert relational data into a binary adjacency matrix and conduct inference on appropriate summary statistics of the graph.  However, this approach can lead to a substantial loss of information, as it fails to capture higher-order interactions in the network.  

In addition, many statistical models for binary adjacency matrices are most natural when the nodes are the fundamental units that are sampled to form the observed graph. However, in many instances, it may be more appropriate to view interactions rather than nodes as sampled units. For instance, an academic collaboration network is formed based on articles written by groups of authors; for such data, it is often the case that the network is constructed from a sample of papers.  Many other network datasets, including those corresponding to actor networks in movies (\cite*{ramasco2004self}), protein-interaction networks (\cite*{murgas2022hypergraph}), shopping transactions (\cite*{Amburg-2020-categorical}), and even the ingredients in cooking recipes, are often collected by sampling interactions rather than nodes.

Hypergraphs present a mathematical framework that naturally allows the modeling of higher-order interactions. A hypergraph consists of a set of nodes $\cV$ and a collection of hyperedges $\cH$, where a hyperedge $h \in \cH$ satisfies $h \subseteq \cV$, and unlike graphs, the cardinality of $h$ may be greater than two.  Recent developments in the statistical analysis of hypergraphs focus mainly on parametric models and further assume a uniform hypergraph, i.e., all hyperedges contain the same number of nodes (see \cite{ghoshdastidar2014consistency}, \cite{kim2018stochastic} and references therein). However, in practice, the size of the interactions often varies. In addition, some hyperedge models do not allow for the repetition of hyperedges, which is undesirable since, for example, the same group of authors might collaborate on multiple articles.

Existing work on non-uniform hypergraph models has focused mainly on clustering within the hypergraph (\cite{ghoshdastidar2017consistency}, \cite{ng2022model}, \cite{li2017inhomogeneous}). Hypergraph Stochastic Block Models (as discussed in \cite{kim2018stochastic}, \cite{pister2024stochastic}, \cite{cole2020exact}, \cite*{zhang2022sparse}, among others) and Latent Space Models (\cite*{turnbull2024latent} and \cite{yu2025modeling}) account for variations among vertices. To the best of our knowledge, no prior work has considered statistical inference for subgraph frequencies for hypergraphs.

Subgraphs are arguably the most important class of network statistics. In certain applications, particular subgraphs can be interpreted as functional subunits within the larger system. For social networks, triangle counts capture the frequency of mutual friendships.    Inference for subgraph frequencies in binary networks has been extensively studied in the literature. \cite*{bickel2011} obtained asymptotic distribution of subgraph frequencies for the widely studied sparse graphon model, and several authors have studied resampling procedures for associated parameters 
(\cite{bhattacharyya2015}, 
\cite{green2022bootstrapping}, 
\cite{levin2025bootstrapping}, \cite{lunde2023}.
Moreover, \cite*{bbb2022} study the asymptotic distribution of subgraphs for fixed population graphs under an induced vertex sampling scheme.

In contrast to binarized networks, hyperedges represent not only the frequency of interaction but also the type of interaction. For instance, a triangle in an author-collaboration network can be formed by three separate papers, two papers, or a single paper; with hypergraphs, we can construct summary statistics that distinguish between these cases.  The existing literature on subgraphs in hypergraphs (\cite*{lee2020hypergraph}, \cite*{lee2024hypergraph}, \cite*{juul2024hypergraph}) has primarily focused on computational properties and certain statistical properties under simple parametric models, and do not thoroughly study the asymptotic properties of subgraph frequencies in hypergraphs.

\subsection{Our Contributions} \label{sec-our-contri}
In this article, we study subgraph frequencies of hypergraphs under a general hyperedge exchangeable model that allows hyperedges of any length and repetition of hyperedges. Edge exchangeable models, introduced by \cite{Crane03072018}, offer a general framework for studying hypergraph data.   We introduce different classes of subgraphs in which the multiplicity of an edge is considered.  One of these classes involves novel edge-colored subgraph frequencies, which captures more fine-grained information about connectivity patterns compared to standard subgraph frequencies. Under mild assumptions, we obtain the asymptotic normality of these classes of subgraphs by exploiting a U-statistic representation of such subgraphs.  These results facilitate statistical inference in cases where interactions are sampled to form the network, which, as we argued previously, is common in practice.  

In addition, we introduce a class of subgraphs based on degree filtering. The study of degree-filtered subgraph frequencies is of interest primarily for two reasons.  First, missing data is ubiquitous in network analysis. For example, in social network analysis, respondents are known to have imperfect recall, in which they may not be able to recall all relevant contacts (see, for example, \cite{Bernard1984}).  As another example, movie databases often store partial information about the actors in a movie, focusing more on high-degree actors. Second, in many network datasets, there are several individuals who do not interact with others frequently. Removing those vertices leads to computationally efficient algorithms.  However, to our knowledge, the robustness of statistical inference under deletion of low-degree nodes has not been investigated theoretically in any framework.  We show that a certain subclass subgraph frequencies is highly robust to the deletion of low-degree nodes, in the sense that both the limiting distribution and target parameter are unaffected by the deletion of these nodes. Our asymptotic theory justifies the use of subsampling to conduct inference for a wide range of these subgraphs.        

We also explore the limiting behavior of subgraph frequencies for binarized hypergraphs that ignores the multiplicity of an edge in the network. We prove that a non-degenerate limiting distribution may often fail to exist for these subgraphs.  We also provide a central limit theorem for a certain class of binarized subgraph frequencies that sheds light on when limiting distributions may exist.  Finally, on the applied side, we collect new hypergraph datasets, which we believe would be useful to the broader community as a testbed for hypergraph methodology.    

\section{Problem Setup and Notation}\label{sec-prob-set-up}

Throughout the paper, we use standard asymptotic notation as follows. For any two positive sequences $\{a_m\}_{m \ge 1}$ and $\{b_m\}_{m \ge 1}$, $a_m = O(b_m)$ means $a_m \le C b_m$ for some $C >0$, $a_m = \omega(b_m)$ means $a_m/b_m \rightarrow \infty$ and $a_m = \Theta(b_m)$ means $c b_m \le a_m \le C b_m$, for large enough $m$ and  $0< c \leq C < \infty$. Moreover, $a_m = o(b_m)$ means $ a_m/b_m \rightarrow 0$.  When it is more convenient, we will also use $a_m \ll b_m$ to denote $a_m = o(b_m)$ and $a_m \gg b_m$ to denote $a_m = \omega(b_m)$. 

\subsection{Exchangeable Hyperedge Models} \label{sec-hyp-excg-model}
Exchangeable hyperedge models offer a natural modeling framework when interactions are fundamental units.  Exchangeable edge models have started to receive attention in various problems, including link prediction (\cite*{JMLR:v17:16-032}), anomaly detection (\cite*{pmlr-v204-luo23a}), and node clustering (\cite*{Zhang03042025}).  Various probabilistic properties of these models have been studied by \cite*{NIPS2016_1a0a283b}, \cite*{10.1214/18-EJS1455}, and \citet{janson-edge-exchangeable}. We introduce this model below.     

Let $\cV$ be a countable collection of vertices.  While real-world graphs involve a finite number of vertices, it is often natural to assume that the number of observed nodes would continue to grow as more hyperedges are sampled, making an infinite vertex set a natural choice for the underlying model.  Without loss of generality, we assume $\cV \subseteq\bbN$. A hyperedge is defined as a subset of $\cV$. Let $(\Omega, \cF, P)$ be a probability space and $h: (\Omega, \cF) \mapsto (\cS, \cP(\cS))$ denote a random hyperedge, where $\cS$ is the collection of all finite subsets of $\cV$ and $\cP(\cS)$ is the power set of $\cS$.  Note that $(\cS, \cP(\cS))$ is a Borel space, which will be useful shortly.    

Following \citet{Crane03072018}, we consider an exchangeable sequence of hyperedges $(h_i)_{i \geq 1}$; that is, for any bijection $\sigma:\bbN \mapsto\bbN$,
\begin{align*}
(h_{\sigma(i)})_{i \geq 1} \disteq (h_{i})_{i \geq 1}.
\end{align*}
Due to De Finetti's Theorem for random variables taking values in a Borel space (see, for example, Theorem 1.1 of \citet{kallenberg-invariance}), it follows that $(h_{i})_{i \geq 1}$ may be represented as a mixture of conditionally independent sequences. We restrict our attention to one component of the mixture, and consider a sample of the form 
\begin{align*}
h_1, \ldots, h_m \simiid P.
\end{align*}
We study statistical inference for subgraph frequencies in a nonparametric setup where assumptions on $P$ are kept to a minimum.   In the following section, we discuss notions of subgraph frequencies for hypergraphs.

\subsection{Subgraph Frequencies} \label{sec-subgraph-freq}

\subsubsection{Preliminaries} \label{sec-subgraph-freq-prelim}
We now prepare notation to define subgraph frequencies of interest. Since edge-colored hypergraphs are less common in the graph theory literature, our notation in this section will be non-standard and will be clearly defined to avoid ambiguity. In the combinatorics literature, (proper) edge colorings typically refer to assignments of colors to edges such that adjacent edges have different colors (see, for example, \citet{Alon_1993}), but we will not enforce this condition.   

\begin{figure}[H]
    \centering
    \begin{subfigure}{0.24\textwidth}
        \centering
        \includegraphics[width = 0.7\textwidth, page =4]{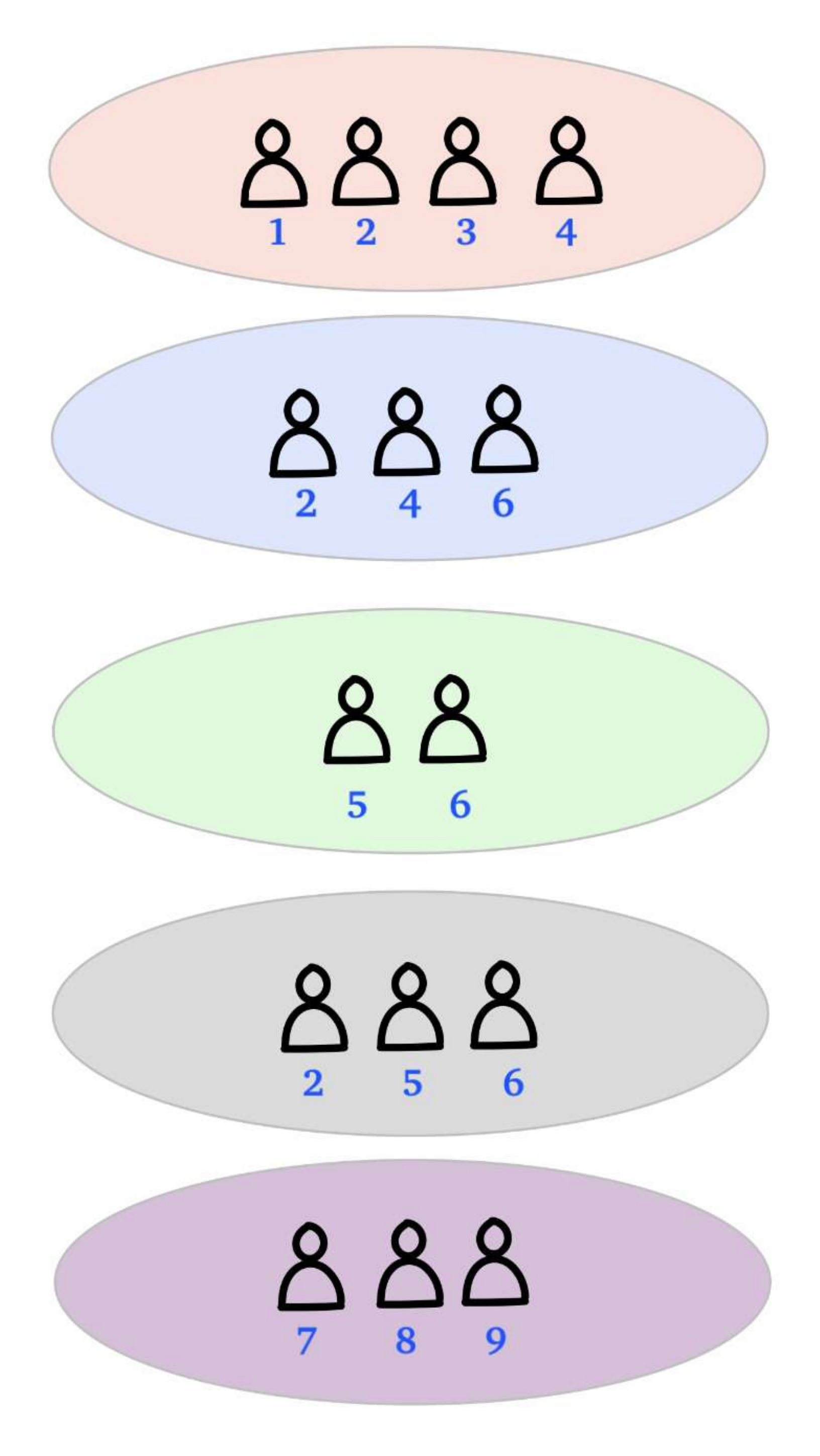}
        \subcaption{}
    \end{subfigure}
    \begin{subfigure}{0.24\textwidth}
        \centering
        \includegraphics[width = 0.7\textwidth, page =5]{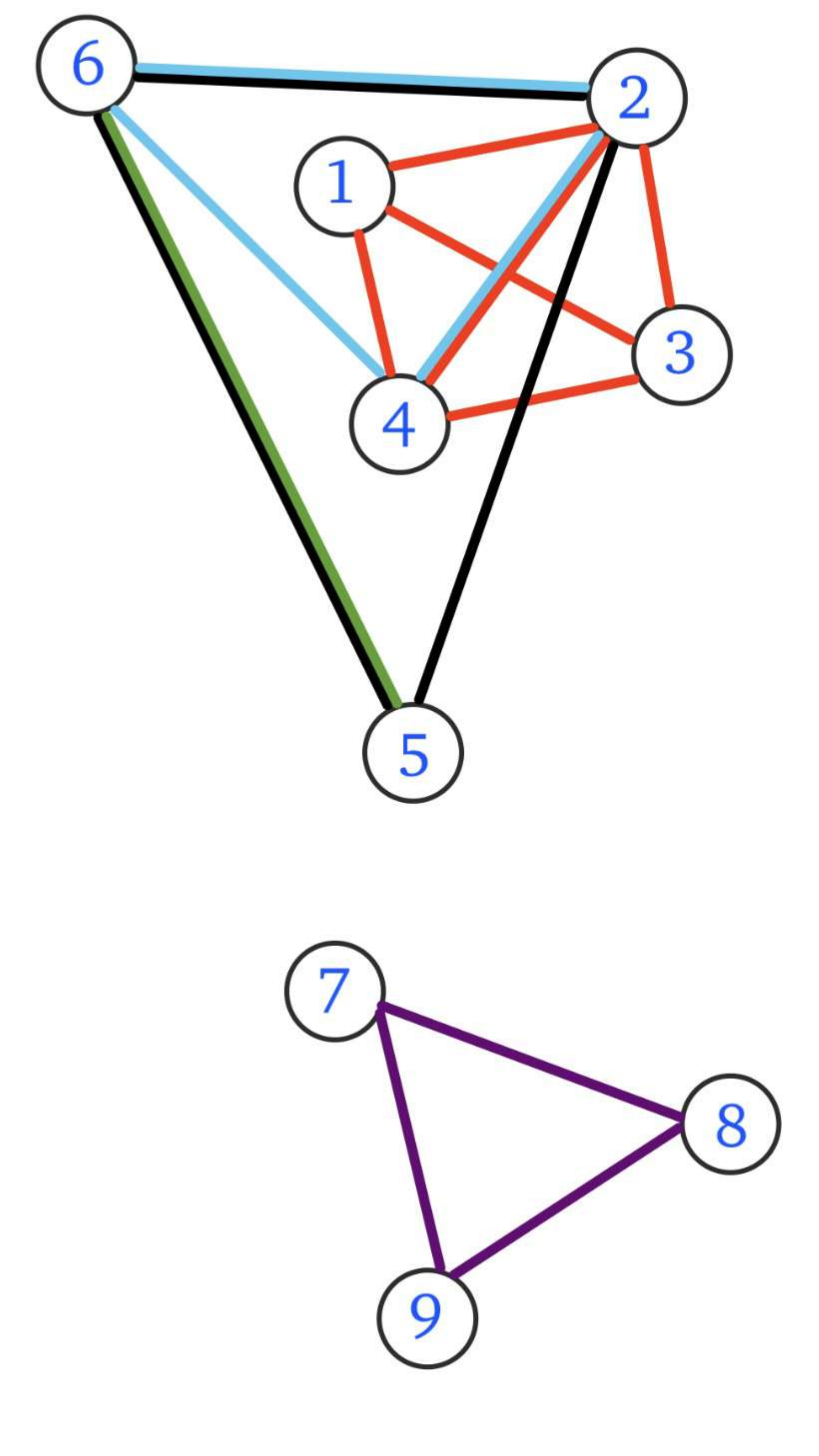}
        \subcaption{}
    \end{subfigure}
    \begin{subfigure}{0.24\textwidth}
        \centering
        \includegraphics[width = 0.7\textwidth, page =6]{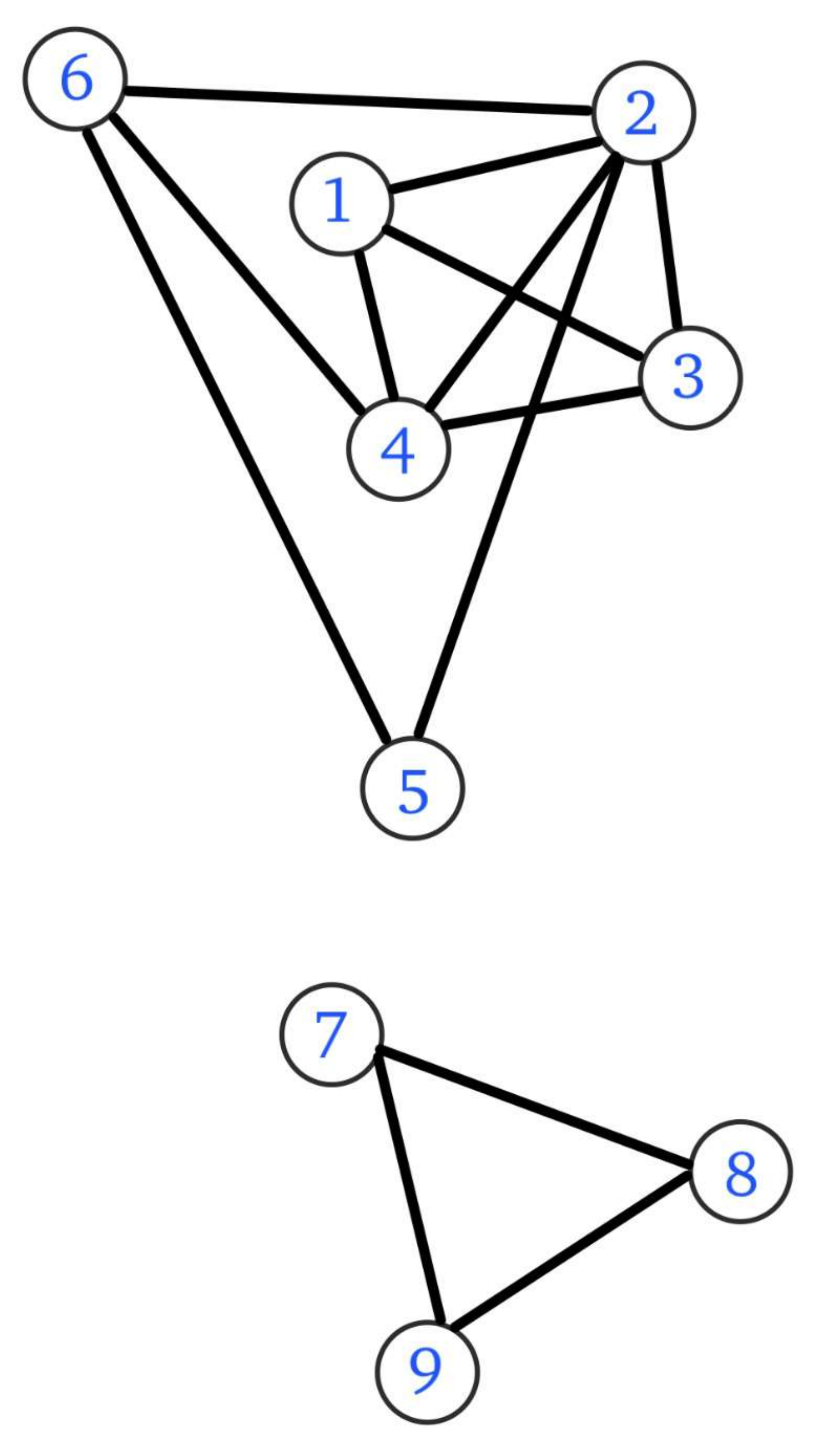}
        \subcaption{}
    \end{subfigure}
    \caption{Figure (a) represents the collaboration between individuals in the network. 
    Figure (b) represents the hypergraph for the collaboration network. Distinct colors are used to indicate different hyperedges. 
    Figure (c) presents the binarized version of the hypergraph. }
    \label{fig-hyp}
\end{figure}

Let $\cH_m = (h_1, \ldots, h_m)$ denote an $m$-tuple of hyperedges.  In this construction, information related to the ordering of the hyperedges is preserved; this ordering may be viewed as an assignment of a color $i$ to the $i$-th hyperedge $h_i$. The hyperedge tuple $\cH_m$ induces an edge-colored graph $G_\fC(\mathcal{H}) = (\cV(G_\fC), \cE(G_\fC))$, where $\cV(G_\fC) = \{a :  a \in h_i \text{ for some } i \in [m] \}$ and $
\cE(G_\fC) = \left\{ (\{a,b\},i) :  \{a,b \} \subseteq h_i \right\}$.  For notational convenience, we will often drop dependence on $\mathcal{H}_m$ and write $G_\fC(\mathcal{H}_m)$ as $G_\fC$.   

In this construction, an edge $\{a,b\}$ may have more than one color; the multiplicity of $\{a,b\}$ is given by $|\{ i :  \{a,b\} \subseteq h_i\}|$. A colored subgraph $K_\fC = (\cV(H_\fC), \cE(H_\fC))$ of $G_\fC$ satisfies $\cV(H_\fC) \subseteq \cV(G_\fC)$ and $\cE(H_\fC) \subseteq \cE(G_\fC)$.  If $K_\fC$ is a subgraph of $H_\fC$, we use the notation $K_\fC \subseteq H_\fC$; if $K_\fC \subseteq H_\fC$ and $ H_\fC \subseteq K_\fC$, we use the notation $K_\fC = H_\fC$.   We let $\cC(H_\fC) \subseteq \{1, \ldots, m\}$ denote the set of colors for which there is at least one edge in $H_\fC$. We refer to colored subgraphs for which each edge is assigned to a single color as simple.  While subgraph frequencies can be defined for both simple and non-simple colored subgraphs, we focus on the simple case, which is more interpretable.

We say that two colored subgraphs $K_\fC$ and $L_\fC$ are color isomorphic, denoted $K_\fC \isoc L_\fC$, if there exist bijections $\sigma: \cV_K \mapsto \cV_L$, $\pi: [m] \mapsto [m]$ such that 
\begin{align*}
(\{a,b\}, i) \in \cE(K_\fC) \iff (\{\sigma(a), \sigma(b)\}, \pi(i)) \in \cE(L_\fC) .
\end{align*}
Our definition of colored isomorphism is a natural extension of the standard notion of graph isomorphism; for colored graphs, we also require a bijection between colors of the two graphs.

\subsubsection{Colored Subgraph Frequencies} \label{sec-col-subgraph-freq}
In what follows, let $h_{i_1, \ldots, i_r}(j_1, \ldots, j_v)$ denote the colored subgraph induced by hyperedges $ \{ h_{i} :  i \in \cI\}$, where $\cI = \{i_1, \ldots, i_r\}$, and vertex set $\cJ = \{j_1, \ldots,  j_v\}$. Moreover, let $\cK_\fC(v, r)$ denote a complete graph on the vertices $1, \ldots ,v$ with $r$ distinct colors. Let $H_\fC$ be a simple colored subgraph of $\cK_\fC(v, r)$  with $v$ vertices, $e$ edges, and $r$ colors.  We propose the following notion of colored subgraph frequency:   

\begin{align*}
T({H_\fC}) &=  \frac{1}{\binom{m}{r}} \sum_{i_1<\dots<i_r} \sum_{j_1 < \cdots < j_v} \  
\sum_{K_\fC}
\indc{K_\fC \subseteq h_{i_1, \ldots, i_r}(j_1,\ldots, j_v)} \\ 
& = \frac{1}{{m \choose r}}\sum_{i_1<\dots<i_r} C(h_{i_1},\ldots, h_{i_r}; H_\fC), \text{ say,} \labthis \label{type-H-sub-hom}
\end{align*}
where $\sum_{K_\fC}$ is the summation over the set 
\begin{align*}
    \clr{K_\fC: K_\fC \isoc H_\fC,  \cV(K_\fC) = \{j_1, \ldots, j_v\}, \ \cC(K_\fC) = \{i_1, \ldots, i_r\} }.
\end{align*}
The colored subgraph frequency above is an unbiased estimate of the parameter 
\begin{align*}
    \theta({H_\fC}) = \Ex{C({h_1,\ldots, h_r}; H_\fC)}\labthis \label{type-H-sub-hom-param}, 
\end{align*}
which may be interpreted as the expected number of subgraphs that contain a copy of $H_\fC$, when the $r$ hyperedges are sampled uniformly from $P$.  Our proposed inferential target is closely related to notions that arise in the graph limits literature. In the notation of \cite*{BORGS20081801}, $\theta(H_\fC)$ may be viewed as a colored analog of the injective homomorphism density $t_\text{inj}(H,G)$.  Although similar, it should be noted that our parameter is related to frequency under uniform sampling of hyperedges, whereas the graph limit analogs may be interpreted as probabilities under node sampling (see, for example, \citet{diaconis-janson-exchangeable-graph}); therefore, the normalizations differ.  One may also define colored analogs of isomorphism densities that appear in the graph limit literature, but for brevity, we focus on homomorphism densities. 
The above expressions often simplify for subgraphs that are of interest in statistical applications.  As some running examples, we consider colored triangle and two-star frequencies, defined below.
\begin{figure}[!htb]
    \centering
    \begin{minipage}{0.3\textwidth}
        \centering
        \includegraphics[width=0.8\textwidth, page=1]{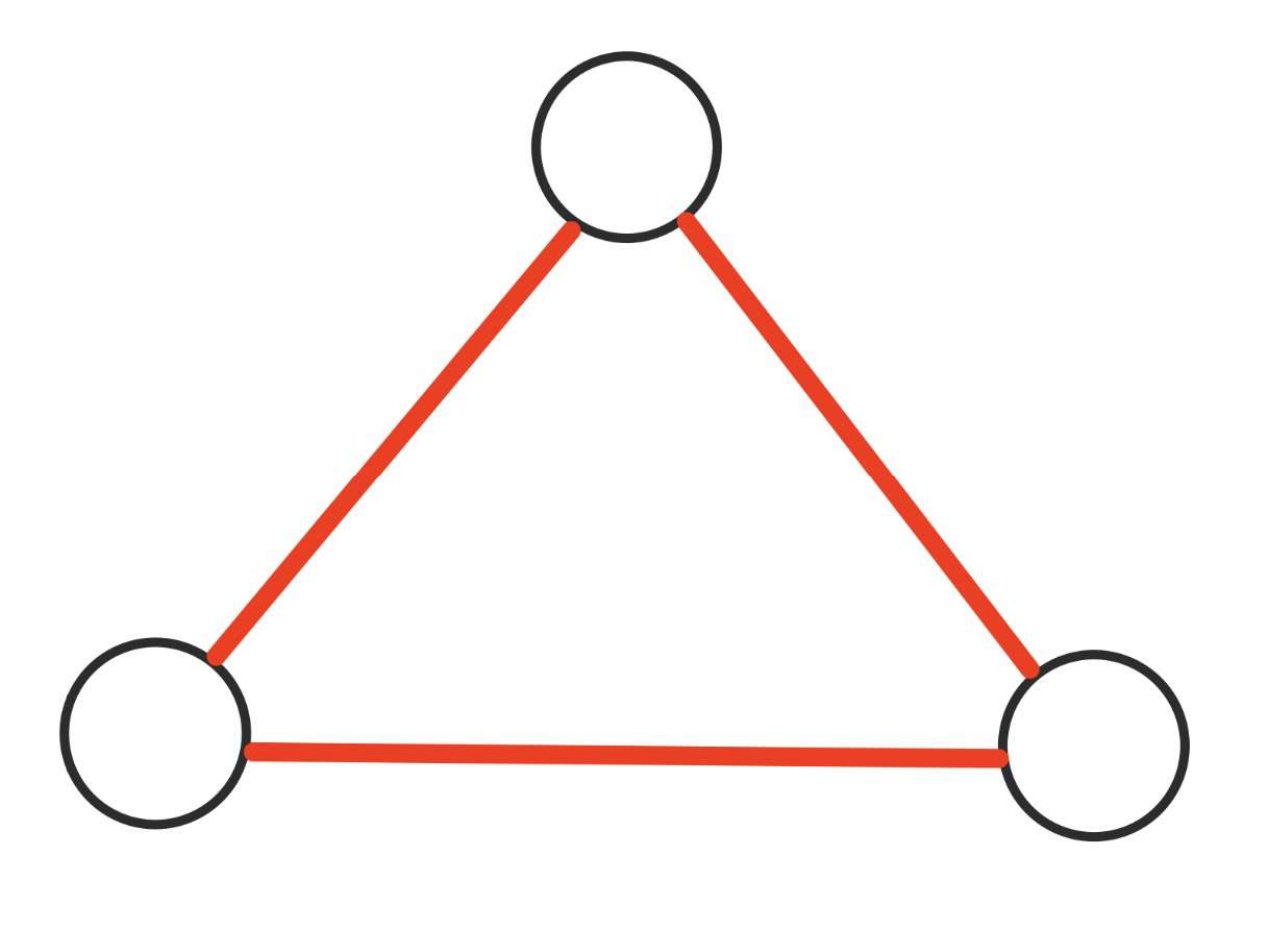}
        \par\vspace{0.5ex}
        (a)
    \end{minipage}
    \hfill
    \begin{minipage}{0.30\textwidth}
        \centering
        \includegraphics[width=0.8\textwidth, page=2]{hyp_image.pdf}
        \par\vspace{0.5ex}
        (b)
    \end{minipage}
    \hfill
    \begin{minipage}{0.30\textwidth}
        \centering
        \includegraphics[width=0.8\textwidth, page=3]{hyp_image.pdf}
        \par\vspace{0.5ex}
        (c)
    \end{minipage}

    \caption{ Panels (a), (b), and (c) correspond to Type 1, Type 2, and Type 3 triangles, respectively. Colors of the edges convey hyperedge membership.}
    \label{fig-tri-examples}
\end{figure}

\begin{example}[Colored Triangle Frequencies] \label{example-1-tri}
There are three possibilities for the number of possible edge colors for a triangle, which correspond to different isomorphism classes:
\begin{enumerate}
    \item[(a)]  Type 1: All three edges of the triangle are formed from the same hyperedge,
     \begin{align*}
    T(\Delta_1) & =  \frac{1}{m} \sum_{1 \le i \le m}\sum_{j_1,j_2,j_3} \indc{h_{i} \ni \{j_1,j_2,j_3\} } \labthis \label{Type 1-tri}.
\end{align*}
    \item[(b)] Type 2: Exactly two of the edges of the triangle are from the same hyperedge,
\begin{align*}
    T(\Delta_2) & =  \frac{1}{\binom{m}{2}} \sum_{ i_1<i_2 }\sum_{j_1,j_2,j_3} \indc{h_{i_1} \ni \{j_1,j_2\}, h_{i_2} \ni \{j_2,j_3\}, h_{i_2} \ni \{j_3,j_1\}} \labthis \label{Type 2-tri}.
\end{align*}
 \item[(c)]  Type 3: All the three edges of the triangle are formed from  different hyperedges,
\begin{align*}
    T(\Delta_3) & =  \frac{1}{\binom{m}{3}} \sum_{ i_1 < i_2 < i_3 } \sum_{j_1,j_2,j_3} \indc{h_{i_1} \ni \{j_1,j_2\}, h_{i_2} \ni \{j_2,j_3\}, h_{i_3} \ni \{j_3,j_1\}}\labthis \label{Type 3-tri}.
\end{align*}
\end{enumerate}
\end{example}

\begin{example}[Colored Two-Star Frequencies]
 For two-stars, there is also the Type 2 two-star, which contains edges formed from different hyperedges and is defined as follows:
\begin{align*}
    S(\wedge_2)  &=  \frac{1}{\binom{m}{2}} \sum_{1 \le i_1<i_2 \le m}\sum_{j_1,j_2,j_3} \{
     \indc{h_{i_1} \ni \{j_1,j_2\}} \indc{h_{i_2} \ni \{j_2,j_3\}}  + \\
    & \qquad \qquad \qquad \indc{h_{i_1} \ni \{j_1,j_2\}}  \indc{h_{i_2} \ni \{j_3,j_1\}} +
     \indc{h_{i_2} \ni \{j_2,j_3\}} \indc{h_{i_2} \ni \{j_3,j_1\}}
    \}
    \labthis \label{Type 2-twostar}.
\end{align*}
\end{example}
In Example \ref{example-1-tri}, it should be noted that Type 2 triangles only occur when at least one of the edges has multiplicity 2. Thus, no induced subgraph can be isomorphic to a Type 2 triangle, but it can contain an isomorphic copy. Moreover, Type 1 triangles are also considered a homomorphism density for two-stars. In general, we refer to a subgraph containing $k$ colors as a Type $k$ subgraph\footnote{For more complicated subgraphs, there may be several color-isomorphic equivalence classes that are Type $k$. }.

\subsubsection{Colorless Subgraph Frequencies} \label{sec-colorless-subgraph-freq}
 Although colored subgraph frequencies contain more information compared to typical notions of subgraph frequencies, there may be instances in which we are primarily interested in studying subgraphs based only on the frequency of interactions among the vertices. We now consider representing a colorless subgraph frequency in terms of colored subgraph frequencies. To this end, we now define restrictions to colorless subgraphs. Suppose that $G_\fC(\mathcal{H}_m) = (\cV(G_\fC),\cE(G_\fC ))$ is an edge-colored graph, and let $\widebar{G}_m = (\cV(\widebar{G}), \cE(\widebar{G}))$ denote its restriction to a simple colorless graph, where $ \cV(\widebar{G}) =\cV(G_\fC)$ and $\cE(\widebar{G}) = \{\{a,b\} :  a,b \in h_i \text{ for some } i \in[m]  \}$\footnote{
One may consider restrictions to multigraphs, where  $\cE(\widebar{G})$ is a multiset.  However, since notions of isomorphism are less standard for multigraphs, we focus on restrictions to simple graphs for simplicity.}.  Let $\cK(v)$ denote the complete graph on $v$ vertices and $H$ be a subgraph of $\cK(v)$ with $v$ vertices and $e$ edges.  If two colorless subgraphs $H$ and $K$ are isomorphic, we denote this as $H \iso K$.  Our notion of colorless subgraph frequency will capture the frequency of colored subgraphs whose colorless restriction is isomorphic to the colorless subgraph of interest.  For triangles, this idea is illustrated in Figure \ref{fig-tri-examples}.   
 
Formally, we define the colorless (homomorphism) subgraph frequency for $H$ as:
\begin{align*}
T(H;r) &=  \frac{1}{\binom{m}{r}} \sum_{i_1<\dots<i_r} \sum_{j_1 < \cdots < j_v} \sum_{k=1}^r \sum_{\scC \in {  \{i_1, \ldots, i_r\} \choose k} } \  
\sum_{K_\fC}
\indc{K_\fC \subseteq h_{i_1, \ldots, i_r}(j_1,\ldots, j_v)},
\labthis \label{type-H-sub-colorless-1}
\end{align*}
where $\sum_{K_\fC}$ is the summation over the set: 
\begin{align*}
    \clr{\text{simple }K_\fC:  \widebar{K} \iso H, \cV(K_\fC) = \{j_1, \ldots, j_v\}, \cC(K_\fC) = \scC}.
\end{align*}
This subgraph frequency may be viewed as an unbiased estimate of the parameter: 
\begin{align*}
\widebar{\theta}(H;r) = \Ex{\widebar{C}(h_1,\ldots, h_r;H,r)}. \labthis\label{type-H-sub-colorless-param-1}
\end{align*}
This parameter is again linked to the expected number of subgraphs present when $r$ hyperedges are sampled at random.  We allow both the subgraph frequency and the associated parameter to depend on $r$.  When computing quantities such as clustering coefficients, it is often more natural to consider the ratio of subgraph frequencies involving the same number of interactions.  Although these parameters have natural statistical interpretations, the corresponding estimators may not be viewed as the total number of copies of $H$ normalized by ${m \choose r}$ due to double counting.   In certain cases, the total subgraph count may be easier to compute, making it more desirable to use as a test statistic in two-sample testing problems.  The total number of copies is captured by:
\begin{align*}
S(H) &= \sum_{k=1}^e \sum_{j_1 < \cdots < j_v}   
\sum_{K_\fC}
\indc{K_\fC \subseteq h_{i_1, \ldots, i_k}(j_1,\ldots, j_v)}, \labthis \label{type-H-sub-colorless-2}
\end{align*}
where $\sum_{K_\fC}$ is the summation over the set 
\begin{align*}
    \clr{\text{simple }K_\fC  :   \widebar{K} \iso H, 
\cV(K_\fC) = \{j_1, \ldots, j_v\},\,
\cC(K_\fC) = \{i_1,\ldots,i_k\}  }.
\end{align*}
One may consider inference for the expectation of this quantity when normalized by its standard deviation.  It turns out that only subgraphs involving $e$ distinct colors will contribute to the fluctuations of this statistic asymptotically.
These colored subgraphs, which we call rainbow subgraphs, have special properties.  In the combinatorics literature, problems related to rainbow subgraphs have been studied by various authors, including \citet{ERDOS199381} and \cite*{alon-rainbow}. In Section \ref{sec-deg-fil}, we show that rainbow subgraphs are highly stable to deletion of low-degree nodes, implying that $S(H)$ is also stable to such deletions.

\subsubsection{Without Multiplicity Subgraph Statistics} \label{sec-without-mult-subgraph}

We now introduce the subgraph frequencies that do not take into account the multiplicity. Subgraph counts based on these binarized hypergraphs are referred to as ``without-multiplicity'' subgraph counts. For any colorless subgraph $H$ we denote the without-multiplicity subgraph frequency as: 
\begin{align}
\widetilde{T}^{(m)}(H) = \sum_{ K \iso H } \indc{ K \subseteq \bar{G}_m } . \label{gen-sub-wo}
\end{align}
Even when hypergraph structure is present, it is common practice to binarize the network. This binarization results in a substantial loss of information; both colored and colorless subgraph frequencies are often preferable.  Nevertheless, given the prevalence of such statistics in practice, we will also study the properties of these statistics.  The probabilistic structure of these statistics differs substantially from the other subgraph frequencies that we propose, necessitating the use of very different technical tools to derive limiting distributions when they exist.

\section{Inference for subgraph counts} \label{sec-inf-subgraph}

In this section, we derive the asymptotic distributions of different types of subgraphs formed by the interactions between the vertices among the hyperedges. First, we state results for with-multiplicity subgraph frequencies in Section \ref{sec-asymp-norm-subgraph}. We consider degree-filtered subgraph frequencies in Section \ref{sec-deg-fil} and asymptotic properties of without-multiplicity subgraph statistics in Section \ref{sec-without-mult}. These two sections contain the core of our technical contributions. Section \ref{sec-subsampling} presents an estimate of the asymptotic variance using subsampling. 

\subsection{Asymptotic Normality for Subgraph Frequencies} \label{sec-asymp-norm-subgraph}
After carefully defining colored subgraph frequencies, it is apparent that these statistics may be interpreted as U-statistics in hyperedges.  While the hyperedges are set-valued, the theory of U-statistics allows the inputs to the kernel to take values in an arbitrary measurable space (see, for example, \citet{korolyuk2013theory}).  This general theory almost immediately implies a central limit theorem for these subgraph frequencies; given the practical implications of this observation for subsequent statistical inference, we state a central limit theorem formally below.

Since an infinite vertex set is often a natural choice for the underlying model, it is natural to wonder when moment conditions that imply the central limit theorem hold.  We show that a mild condition on the distribution of hyperedge cardinality is sufficient.  In what follows, let $T_\fC = (T(H_\fC^{(1)}), \ldots, T(H_\fC^{(p)}))$ denote a vector of colored subgraph frequencies, where $H_\fC^{(k)}$ involves $r_k$ colors and $e_k$ edges.  Moreover, let $\theta_\fC = (\theta(H_\fC^{(1)}), \ldots \theta(H_\fC^{(p)}))$ denote the mean of $T_\fC$.  Similarly, let $T = (T(F^{(1)},r_1), \ldots, T(F^{(p)},r_p))$ and $S = (S(F^{(1)})/{m \choose e_1}, \ldots,S(F^{(p)})/{m \choose e_p})$ denote vectors related to colorless subgraph frequencies $F^{(1)}, \ldots, F^{(p)}$.   Let $\theta$ be the mean of $T$ and $\gamma$ be a vector in which the $k$-th coordinate corresponds to $\theta(F_\fR^{(k)})$, the parameter related to the rainbow subgraph with colorless restriction equal to $F^{(k)}$.  In both cases, let $E = \max_{1 \leq k \leq p} e_k$. We have the following result:

\begin{prop}[Central Limit Theorem for Subgraph Frequencies]\label{prop3.1-asymp-norm} Suppose that the following condition holds:
 \begin{align}
        \sum_{n=1}^\infty n^{2E}\pr{\md{h} = n} < \infty, \label{clt-assump-2}
    \end{align}
Then, for some finite, positive-semidefinite matrices, $\Sigma_\fC$, $\Sigma$, $\Gamma_\fR$:  
\begin{enumerate}
\item[(a)] (Colored Subgraph) 
\begin{align*}
 \sqrt{m}(T_\fC - \theta_\fC) \darw  N(0, \Sigma_\fC),
\end{align*}
\item[(b)] (Colorless Subgraph) 
\begin{align*}
 \sqrt{m}(T - \theta) \darw  N(0, \Sigma), \quad  \sqrt{m}(S - \gamma) \darw  N(0, \Gamma_\fR).
\end{align*}
\end{enumerate}
\end{prop}
For completeness, explicit forms for these covariance matrices are provided in the proof of the above theorem, which is given in the Supplement. Estimation of the covariance matrix via subsampling is discussed in Section \ref{sec-subsampling}. It is worth pointing out that the asymptotic normality results stated above do not require any sparsity conditions.  In contrast, under the sparse graphon model, standard central limit theorems require the average degree to diverge as $n \rightarrow \infty$, and stronger conditions are often needed for cyclic graphs (see \cite{bickel2011}).  It should be emphasized that edge-exchangeable graphs allow for sparse graph sequences (see \citet{janson-edge-exchangeable}).  One of the main differences is that our subgraph frequencies are normalized by quantities related to hyperedges rather than nodes.   

We would also like to point out that for the subgraph frequencies considered in the above theorem, we would not expect the kernel to be degenerate except in trivial cases.  The kernel corresponding to $T(H_\fC)$ may be expressed as a sum of products of Bernoulli random variables indicating whether appropriate nodes are contained in a hyperedge of a subgraph that is color-isomorphic to the subgraph of interest. Cases in which such a kernel is degenerate include instances where all summands are mean 0, or where the hyperedge is almost surely constant. 

Finally, we point out that exact subgraph counting is often computationally expensive.  In practice, we suggest the use of incomplete U-statistics in which only a subset of hyperedge tuples are sampled.  The work of \cite{chen-kato} implies that a central limit theorem still holds and the asymptotic variance is not affected so long as the number of hyperedge tuples selected is $\omega(m)$.  We provide additional details regarding the implementation of this approach in Section \ref{sec-sim}.           

\subsection{Degree Filtering} \label{sec-deg-fil}

In this section, we study the theoretical properties of what we refer to as degree-filtered subgraph frequencies, which involve subgraph frequencies for which low-degree nodes are removed.  A formal definition is given in equation \eqref{def-deg-fil-sub}.  Missing data is a ubiquitous problem with network data. The robustness of various network measures to missingness has garnered substantial attention in the network science literature (for example, \cite*{BORGATTI2006124}, \cite{nakajima-clustering}, \cite{KOSSINETS2006247}, \cite{SMITH2013652}, \cite{huisman2014imputation}).
However, most studies of robustness to node missingness consider the deletion of nodes uniformly at random, and the investigation is often limited to simulation studies. As discussed previously, in many settings, low-degree nodes are often more likely to be missing.  It is of substantial interest how many such low-degree nodes can be removed without affecting the estimation and inference of underlying parameters.

In the present work, we define a node $j$'s degree by its hyperdegree, defined as $D_j = \sum_{i=1}^m \indc{h_i \ni j}$.  Hyperdegrees, which measure how many hyperedges a node participates in, is a standard generalization of the notion of degree to hypergraphs; see, for example, \citet{10.5555/2500991}. We define the degree-filtered color subgraph frequency as: 
\begin{align}
T_d(H_\fC) &= 
\frac{1}{\binom{m}{r}}
\sum_{i_1<\cdots<i_r} \sum_{j_1 < \cdots < j_v}  
\sum_{K_\fC}
\indc{K_\fC \subseteq h_{i_1, \ldots, i_r}(j_1,\ldots, j_v), D_{j_l} \ge d \;\;\forall\;\; l \in [v]},
\label{def-deg-fil-sub}
\end{align}
where the $\sum_{K_\fC}$ is the summation over the set 
\begin{align*}
    \clr{K_\fC: K_\fC \isoc H_\fC, \cV(K_\fC) = \{j_1, \ldots, j_v\},\,\cC(K_\fC) = \{i_1, \ldots, i_r\}}.
\end{align*}
The behavior of this degree-filtered subgraph frequency will shed light on the behavior of colored subgraph frequencies when nodes with degree less than $d$ are randomly missing or adversarily removed; this point is discussed in Remark \ref{remark-filtering-R}. For simplicity, we state our results for degree-filtered colored subgraph densities.  Since colored subgraph densities serve as building blocks for colorless subgraph densities, our results for colored subgraph frequencies shed light on the behavior of their colored counterparts as well.  It turns out that the behavior of these degree-filtered subgraph frequencies is markedly different for finite vertex and infinite-vertex models; in the next subsection, we consider the finite vertex case.
\begin{figure}[H]
    \centering
        
    \begin{minipage}{0.45\textwidth}
        \centering
        \includegraphics[width=0.7\textwidth, page=7]{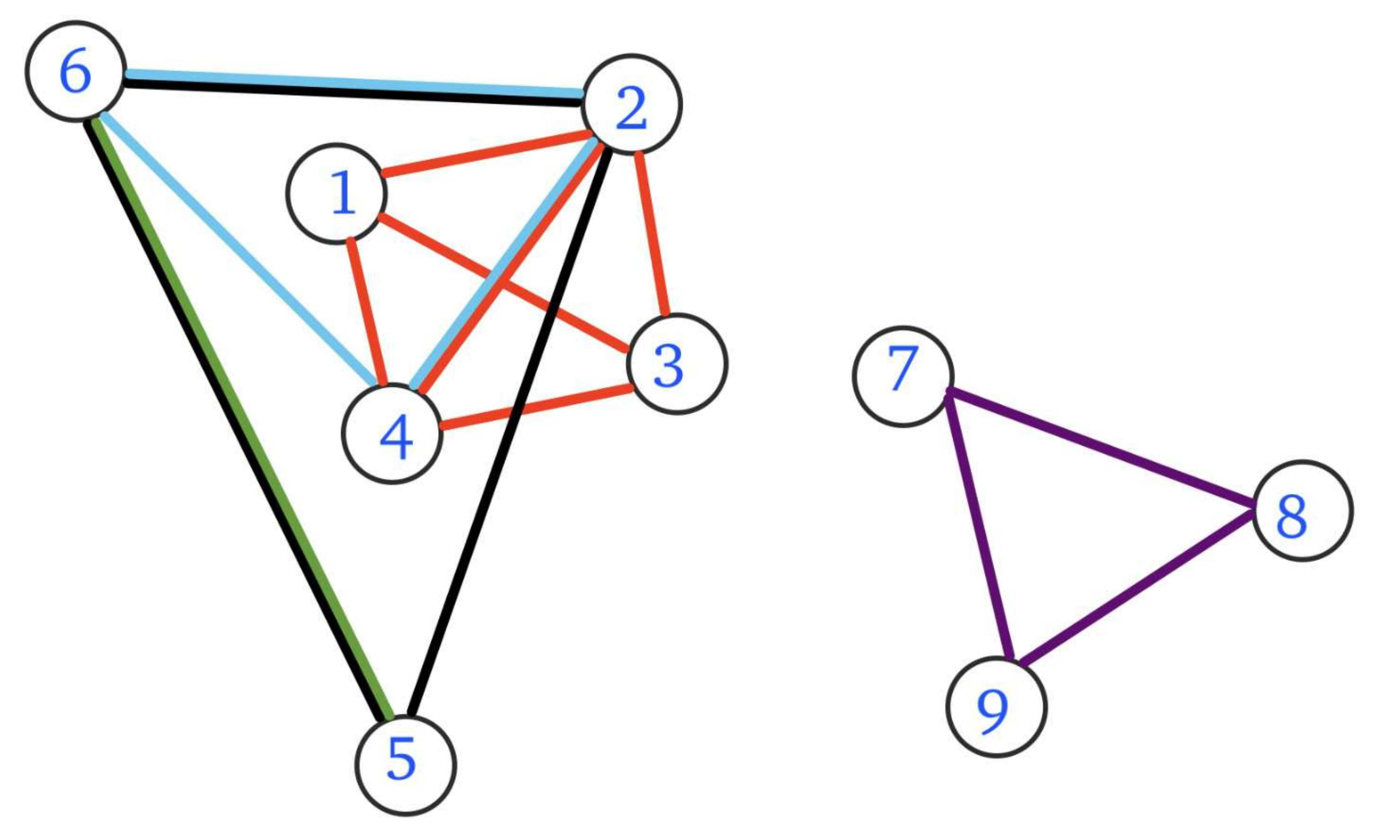}
        \par\vspace{0.5ex}
        (a)
    \end{minipage}
    \hfill
    \begin{minipage}{0.45\textwidth}
        \centering
        \includegraphics[width=0.7\textwidth, page=8]{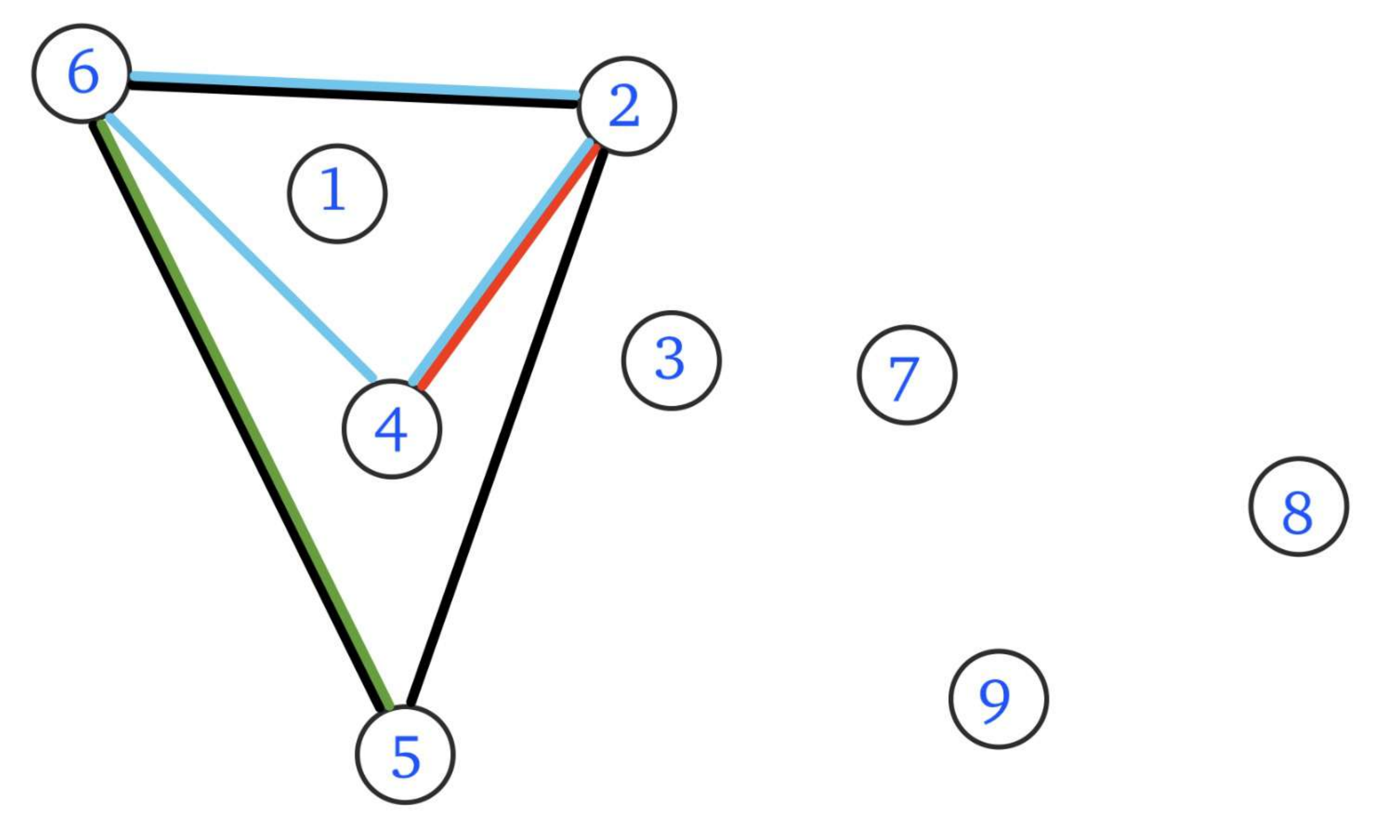}
        \par\vspace{0.5ex}
        (b)
    \end{minipage}
     \caption{Figure (a) represents the collaboration between individuals in the network. Distinct colors are used to indicate different hyperedges. Figure (b) illustrates a degree filtered hypergraph with $d=2$. }
    \label{fig-hyp-deg}
\end{figure}

\subsubsection{Finite Vertex Sets} \label{sec-deg-fil-finite}

We start by considering the case where $|\cV| <\infty$. In this case, any subgraph frequency is highly stable to the deletion of low-degree nodes.

\begin{prop}\label{prop3.2-deg-fil-finite}[Deletion Stability For Finite-Vertex Models] Suppose that $H_\fC$ is a colored subgraph and that (\ref{clt-assump-2}) holds. Moreover, suppose that $d \ll m$.  Then, 
\begin{align}
\sqrt{m}[T(H_\fC) - T_d(H_\fC)] \parw 0.
\label{eq-filtered-cp}
\end{align}
Consequently, 
\begin{align}\label{eq-filtered-limit}
\sqrt{m}[T_d(H_\fC) - \theta(H_\fC)] \darw N(0, \sigma_{H_\fC}^2),
\end{align} 
where $\sigma_{H_\fC}^2$ is the asymptotic variance of $\sqrt{m}[T(H_\fC)-\theta(H_\fC)]$.
\begin{remark}
\label{remark-filtering-R}
For any $\widehat{T}$ satisfying $ T_d(H_\fC) \leq  \widehat{T} \leq T(H_\fC)$, it is clear that the first claim of Proposition \ref{prop3.2-deg-fil-finite} implies $\sqrt{m}[T(H_\fC) - \widehat{T}] \parw 0$.  Therefore, the above result implies that any subgraph frequency $\widehat{T}$ for which nodes with degree less than $d$ are possibly removed (either randomly or adversarially), the limiting distribution would still be the same as that in the second claim of Proposition \ref{prop3.2-deg-fil-finite}.   
\end{remark}
\end{prop}
This result suggests that all subgraph frequencies are highly robust to the deletion of low-degree nodes, but we believe that these asymptotic arguments may not provide useful practical guidance. An important property of the finite vertex model driving the above result is that all nodes with nonzero probability of appearing have expected hyperdegrees between $cm$ and $Cm$ for appropriate $0<c \leq C \leq 1$.  However, most real-world networks are known to be sparse, with heavy-tailed degree distributions (see, for example, \cite{10.1137/S003614450342480}).  
In sparse graphon theory, bounded graphons are also unable to capture heavy-tailed degree distributions, while unbounded graphons can (see \citet{borgs-lp-part-one}).
For exchangeable hyperedge models, it appears that finite-vertex models suffer from similar issues, and infinite-vertex models may be more appropriate in general.  
\subsubsection{Countably Infinite Vertex Sets} \label{sec-deg-fil-countable}
The theory for degree-filtered subgraphs is much richer for models with an infinite vertex set. In this case, it turns out that the stability of colored subgraph frequencies to the deletion of low-degree nodes strongly depends on the number of colors relative to the number of edges in the subgraph of interest. Rainbow subgraphs are the most stable to such deletions, whereas subgraphs involving just one color are the least stable. We first state a general result for rainbow subgraphs.  

For rainbow subgraphs, it turns out that deletion stability will also depend on the appearance probabilities of individual nodes.  Let $p_{i} = \pr{h \ni i}$ denote the probability that node $i$ appears in a hyperedge, and $p_{(j)}$ denote the $j$-th largest element of $(p_{j})_{j \geq 1}$.  The sequence $(p_{(j)})_{j \geq 1}$ is closely related to the hyperdegree distribution associated with $P$.  For subgraphs involving fewer colors, other joint appearance probabilities will also play a role in deletion stability.

 Intuitively, if $(p_{(j)})_{ j \geq 1}$ decays quickly, then the number of ``major players'' who are likely to appear in a given subgraph is smaller, and the deletion of low-degree nodes has less of an impact.  One might also expect that the structure of the subgraph plays a role in deletion stability; this is indeed the case, but there is a subtle interplay with the rate of decay of $(p_{(j)})_{j \geq 1}$.  For concreteness, suppose that $p_{(j)} \lesssim j^{-\alpha}$ for some $\alpha >2$; in the Supplement, we state a more general result that yields sharper statements under stronger decay conditions, such as exponential decay.    

 Let $\widebar{d}_{\widebar{H}} = \min_{j \in \cV(\widebar{H})} d_{\widebar{H}}(j)$ be the minimum degree of a vertex in the subgraph $\widebar{H}$. For the colorless subgraph $\widebar{H}$, let $d_{\widebar{H}}(j) = |\{ j \in \cV(\widebar{H}) :   \{i,j\} \subseteq \cE(\widebar{H})\}|$.  Moreover, for each $k \in \{1, \ldots, v\}$, define the following quantity, which captures the largest number of edges that involve a subset of nodes of cardinality $k$:
 \begin{align}
 \label{eq-n-k}
 N_k = \max_{A \subseteq \cV(\widebar{H}) \ s.t. \ |A| = k} \left|\left\{ \{i,j\} \in \cE(\widebar{H}) \ \bigr\rvert \ i \in A \right\} \right|.
 \end{align}
 We have the following result:
\begin{theorem}[Deletion Stability for Rainbow Subgraphs Under Polynomial Decay] \label{thm3.3-df-gen-poly}
Suppose that $H_{\fR}$ is a rainbow subgraph with $v$ vertices and $e$ edges.  Moreover, suppose that $p_{(j)} \ll j^{-\alpha}$, where $\alpha > 2$, and  (\ref{clt-assump-2}) holds. Then, 
    \begin{align*}
    \sqrt{m}[T(H_\fR) - T_d(H_\fR)] \parw 0
    \end{align*}
    as $\mti$, provided $ d \ll m^{1-\beta}$, where
\begin{align*} 
    \beta &= \max \left\{\frac{2}{\widebar{d}_{\widebar H}}\plr{\frac{1}{2}+\frac{1}{\alpha}}, \frac{1}{(N_1 - 1)}\plr{\frac{1}{2}+\frac{v-1}{\alpha}}, \frac{1}{e}\plr{\frac{1}{2}+\frac{v}{\alpha}},\right.\\ & \qquad\qquad\qquad\qquad
    \left.\clr{\frac{1}{N_k}\plr{\frac{1}{2}+\frac{v-k}{\alpha}}}_{k=2}^{v-2}\right\}, \labthis \label{beta-poly-def}
\end{align*}
and $N_k$ is defined in (\ref{eq-n-k}). Consequently,
\begin{align}
\sqrt{m}[T_d(H_\fR) - \theta(H_\fR)] \darw N(0, \sigma_{H_\fR}^2),
\end{align} 
where $\sigma_{H_\fR}^2$ is the asymptotic variance of $\sqrt{m}[T(H_\fR)-\theta(H_\fR)]$. 
\end{theorem}

Some comments are in order.  First, the central limit theorem above holds for any fixed $d$.  Second, since $S(H)$ is asymptotically equivalent to a rainbow subgraph, our result above suggests that this colorless subgraph frequency is generally highly robust to the deletion of low-degree nodes.  Moreover, the above result is sharp in general, as our next result demonstrates.

\begin{prop}[Sharpness of Degree-Filtering Condition for Two-Stars]\label{prop3.4-sharpdeg} 
Suppose that $H_\fR$ is a rainbow two-star and $p_j \propto j^{-\alpha}$ for $\alpha >4$. Furthermore, if $d \ll m^{1-\beta} = m^{1/2 - 2/\alpha}$, then 
\begin{align*}
    \sqrt{m}\bbE[T(H_\fR) - T_d(H_\fR)] \rightarrow 0.
\end{align*}  Conversely, if $\alpha \in (4,8)$ and $d \gg m^{1-\beta}$, then there exists $P$ such that 
\begin{align*}
    \sqrt{m}\bbE[T(H_\fR) - T_d(H_\fR)] \rightarrow \infty.
\end{align*}
\end{prop}

General statements are possible for rainbow and unicolor subgraphs, but for color patterns that are in between, the structure of the subgraph plays a larger role, and conditions can become much more complicated. Nevertheless, to shed light on the behavior of other types colored subgraphs to deletion of low-degree nodes, we return to our running example involving triangles, and derive conditions under which degree filtering has an asymptotically negligible impact.  In what follows, for $i \in [3]$, let $T_d(\Delta_i)$ denote a degree-filtering triangle frequency. Let $\cC^{(d)}$ denote the set: 
\begin{align*}
\cC^{(d)} = \left\{\{j_1,j_2,j_3\}  :  \ p_{j_l} > \frac{1}{2m} \ \forall \ l \in [3], \ p_{j_k} < \frac{2d}{m} \text{ for some } k \in [3] \right\}, \labthis \label{def-set-C}
\end{align*}
The following theorem provides sufficient conditions for the asymptotic normality of degree-filtered triangles frequencies.

\begin{theorem}[Central Limit Theorems for Degree-Filtered Colored Triangles] \label{thm3.5-df-tri-poly}
Suppose that $p_{(j)} \ll j^{-\alpha}$, where $\alpha > 2$. Then the following statement holds: 
\begin{enumerate}
    \item [(i)] $\sqrt{m}[T_{d}(\Delta_1)-\theta(\Delta_1)] \darw N(0,\sigma_{\Delta_1}^2)$ for Type 1 triangles holds provided $\pr{h \cap \cC^{(d)} \ne \phi }\ll {1}/{m}$,
    \item [(ii)] $\sqrt{m}[T_{d}(\Delta_2)-\theta(\Delta_2)] \darw N(0,\sigma_{\Delta_2}^2)$ for Type 2 triangles holds if $ d \ll m^{\min\{1/3-1/\alpha, 1/2-2/\alpha\}}$,
    \item [(iii)] $\sqrt{m}[T_{d}(\Delta_3)-\theta(\Delta_3)]\darw N(0,\sigma_{\Delta_3}^2)$ for Type 3 triangles holds if $d \ll m^{1/2-1/\alpha}$,
\end{enumerate}
where $\theta(\Delta_k)$ and $\sigma_{\Delta_k}^2$ are the expected density and the variance of Type $k$ triangles in Theorem \ref{prop3.1-asymp-norm} for all $ k=1,2,3$.
\end{theorem}

\begin{remark}
   For Type 3 triangles,  we attain 
   $\beta = 1/2+1/\alpha$ in 
   by substituting $r=3$, $e =3$, $\widebar{d}_{\widebar{H}} = 2$, and $N_1 = 3$ in  Theorem \ref{thm3.3-df-gen-poly}. 
\end{remark}

We now give a high-level summary of the proofs of Theorems \ref{thm3.3-df-gen-poly} and \ref{thm3.5-df-tri-poly}, which require delicate arguments. For a given colored subgraph $H_\fC$, we subdivide all $v$-tuples $\{(j_1,\dots,j_v): 1\le j_1 <\dots<j_v<\infty\}$ into mutually exclusive and exhaustive sets based on the tail decay rates of the node inclusion probabilities $p_j$ with $0 < \delta_m' < \delta_m <1$ as follows:
\begin{enumerate}
    \item [(a)] Frequent vertex interactions: 
    \begin{align*}
        \cA_{v,m} = \{ (\jonev): \text{all of the $p_j \ge \delta_m$} \},
    \end{align*}
    \item [(b)] Interactions involving at least one rare vertex: 
    \begin{align*}
        \cB_{v,m} = \{ (\jonev): \text{at least one of the } p_j \le \delta_m' \},
    \end{align*}
    \item [(c)] All other interactions: 
    \begin{align*}
        \cC_{v,m} = (\cA_{v,m} \cup \cB_{v,m})^c.
    \end{align*}
\end{enumerate}
The values of $ \delta_m'$ and $\delta_m$ are carefully tuned to control the behavior of all three components and depend on the type of colored subgraph.  Notice that if $\delta_m$ is chosen large enough, then common vertices with $p_j \geq \delta_m$ will satisfy $D_j \geq d$ with appropriately high probability. Similarly, if the vertex is rare enough with $p_j \leq \delta_m'$, for appropriate choice of $\delta_m'$,  then these subgraphs will not appear with high probability, so degree filtering will have no impact.  The main technical difficulty is showing that subgraph counts involving tuples in $\cC_{v,m}$ are negligible.  These tuples include vertices that are neither common nor rare for which the expected degrees are close to the degree-filtering cutoff. The number of such vertices depends on $p_{(j)}$, and the number of possible copies formed that involve these vertices depends on the subgraph. Sharp general bounds are attainable when there is more independence in the system, leading to better stability bounds for rainbow subgraphs. 

Type 2 triangles involve less independence, but the impact of degree filtering can still be characterized in terms of the tail decay of $p_{(j)}$.  For Type 1 triangles, it is possible that hyperedges typically resemble core-periphery structures for which common vertices frequently co-occur with rarer vertices, leading to sensitivity to the deletion of low-degree nodes.  The condition for Type 1 triangles rules out this case by controlling the probability that any of these problematic vertex tuples appear in a random hyper-edge. While it is possible to construct data generating processes that satisfy this condition, we believe that this condition suggests that one should typically not expect unicolor subgraphs to be stable to deletion.

\subsection{Without Multiplicity} \label{sec-without-mult}
We now study the effects of binarization on statistical inference for subgraph frequencies. We first state a result that suggests that inferential tools based on approximating a non-degenerate limiting distribution will fail for models involving a finite vertex set.  Binarized subgraph frequencies are far different from typical functionals for which central limits arise that involve smooth functions of weakly dependent variables.  Nevertheless, we provide an example for which a central limit theorem arises under an infinite-vertex model, suggesting that inference based on the normal approximation can be possible in these settings.

\subsubsection{Finite Vertex Sets} \label{sec-without-mult-finite}

In this subsection, we will consider an i.i.d. hyperedge model for which $|\cV| < \infty$ is finite. Similar to degree-filtering, the finite-vertex case exhibits different behavior compared to the infinite-vertex case, but these differences are more pronounced under binarization.  In fact, we have the following result:

\begin{prop} \label{prop3.6-as-wm-fin}
Suppose that $|\cV| < \infty$.  Then, for any subgraph $H$ there exists some $\cS \in\bbN$ such that
\begin{align*}
\pr{ \exists \ M \ \text{such that} \ \forall \  m \geq M, \  \widetilde{T}^{(m)}(H) = \cS} =1 .
\end{align*}
\end{prop}

The above result is closely related to the infinite monkey theorem, and exploits the fact that any subgraph that has a non-zero probability of appearing will appear eventually.  While such arguments are highly asymptotic in nature, they do establish that finite vertex sets exhibit certain pathologies.  At the very least, the assumption of an infinite or growing vertex set seems to be crucial to obtain a non-degenerate limiting distribution.  In the next section, we establish a positive result for binarized quantities related to exchangeable hyperedge models.

\subsubsection{Countably Infinite Vertex Sets}
\label{sec-without-mult-countable}

We now show that a subgraph frequency related to the number of unique $k$-order interactions satisfies a central limit theorem.  To define such a statistic, it will be more convenient to consider isomorphisms rather than homomorphisms; since this differs from the exposition earlier in the paper, we provide an explicit expression for this statistic: 
\begin{align*}
\widetilde{T}_k^{(m)} = \sum_{j_1< \ldots < j_k} \indc{ \{ j_1, \ldots, j_k\} \in \cH_m }.
\end{align*}
To study the behavior of this statistic, it will be convenient to sort certain appearance probabilities in descending order.  Let $p_{j_1, \dots, j_k} = \pr{ h =  \{j_1, \ldots, j_k\}}$ and $X_j = \indc{\cH_m \ni \scA_j}$, where $\scA_j$ corresponds to the $j$-th largest value of $p_{j_1, \dots, j_k}$, and ties are broken arbitrarily. Analogously, define $Y_{ij} = \indc{h_i = \scA_j}$ and $p_j = \pr{Y_{1j} =1}$.  For any $\delta >0$, let $d_1 = c_\delta m^{1/\alpha}$ and $d_2 = C_\delta m^{1/\alpha}$, where $c_\delta$ and $C_\delta$ will be chosen to satisfy certain properties. Consider the following decomposition of the sum into three components:    
\begin{align}
    Z_{1,\delta} = \sum_{j=1}^{d_1-1}  X_j,\quad Z_{2,\delta} = \sum_{j=d_1}^{d_2-1}  X_j, \quad  Z_{3,\delta} = \sum_{j=d_2}^{\infty} X_j. \label{decomp-wo}
\end{align}
Moreover, let $\Delta_m^2 = \var(\widetilde{T}_k^{(m)})$.  We have the following result:  

\begin{theorem}[Central Limit Theorem for Unique $k$-order Interactions] \label{th3.7-wo-mul-thm}
Suppose that $p_j \propto 1/j^\alpha$ for any $\alpha>2$.  Then, for any $0 < \delta < 1$, there exists $ 0 < c_\delta <C_\delta$ such that the following hold:
\begin{enumerate}
    \item [(i)] For $m$ large enough, 
    \begin{align*}
      \pr{\frac{\md{Z_{1,\delta}-\bbE(Z_{1,\delta})}}{\Delta_m} > \delta } \leq \delta.  
    \end{align*}
    \item [(ii)] For $m$ large enough, $Z_{3,\delta}$ permits an independent approximation of the form: 
    \begin{align*}
    \pr{\frac{\md{ Z_{3,\delta} -  \sum_{i=1}^m\sum_{j = d_2}^\infty Y_{ij} }}{\Delta_m} > \delta}   \leq \delta.
    \end{align*}
    Moreover,  
    \begin{align*}
        0 < \lim_{m \rightarrow \infty}\frac{\var(Z_{3,\delta})}{\Delta_m} <1.
    \end{align*}
    \end{enumerate}
    Furthermore, as $m \rightarrow \infty$, $\Delta_m \rightarrow \infty$ and 
    \begin{align*}
        \frac{\widetilde{T}_k^{(m)} - \bbE(\widetilde{T}_k^{(m)})}{\Delta_m} \darw N(0,1).
    \end{align*} 
\end{theorem}

Our theorem above sheds light on the behavior of various components of the sum.  The $d_1$ most frequently appearing $k$-tuples appear with very high probability, and do not contribute to fluctuations of the statistic.  In contrast, all rare $k$-tuples represented by $Z_{3,\delta}$ appear with such low probability that they may be approximated by a with-multiplicity statistic.  $Z_{3,\delta}$ contributes to the fluctuations of the statistic asymptotically, but so does $Z_{2,\delta}$.  A key part of our proof is hinges on the observation that $Z_{2, \delta}$ is in fact negatively dependent; therefore, central limit theorems for such sums may be adapted (cf. \cite{newman1984}). By invoking negative dependence of both $Z_{2,\delta}$ and $Z_{3,\delta}$, we show that $\delta$ can be chosen so that the distribution of $(Z_{2,\delta} + Z_{3,\delta})/\Delta_m$ is arbitrarily close to that of the desired normal limit.   

Thus, for unique $k$-order interactions, the probabilistic structure is more delicate than the with-multiplicity variants. Our characterization of the dependence structure lay the groundwork for statistical inference.  It should be noted that our arguments strongly depend on the negative dependence property, which likely does not generalize to other types of without-multiplicity subgraph frequencies.  In these cases, it is likely that very different probabilistic tools would be needed to derive limiting distributions.  We leave an investigation of these other cases to future work.

\subsection{Variance Estimate Using Subsampling} \label{sec-subsampling}
Subsampling-based approaches to quantifying uncertainty have seen renewed interest in modern settings due to their computational tractability.  The standard approach to subsampling pioneered \citet{10.1214/aos/1176325770} also works under weaker regularity conditions than the bootstrap. While the standard quantile-based approach is very general, when the limiting distribution is normal, it does not make use of this information directly.  Empirically, we see that approaches based on a normal approximation combined with a subsampling-based estimator of the variance can yield better finite sample performance.  Subsampling-based approaches to variance estimation are discussed to some extent in \citet{10.1214/aos/1176325770}, but the properties of such estimators have not been studied previously in the case of U-statistics.  Since subsampling-based variance estimators are used in our simulations and data analysis, for completeness, we state a proposition below that establishes consistency of this variance estimator under mild conditions.
Let $\euA_{1,m,b}, \ldots,  \euA_{N,m,b}  \in { [m] \choose b}$ denote subsets of size $b$ chosen with replacement from $[m]$.  Let $S_{1}^{(b)}, \ldots, S_{N}^{(b)}$ be $p$-dimensional statistics, where for $j \in [N],$ 
\begin{align*}
    S_j^{(b)} = \plr{f_1((h_{i})_{i \in \euA_{j,m,b}}) ,\ldots,  f_p( (h_{i})_{i \in \euA_{j,m,b}})},
\end{align*} and $\widebar{S}^{(b)} = \frac{1}{N} \sum_{j=1}^N S_j^{(b)}$. Furthermore, for $k \in [p]$, let $S_{0,k}^{(b)} = f_k(h_1,\ldots,h_b)$ and denote the covariance estimator as 
\begin{align*}
    \Lambdahsub = \frac{1}{N} \sum_{j=1}^N \plr{S_j^{(b)} - \widebar{S}^{(b)}} \plr{S_j^{(b)} - \widebar{S}^{(b)}}^\transpose
\end{align*}

In the following proposition, we will show that the subsampling variance consistently estimates a covariance matrix related to $S_j$ if this statistic can be well-approximated by a U-statistic.  Both degree-filtered subgraph frequencies and incomplete versions of count frequencies fall under this category under appropriate conditions.  For each $k \in [p]$, let $g_k:\bbR^{r_k} \mapsto \bbR$ be the kernel of the U-statistic of order finite $r_k$, and for $k \in [p]$, let $U_{0,k}^{(b)} = \frac{1}{{ b \choose r_k}} \sum_{i_1 < \cdots < i_{r_k}} g_k( h_{i_1} ,\ldots, h_{i_{r_k}})  $.

\begin{prop}[Consistency of Subsampling Variance Estimate]\label{prop3.8-subsampling}
Suppose that $b = o(m)$, $N \rightarrow \infty$, and $ b\bbE[(S_{0,k}^{(b)} - U_{0,k}^{(b)})^2] = o(1)$ and $\bbE[g_k^4(h_1, \ldots, h_{r_k})] < \infty$, for all $k \in [p]$.  Then, 
\begin{align*}
\Lambdahsub \parw \Lambda,
\end{align*}
where $\Lambda$ is the asymptotic covariance matrix of $\sqrt{m}(U_{0,1}^{(m)}, \ldots, U_{0,p}^{(m)})$.  
\end{prop}

\begin{remark}
For smooth functions, our variance estimator may be combined with a plug-in estimate of the gradient to yield a consistent estimator for the asymptotic variance of the corresponding normal approximation.
\end{remark}
\begin{remark}
For degree-filtered subgraph frequencies, our proofs establish $L^1$ rather than $L^2$ convergence. Our proofs may be adapted to study $L^2$ convergence at the cost of more tedious proofs.  An appropriate uniform integrability assumption and the established $L^1$ convergence would also imply the required $L^2$ convergence. 
\end{remark}

\section{Simulations} \label{sec-sim}

\subsection{Simulation Setup} \label{subsec-sim-setup}

In this section, we study the properties of subsampling for different types of subgraph densities. In particular, we assess the coverage probabilities of Type 2 two-stars and colorless triangles of the form (\ref{type-H-sub-colorless-2}).  We consider the following model: for each $ 1\le i \le m,$ the cardinality of the hyperedge set, $|h_i|$ follows a discrete distribution on the set $\{2, 3, \dots\}$ with probabilities $\pr{|h_i| = n } \propto {6^n}/{n!}$ for all $ n \ge 2$. Here, $h_i$ is a simple random sample without replacement of size $|h_i|$, with appearance probabilities for each vertex $j$ being $\pr{h_i \ni j  } 
    \propto 1/j^2$ for $j \ge 1$.

 We generate $m=\{500,1000\}$ hyperedges with $n = 1000$ vertices using the hypergraph generating scheme described above. The empirical variance of the subgraph densities is computed using 200 independent Monte Carlo iterations and is treated as an estimate of the true variance. We then compute the estimate of the variance via subsampling for the subgraph densities. 
 The subgraph frequencies are approximated using incomplete U-statistics (see \cite{chen-kato}). For any $r$-th order U-statistic of sample size $m$, the incomplete U-statistic approximation allows us to reduce the computational complexity from $m \choose r$ to any $N = \omega(m)$, enabling scalable computation. This estimation step involves calculating the subgraph densities by sampling $N = m^{1.1}$ hyperedge tuples.

To construct the confidence interval for the model parameters, we use the subsample scheme as follows: 
we perform subsampling on the generated hypergraphs using subsamples of size $b = Cm /\log(m)$, for some positive constant $C$, specified in Table \ref{sim-cp-tab-1} and Table \ref{sim-cp-tab-2}. We choose $N=1000$ subsamples with replacement. Hypergraphs formed using the subsamples consists of $b$ hyperedges which are sampled without replacement from $m$ hyperedges. Similarly as above, we approximate the counts of Type 2 two-stars for subsampling using incomplete U-statistics with a sample size of $N = b^{1.1}$ from the subsample consisting of $b$ hypergraphs. Finally, we repeat the subsampling procedure 200 times to obtain the empirical coverage of the confidence intervals in Table \ref{sim-cp-tab-1} and Table \ref{sim-cp-tab-2}.

\subsection{Simulation Results} \label{sec-sim-res}

In this section, we analyze the simulation results related to subgraph densities based on the hypergraph models discussed in Section \ref{subsec-sim-setup}. 
Both Table \ref{sim-cp-tab-1} and Table \ref{sim-cp-tab-2} show that the proposed subsampling strategy is effective for a wide range of $C$'s. The empirical coverage probabilities of the confidence intervals are generally close to the nominal level for both sample sizes. Although there are some instances where the coverage appears to be slightly conservative, it should be noted that our choice of $N$ for the incomplete U-statistic involves substantial downsampling. It is reassuring that even in this case, coverage is reasonable for a wide choice of subsample sizes. Owing to space limitations, additional subsampling density plots for colorless triangle densities that show excellent approximation of the true sampling distribution by the subsampling distribution have been deferred to the Supplement.

\begin{table}[!htb]
\caption{Empirical Coverage probabilities of 95\% confidence intervals for Type 2 two-stars using subsampling.}\label{sim-cp-tab-1}
\centering
\centering
\begin{tabular}{ccccccccc}
\hline
    $C$         & $1$    & $1.2$   & $1.4$ & 1.6 & 1.8 & 2 & 2.2 & 2.4  \\ \hline
$m = 500$   & 0.970 & 0.950 & 0.970 & 0.985 & 0.945 & 0.975 & 0.960 & 0.955 \\
$m = 1000$ & 0.990 & 0.985 & 0.970 & 0.980 & 0.985 & 0.960 & 0.960 & 0.975\\  \hline
\end{tabular}
\end{table}

\begin{table}[!htb]
\caption{Empirical Coverage probabilities of 95\% confidence intervals for colorless triangle density using subsampling.}\label{sim-cp-tab-2}
\centering
\begin{tabular}{ccccccccc}
\hline
    $C$         & $1$    & $1.2$   & $1.4$ & 1.6 & 1.8 & 2 & 2.2 & 2.4  \\ \hline
$m = 500$   & 0.955 & 0.950  & 0.950 & 0.985 & 0.980 & 0.990 & 0.990 & 0.985 \\
$m = 1000$  & 0.945 & 0.940 & 0.945 & 0.980 & 0.975 & 0.990 & 0.970 & 0.980 \\ \hline
\end{tabular}
\end{table}

\section{Real Data} \label{sec-real-data}

\subsection{Data Description} \label{sec-data-des}

In this section, we analyze the collaboration patterns in statistics, epidemiology, and movies. To this end, we have curated four collaboration networks that feature the coauthorship networks from two statistics journals and one epidemiological journal; the fourth network is based on collaborations between actors in various movies. Existing datasets on coauthorship networks are already binarized and thus not rich enough to carry out hypergraph-based analysis. To construct coauthorship networks in statistics and epidemiology, we used bibliographical data from three journals published by Taylor \& Francis. The statistics journals considered are `Journal of the American Statistical Association - Theory and Methods' (JASA) and `Journal of Computational and Graphical Statistics' (JCGS). For epidemiology we considered `Emerging Microbes \& Infections' (EMI). The dataset was curated by downloading the bibliography files for the published issues of the respective journals for the past 10 years (2015-2024) from the Taylor \& Francis website. The published articles form the hyperedge set and the authors constitute the set of nodes in the hypergraph. To construct the movie-actor collaboration network, we collected data from the list of all American movies released over the past 10 years (2015-2024). In this hypergraph, the movies form the hyperedges and actors/actresses serve as the vertices. 

We use hypergraph-based summary statistics to distinguish between collaborative patterns.  Since our hypergraph statistics are new, there is no obvious choice for test statistics to compare the different networks a priori.  To avoid double-dipping, we split each sample of hyperedges into a training set used for choosing test statistics, and a test set to conduct inference. We split each dataset as follows: 20\% for selecting statistics and 80\% for inference purposes. Type 2 clustering coefficients and two-star frequencies were chosen since they effectively distinguished networks on the selection set data.  A Type 2 clustering coefficient is defined as the ratio of a Type 2 triangle frequency to a Type 2 two-star frequency; one can define other classes of colored clustering coefficients in an analogous manner.

To facilitate comparison, we conduct analysis on a binarized adjacency matrix under a sparse graphon model assumption. The binarization of the hypergraphs was performed by identifying the presence of at least one shared hyperedge between each pair of vertices. However, it is important to note that the resulting binarized networks are highly sparse, and under such conditions, the asymptotic normality of subgraph density statistics may fail to hold (see \cite{bickel2011} for further details).  In particular, for asymptotic normality to hold for triangle frequencies, $\rho_n = \omega(n^{-2/3})$ is required, so inference involving variants of the clustering coefficient is particularly challenging. Moreover, as discussed previously, when fundamental units are interactions rather than nodes, an exchangeable hyperedge model may be more appropriate.

\subsection{Results} \label{sec-data-res}

Using the test data, we construct confidence intervals for the ratios of the above-mentioned subgraphs (or the function of subgraphs) for pairwise comparison between these networks, assuming that they are independent. Subsampling is conducted with subsamples of sizes $1.5m/\log(m)$ for Type 2 clustering coefficient and Type 2 two-stars. The summary measures for each subsample are approximated via incomplete U-statistic with a sample of size $N=b^{1.5}$ drawn from the subsample consisting of $b$ hyperedges. Finally, the confidence intervals are constructed using 1000 with replacement samples. We adopt the following testing criteria: if the 95\% confidence interval for the ratio of subgraph density (or functions of subgraph densities) between two networks does not include the null the value 1, then the null hypothesis about equality of the two hypergraphs is rejected at 5\% level of significance.

Confidence intervals for pairwise comparisons are shown in Table \ref{ratio-ci-tab-2}. For Type 2 two-stars, the collaboration pattern between statistics and epidemiology appears significantly different, as none of the confidence intervals include the null value 1. The intervals for the pairs (JCGS, JASA) and (JASA, EMI) fall well below 1, indicating notable structural differences. Interestingly, even within the same field, JASA and JCGS exhibit different collaborative patterns. Type 2 clustering coefficients distinguish between academic and Hollywood networks. In particular, all upper bounds for comparisons involving the Movie network lie below one, suggesting that researchers with aligned interests collaborate more consistently than actors playing similar roles. For (JCGS, JASA), the lower bound exceeds one, implying greater group-level collaboration among JCGS authors. These distinctions are not detected by binarized summary measures, strengthening our argument that hypergraphs provide richer structural information.
\begin{table}[!htb]
\centering
\caption{95\% Confidence intervals of ratio subgraph counts or their functionals for different real-life networks using subsampling on the test dataset.}
\label{ratio-ci-tab-2}
\begin{tabular}{lcccc}
\hline
 &
  \begin{tabular}[c]{@{}c@{}}Type 2 \\ two-star \\  frequency \end{tabular} &
  \begin{tabular}[c]{@{}c@{}}Type 2 \\clustering \\ coefficient\end{tabular} &
  \begin{tabular}[c]{@{}c@{}}Binarized \\ two-star \\ density \end{tabular} & 
  \begin{tabular}[c]{@{}c@{}}Binarized \\ clustering \\ coefficient \end{tabular}\\ 
  \hline
JCGS VS JASA    &   (0.562, 0.601) & (1.612, 1.801)    &  (0.274, 1.033)   &   (0.109, 1.984)  \\ 
JCGS vs EMI     &   (0.799, 0.802)  & (0.772, 0.829)    &  (0.454, 0.882)   &   (0.577, 0.981)  \\
Movie vs JCGS   &   (-0.983, 1.282) & (0.121, 0.178)  &  (0.903, 2.199)  &   (0.173, 0.343)  \\  
JASA vs EMI     &   (0.454, 0.484)  & (0.450, 0.488)    &  (0.586, 1.045)   &   (0.248, 0.646)  \\  
Movie vs JASA   &   (-0.306,  0.817) & (0.236, 0.275)   &   (0.844, 1.806)  &   (0.113, 0.563)  \\ 
Movie vs EMI    &   (0.112, 0.127)  & (0.113, 0.126)    &  (1.227, 1.453)   &   (0.201, 0.269)  \\ 
\hline
\end{tabular}
\end{table}

\section{Discussion} \label{sec-dis}
In this article, we address the problem of statistical inference for subgraphs within a general hypergraph framework. 
Theoretical guarantees for colored and degree -filtered subgraphs are notably sharp and derived under mild assumptions, thus broadening its practical applicability. 
For inferential purposes, we also analyze a general subsampled estimator of variance. Our empirical study using real-world hypergraph data uncovers diverse collaboration patterns across domains.

The analysis presented here represents only the ``tip of the iceberg". It would be valuable to extend this framework to dynamic subgraph settings (\cite*{mancastroppa2024structural}, \cite{comrie2021hypergraph}), where edge exchangeability may not be an appropriate assumption. 
Another promising avenue is the development of inferential tools for approximate subgraph counting algorithms for colored subgraph frequencies.  While approximate counting algorithms may not always allow one to conduct inference for the same target parameter, these algorithms may still be of substantial interest in two-sample testing problems.

 \section{Supplement and Data Availability}
The Supplementary material contains proofs of the theorems, propositions along with their other supporting technical lemmas and additional simulation results. Codes and the data used in the paper are available on this website \href{https://github.com/ayoushmanb/hypergraph_subgraph}{https://github.com/ayoushmanb/hypergraph\_subgraph}.

\mybib

\supptitle






The Supplement is organized as follows: 
Section \ref{sec: res colored} contains the proofs for colored and colorless subgraph frequencies in Section \ref{sec-asymp-norm-subgraph}. Section \ref{sec: resdeg-fill} contains the proofs of the results in Section \ref{sec-deg-fil-finite} and \ref{sec-deg-fil-countable} for both finite and countably infinite vertex set models. Section \ref{sec: without mult} contains the proofs of subgraph frequencies under without multiplicity in Section  \ref{sec-without-mult-finite} and \ref{sec-without-mult-countable}. Section \ref{sec: res subsamp} contains the proofs regarding subsampling in Section \ref{sec-subsampling}. Finally, additional results related to degree-filtered triangles under exponential decay, additional simulation and real data analysis have been presented in Section \ref{sec: addres deg-fil}, Section \ref{sec: add simu}, and Section \ref{sec: add realdat} respectively.

\section{Proofs for Section \ref{sec-asymp-norm-subgraph}}\label{sec: res colored}

\begin{proof}[Proof of Proposition \ref{prop3.1-asymp-norm}]

    To prove (a), we use the first representation of colored homomorphism subgraph frequencies from equation \eqref{type-H-sub-hom} and note the $k$-thelement of $T_\fC$ is an U-statistic of order $r_k$ based on i.i.d. hyperedges $(h_1, \dots, h_m)$. We write the Haj\'ek projection term for $T_\fC$ as $\vp_{\fC}(h_1) = \plr{\vp_{\fC}^{(1)}(h_1), \dots, \vp_{\fC}^{(p)}(h_1)}$ where 
    \begin{align*}
        \vp_{\fC}^{(k)}(h_1) = r_k\plr{\Ex{C^{(k)}(h_1, \dots, h_{r_k})\mid h_1} - \theta_\fC^{(k)}},
    \end{align*}
    for all $ k =1, \dots, p$. Moreover, we denote the $\Sigma_\fC =  \Var{\vp_{\fC}(h_1)}$ as the covariance matrix of the first-order projection term. 

    Due to Cramer-Wald device, it is enough to show that for any vector $(l_1, \dots, l_p)\in \rl^p$, we have 
    \begin{align}
       \sqrt{m} \sumkp l_k\plr{T_\fC^{(k)} - \theta_\fC^{(k)}} \darw N\plr{0, \sum_{1 \le k , k' \le p} l_k l_{k'} \Sigma_\fC^{(k,k')}}. \label{gen-sub-mult-norm}
    \end{align}
    We next show that the variance of the linear combination of the kernels is finite as follows:
    \begin{align*}
        &\Var{ \sumkp l_k C^{(k)}\plr{h_1, \dots, h_{r_k}; H_\fC^{(k)}}}\\
        & \lesssim \sumkp l_k^2 \;\;\Ex{ C^{(k)}\plr{h_1, \dots, h_{r_k}; H_\fC^{(k)}}}^2 \\
        & = \sumkp l_k^2  \sum_{n_1, \dots, n_{r_k} \ge 1} \Ex{\plr{ C^{(k)}\plr{h_1, \dots, h_r; H_\fC^{(k)}}}^2 \mid |h_1| = n_{1}, \cdots |h_{r_k}| = n_{r_k}}\\
        & \le \sumkp l_k^2 \;\; \prod_{s=1}^{r_k} \plr{\sum_{n_s=1}^{\infty} n_s^{2E}\;\;\pr{|h| = n_s} } < \infty
    \end{align*}
    where the first inequality follows from $C_r$-inequality and the last inequality follows from the assumption \eqref{clt-assump-2}. A direct application of CLT provides us
    \begin{align}
        \frac{1}{\sqrt{m}} \sum_{i=1}^m \sumkp l_k \; \vp_\fC^{(k)}(h_i) \darw N\plr{0, \sum_{1 \le k , k' \le p} l_k l_{k'} \Sigma_\fC^{(k,k')}}.\label{lin-gen-sub}
    \end{align}
    Now we define the remainder term as,
    \begin{align}
        R_\fC = \sumkp l_k \plr{T_\fC^{(k)} - \theta_\fC^{(k)}} - \frac{1}{m} \sum_{i=1}^m \sumkp l_k \;\vp_\fC^{(k)}(h_i). \label{rem-gen-gen-decomp}
    \end{align}
    We aim to show that $\sqrt{m}R_\fC \parw 0$. To this end, we apply $C_r$-inequality to obtain the following:
    \begin{align*}
        & \Var{R_\fC} \\
        & \lesssim \sumkp l_k^2\;\; \Var{T_\fC^{(k)} - \theta_\fC^{(k)} - \frac{1}{m} \sum_{i=1}^m \vp_\fC^{(k)}(h_i)}\\
        & = \sumkp l_k^2\;\; \clr{\Var{T_\fC^{(k)} - \theta_\fC^{(k)}} - \Var{\frac{1}{m} \sum_{i=1}^m \vp_\fC^{(k)}(h_i)}}\\
        & \le \sumkp l_k^2\;\; \frac{1}{\binom{m}{r_k}^2} \sum_{ |\{i_1, \dots,i_{r_k}\} \cap \{l_1,\dots,l_{r_k}\}| \ge 2} \Cov{ C\plr{h_{i_1}, \dots, h_{i_{r_k}}; H_\fC^{(k)}}, C\plr{h_{l_1}, \dots, h_{l_{r_k}}; H_\fC^{(k)}} }\\
        & \le \sumkp l_k^2\;\; \frac{1}{\binom{m}{r_k}^2} \sum_{ |\{i_1, \dots,i_{r_k}\} \cap \{l_1,\dots,l_{r_k}\}| \ge 2} \Ex{ C\plr{h_{i_1}, \dots, h_{i_{r_k}}; H_\fC^{(k)}}, C\plr{h_{l_1}, \dots, h_{l_{r_k}}; H_\fC^{(k)}} }\\
         &\le \sumkp l_k^2\;\; \frac{1}{\binom{m}{r_k}^2} \sum_{ |\{i_1, \dots,i_{r_k}\} \cap \{l_1,\dots,l_{r_k}\}| \ge 2}\;\; \sum_{\substack{n_{i_1}, \dots, n_{i_{r_k}} \ge 1 \\ n_{l_1}, \dots, n_{l_{r_k}} \ge 1}} \plr{n_{i_{1}} \dots n_{i_{r_k}}}^R \plr{n_{l_{1}}\dots n_{l_{r_k}} }^R \\
         & \hspace{4cm}\pr{|h_{i_1}| = n_{i_1}, \cdots, |h_{i_{r_k}}| = n_{i_{r_k}}, |h_{l_1}| = n_{l_1}, \cdots, |h_{l_{r_k}}| = n_{l_{r_k}}}\\
          & = O\left(\dfrac{1}{m^2}\right),
    \end{align*}
    where the last inequality follows from the assumption \eqref{clt-assump-2}. Therefore, applying Chebyshev's inequality we get $\sqrt{m}R_\fC \parw 0$. Finally, we combine this with \eqref{lin-gen-sub} and apply Slustky's lemma to obtain \eqref{gen-sub-mult-norm}. This concludes the proof of (a).

    To prove the first part of (b),  we use the representation of the colorless triangles given in \eqref{type-H-sub-colorless-1}. Thus, the $k$-th element of $T$ can also be written as an U-statistic of order $r_k$ with i.i.d. random variables $(h_1, \dots, h_m)$. Moreover, we denote the Haj\'ek projection term of $T$ as $\vp(h_1) = \plr{\vp^{(1)}(h_1), \dots,  \vp^{(p)}(h_1)}$ where 
    \begin{align*}
        \vp^{(k)}(h_1) = r_k\plr{\Ex{\widebar C^{(k)} (h_1, \dots, h_r) \mid h_1 } - \theta^{(k)} },
    \end{align*}
    for all $ k =1, \dots, p$ and denote the $\Sigma = \Var{\vp(h_1)}$ as the variance of the projection term. Therefore, a similar calculation as shown in part (a) will yield the asymptotic normality of of $\sqrt{m} (T-\theta)$. 

    For the second part of (b), we note that the highest order term in $S(F^{(k)})$ is contributed by the rainbow subgraph $F_\fR^{(k)}$. Therefore, we have
    \begin{align*}
        \frac{S(F^{(k)})}{\binom{m}{e_k}} = T( F_\fR^{(k)}) + O_P\plr{\frac{1}{m}},
    \end{align*}
    for each $ 1 \le k \le p$. Hence, the remainder of the proof follows from part (a) of the proposition.
\end{proof}

\section{Proofs of Section \ref{sec-deg-fil}}\label{sec: resdeg-fill}

\begin{proof}[Proof of Proposition \ref{prop3.2-deg-fil-finite}]


First, we show that $\sqrt{m}[T(H_\fC) - T_d(H_\fC)] \stackrel{P}{\rightarrow} 0$.
Since $|\cV| = n < \infty$, there exists a $\delta > 0$ such that $0 < \delta < p_{(n)}$.
Recall that, for all $1 \le j \le n,$ we have 

\begin{align*}
D_j = \sum_{i=1}^m \indc{h_i \ni j} \sim \Bin{m,p_j}.  
\end{align*}
Then, for any $\epsilon > 0$ such that $d \le (1-\epsilon)mp_j,$ Chernoff's bound yields
\begin{align}\label{probbound}
\pr{D_j \le d} \le \exp\clr{- \dfrac{\epsilon^2}{2}mp_j} \le \exp\clr{- \dfrac{\epsilon^2}{2}mp_{(1)}}.   
\end{align}
for any $ j =1, \dots,n$. Next, we obtain the $L^1$ convergence for $ \sqrt{m}[T(H_\fC) - T_d(H_\fC)].$
Note that, 
\begin{align*}\label{L1finver}
    0  \le \Ex {\sqrt{m}[T(H_\fC) - T_d(H_\fC)]} & \le \sqrt{m}\,n^r\, \pr{D_1 \le d, \dots, D_r \le d}\\
    & \lesssim \sqrt{m}\,n^r\,\exp\clr{- \dfrac{\epsilon^2}{2}m  p_{(n)}} = o(1),
\end{align*}
where the third inequality follows from \eqref{probbound} and the fourth inequality is due to the fact that $p_{(1)} > \delta > 0$. Thus, the first conclusion follows. 

The second claim directly follows from Proposition \ref{prop3.1-asymp-norm} and Slutsky's lemma.
\end{proof}

Before the proving the Theorem \ref{thm3.3-df-gen-poly}, we first state the following lemma to establish the asymptotic normality of rainbow subgraph frequency $\THR$ for any general tail decay rate for $\plr{p_{(j)}}_{j \in \cV}$.
Furthermore, we denote $p_{j_1, \dots, j_k} = \pr{h \ni \{j_1, \dots j_k\}}$ as the the $k$-th order inclusion probability of the vertices $\{j_1, \dots, j_k\}$ for any $k \ge 2$. 
The proof of the lemma depends on the analyzing the behavior of rainbow subgraph frequency $\THR -\TdHR$ on the sets $\cA_{v,m}$, $\cB_{v,m}$ and $\cC_{v,m}$ defined in Section \ref{sec-deg-fil-countable}.  We further decompose $\cC_{v,m} = \bigcup_{k=0}^{v-1}\cC_{v,m}^{(k)}$ where $\cC_{v,m}^{(k)} = \{(\jonev) \in \cC_{v,m} : \text{exactly $k$ of $p_j \ge \delta_m$}\}$. We define 
\begin{align*}
    \cD_{v,m} = \clr{j\in \{\jonev\}: (\jonev) \in \cC_{v,m} \text{ and } \delta_m' < p_j < \delta_m }
\end{align*}
as the set of the `not-so-frequent' vertices.
Recall that for any vertex $j \in \cV$, $D_j = \sum_{i=1}^m \indc{h_i \ni j}$ denotes the degree of the vertex.  We also define 
 \begin{align*}
 \widetilde{D}_{j, -\clr{i_1,\dots,i_k}} = \sum_{\substack{i=1 \\ i\neq i_1,\dots,i_k}}^m  \indc{h_i \ni j},    
 \end{align*}
 as the degree of the vertex $j$ after deleting the hyperedges $(h_{i_1},\dots,h_{i_k})$. The following theorem establishes the asymptotic normality of the degree-filtered subgraph density $\TdHR$.

\begin{lemma} [Deletion Stability for Rainbow Subgraphs] \label{clt-df-gen}


    Assume that $m\delta_m' \ll d \ll m\delta_m$ and $m\delta_m \to \infty$ and the condition \eqref{clt-assump-2} hold for subgraph $H_\fR$. Furthermore, if we assume that the following asymptotic bounds hold: 
    \begin{equation} \label{ass-df-gen}
    \begin{split}
        & \sqrt{m} |\cD_{v,m}| \delta_m^{\widebar{d}_{\widebar{H}}/2} = o(1),
        \quad \sqrt{m} |\cD_{v,m}|^{v-k} \delta_m^{N_k} = o(1) \quad\text{for all $2 \le k \le v-2$},\\
        & \sqrt{m}|\cA_{v,m}|^{1/v}|\cD_{v,m}|^{v-1} \delta_m^{N_1} = o(1),\,\, {\rm{and}}
        \quad \sqrt{m}|\cD_{v,m}|^v \delta_m^e = o(1), 
        \end{split}
    \end{equation}
    then we have
    \begin{align}
        \sqrt{m} [\THR - \TdHR] \parw 0. \label{prob-conv-gen-sub}
    \end{align}
    Consequently, we also have 
    \begin{align}
        \sqrt{m} [\TdHR - \theta(H_\fR)] \darw N(0,\sigma^2_{H_\fR}), \label{asymp-norm-gen-sub}
    \end{align}
    where $\theta(H_\fR)$ and $\sigma_{H_\fR}^2$ are defined in Proposition \ref{prop3.1-asymp-norm}.
\end{lemma}

\begin{proof}[Proof of Proposition \ref{clt-df-gen}]


    We decompose the degree filtered subgraph density as
\begin{align*}
    \TdHR = T^{(1)}_d(H_\fR) + T^{(2)}_d(H_\fR) + T^{(3)}_d(H_\fR),
\end{align*}
according to the sets $\cA_{v,m}$, $\cB_{v,m}$ and $\cC_{v,m}$ respectively.
We similarly decompose $\THR = T^{(1)}(H_\fR)+ T^{(2)}(H_\fR)+T^{(3)}(H_\fR)$.
We aim to show that $\sqrt{m}( \THR - \TdHR )\parw 0$.

For any $j \in \{\jonev\}$ such that $(\jonev) \in \cA_{v,m}$ and fix $\epsilon>0$ such that $d \le (1-\epsilon)mp_{j_{1}}$, we apply Chernoff's bound to get
\begin{align*}
    P(D_{{j}} \le d) &\le \exp\clr{-\dfrac{\epsilon^2}{2}mp_{j}} \le \exp\clr{-\dfrac{\epsilon^2}{2}m\delta_m}.
\end{align*}
Also recall that $|\cA_{v,m}| \lesssim {\delta_m^{-r}}$ which implies
\begin{align*}
    0 \le  \Ex{\sqrt{m}\plr{ T^{(1)}(H_\fR) - T^{(1)}_d(H_\fR)}} 
    & \lesssim \sqrt{m} |\cA_{v,m}| \exp\clr{-\dfrac{\epsilon^2}{2}m\delta_m}\\
    & \lesssim \dfrac{\sqrt{m}}{\delta_m^r}\exp\clr{-\dfrac{\epsilon^2}{2}m\delta_m}\\
    & = o(1). \labthis \label{exp-3.6.1}
\end{align*}
This implies $\sqrt{m}\plr{ T^{(1)}(H_\fR) - T^{(1)}_d(H_\fR)}\larw{1} 0$, moreover we get $\sqrt{m}\plr{ T^{(1)}(H_\fR) - T^{(1)}_d(H_\fR)}\parw 0$.

We now proceed to prove that $\sqrt{m}\plr{ T^{(2)}(H_\fR) - T^{(2)}_d(H_\fR)} \parw 0$. Recall that
\begin{align*}
    T_d^{(2)}(H_\fR) &= 
\frac{1}{\binom{m}{e}} \sum_{i_1<\dots<i_e} \sum_{(\jonev) \in \cB_{v,m}}  \sum_{K_\fR}  \indc{K_\fR \subseteq h_{i_1, \ldots, i_e}(j_1,\ldots, j_v)}\indc{D_{j_1} \ge d, \dots, D_{j_v} \ge d},
\end{align*}
where we denote $\sum_{K_\fR}$ be summation over the set 
\begin{align*}
    \clr{ K_\fR : K_\fR \isoc H_\fR, \mathcal{V}(K_\fR) = \{j_1, \ldots, j_v\}, \mathcal{C}(K_\fR) = \{i_1, \ldots, i_r\} }.
\end{align*} 
Note that for any $1 \le i_1<\dots<i_e \le m$ and $j_1 < \dots< j_v$, the following inequality always holds true:
\begin{align*}
\indc{D_{j_1} \ge d, \dots, D_{j_r} \ge d} \le \indc{\widetilde{D}_{j_1,-\clr{i_1,\dots,i_e}} \ge d-e,\dots,\widetilde{D}_{j_3,-\clr{i_1,\dots,i_e}} \ge d-e}.
\end{align*}
Again, we choose $\epsilon \in (0,1)$ such that $ m-d < (1-\epsilon) (m-e)(1-p_j)$ and apply Chernoff's bound to obtain the following,

\begin{align*}
    & \max_{(\jonev) \in \cB_{v,m}} \clr{\pr{\widetilde{D}_{j_1,-\clr{i_1,\dots,i_e}} \ge d-e,\dots,\widetilde{D}_{j_r,-\clr{i_1,\dots,i_e}} \ge d-e}}\\
    & \le \max_{(\jonev) \in \cB_{v,m}} \min_{j_1,\dots,j_r} \clr{\pr{\widetilde{D}_{j, -\clr{i_1,\dots,i_3}} \ge d-e}}\\
    & \le \max_{(\jonev) \in \cB_{v,m}} \min_{j_1,\dots,j_r} \exp\clr{-\frac{\epsilon^2}{2}(m-e)(1-p_j)}
    \le \exp\clr{-\frac{\epsilon^2}{2}(m-e)(1-\delta_m)}, \labthis \label{set-B-prob-bound}
\end{align*}
for any $1 \le i_1<\dots<i_e\le m$. Now, we use Markov's inequality to obtain the following bound:
\begin{align*}
    &\pr{\sqrt{m}|T_d^{(2)}(H_\fR)| \ge \epsilon} \\
    &\le \frac{\sqrt{m}}{\epsilon\binom{m}{e}} \sum_{i_1<\dots<i_e} \sum_{\cB_{v,m}}  \sum_{K_\fR}  \bbE\left(\indc{K_\fR \subseteq h_{i_1, \ldots, i_e}(j_1,\ldots, j_v)}\right.\\
    &\qquad\qquad\qquad\qquad \left.\indc{\widetilde{D}_{j_1, -\clr{i_1,\dots,i_e}} \ge d-e, \dots, \widetilde{D}_{j_2, -\clr{i_1,\dots,i_e}} \ge d-e} \right)\\
    &= \frac{\sqrt{m}}{\epsilon\binom{m}{e}} \sum_{i_1<\dots<i_e} \sum_{\cB_{v,m}}  \sum_{K_\fR}  \pr{K_\fR \subseteq h_{i_1, \ldots, i_e}(j_1,\ldots, j_v)}\\
    &\qquad\qquad\qquad\qquad\pr{\widetilde{D}_{j_1, -\clr{i_1,\dots,i_e}} \ge d-e, \dots, \widetilde{D}_{j_2, -\clr{i_1,\dots,i_e}} \ge d-e}\\
    &  = \dfrac{\sqrt{m}}{\epsilon\binom{m}{e}} \sum_{\cB_{v,m}}\sum_{i_1 < \dots< i_r} \sum_{K_\fR} \plr{\prod_{(a,b)\in \cE(K_\fR)} p_{j_aj_b}}\,\,\pr{\widetilde{D}_{j_1, -\clr{i_1,\dots,i_e}} \ge d-e, \dots, \widetilde{D}_{j_r, -\clr{i_1,\dots,i_e}} \ge d-e}\\
    & \le \dfrac{\sqrt{m}}{\epsilon} \max_{\cB_{v,m}} \clr{\pr{\widetilde{D}_{j_1, -\clr{i_1,\dots,i_e}} \ge d-e, \dots, \widetilde{D}_{j_r, -\clr{i_1,\dots,i_e}} \ge d-e}} \left(\sum_{\cB_{v,m}} \sum_{K_\fR} \prod_{(a,b)\in \cE(K_\fR)} p_{j_aj_b}\right)^3\\
   & = o(1), \labthis \label{exp-3.6.2}
\end{align*}
where the last equality follows from exponential decay in \eqref{set-B-prob-bound}. Therefore applying Markov's inequality we have $\sqrt{m} T_d^{(2)}(H_\fR) \parw 0$. A similar calculation will also yield $\sqrt{m} T^{(2)}(H_\fR) \parw 0$.

We now show that $ \sqrt{m} T^{(3)}(H_\fR) \parw 0$ separately for each $\clr{\cC_{v,m}^{(k)}}_{k=0}^{v-1}$. To this end, we define the local subgraph counts as
\begin{align*}
    T_\fR(j_1,\dots,j_{v-k};\cC_{v,m}^{(k)}) = \sum_{j_{v-k+1}<\dots<j_v} \sum_{K_\fR}\indc{ K_\fR \subseteq h_{\ioner}(\jonev)},
\end{align*}
for each $j_1,\dots,j_{v-k} \in \cD_{v,m}^{(k)}$, $1 \le i_1 < \dots <i_r \le m$ and $ 0 \le k \le v-1$. We compute the following:

\begin{align*}
   0 \le &  \Ex{\frac{\sqrt{m}}{\binom{m}{e}} \sum_{i_1<\dots<i_r} \sum_{\cC_{v,m}^{(v-1)}} \sum_{K_\fR} \indc{ K_\fR \subseteq h_{\ioner}(\jonev)} \indc{D_{j_1} \ge d, \dots, D_{j_v} \ge d}} \\
   & \le \Ex{\frac{\sqrt{m}}{\binom{m}{e}} \sum_{i_1<\dots<i_r} \sum_{\cC_{v,m}^{(v-1)}} \sum_{K_\fR} \indc{ K_\fR \subseteq h_{\ioner}(\jonev)}}\\
   & = \sqrt{m} \Ex{\sum_{j_1 \in \cD_{v,m}} T_\fR(j_1;\cC_{v,m}^{(v-1)}) \prod_{l=1}^e \indc{h_{i_l} \ni j_1} }\\
   & \le \sqrt{m} \sum_{j_1 \in \cD_{v,m}} \sqrt{\Ex{ T_\fR^2(j_1;\cC_{v,m}^{(v-1)})} p_{j_1}^{d_{\widebar H}(j_1)}}\\
   & = O\plr{\sqrt{m} \sum_{j_1 \in \cD_{v,m}} p_{j_1}^{{\widebar d}_{\widebar H}/2}}\\
    & = O\plr{\sqrt{m}|\cD_{v,m}| \delta_m^{\widebar{d}_{\widebar{H}}/2} }\\
    & = o(1), \labthis \label{exp-3.6.3}
\end{align*}
where the last equality follows from the assumption \eqref{ass-df-gen}.

For $2 \le k \le v-2$ we have,
\begin{align*}
0 \le &  \Ex{\frac{\sqrt{m}}{\binom{m}{e}} \sum_{i_1<\dots<i_r} \sum_{\cC_{v,m}^{(k)}} \sum_{K_\fR} \indc{ K_\fR \subseteq h_{\ioner}(\jonev)} \indc{D_{j_1} \ge d, \dots, D_{j_v} \ge d}} \\ 
& \le \Ex{\frac{\sqrt{m}}{\binom{m}{e}} \sum_{i_1<\dots<i_r} \sum_{\cC_{v,m}^{(k)}} \sum_{K_\fR} \indc{ K_\fR \subseteq h_{\ioner}(\jonev)}}\\
& = \sqrt{m} \sum_{j_1, \dots, j_{v-k} \in \cD_{v,m}} \Ex{ T_\fR(j_{1}, \dots, j_{v-k};\cC_{v,m}^{(k)}) \prod_{l=1}^e \indc{h_{i_l} \ni \{j_1,\dots, j_{v-k}\}\cap\cD_{v,m}^{(k)}} }\\
    & \le \sqrt{m} \sum_{j_1, \dots, j_{v-k} \in \cD_{v,m}} \sqrt{\Ex{ T_\fR^2(j_{1}, \dots, j_{v-k};\cC_{v,m}^{(k)})} \prod_{l=1}^{v-k} p_{j_l}^{d_{\widebar H}(j_l)}}\\
    & = O\plr{\sqrt{m} |\cD_{v,m}|^{v-k} \delta_m^{N_k}}\\
    & = o(1), \labthis \label{exp-3.6.4}
\end{align*}
where the last equality follows from the assumption \eqref{ass-df-gen}.

Recall that $\cC_{v,m}^{(1)} =\{ \{\jonev\} \in \cC_{v,m}: \text{exactly one of the $j_l$'s with}\,\, p_{j_l} \ge \delta_m\}$. Thus, we have the following:
\begin{align*}
    0 \le &  \Ex{\frac{\sqrt{m}}{\binom{m}{e}} \sum_{i_1<\dots<i_r} \sum_{\cC_{v,m}^{(1)}} \sum_{K_\fR} \indc{ K_\fR \subseteq h_{\ioner}(\jonev)} \indc{D_{j_1} \ge d, \dots, D_{j_v} \ge d}} \\ 
& \le \Ex{\frac{\sqrt{m}}{\binom{m}{e}} \sum_{i_1<\dots<i_r} \sum_{\cC_{v,m}^{(1)}} \sum_{K_\fR} \indc{ K_\fR \subseteq h_{\ioner}(\jonev)}}\\
& = \sqrt{m}\sum_{\cC_{v,m}^{(1)}} \sum_{K_\fR} \prod_{(a,b) \in \cE(K_\fR)} p_{j_aj_b}\\
& = O\plr{|\cA_{v,m}|^{1/v} |\cD_{v,m}|^{v-1} \delta_m^{N_1}}\\
& = o(1) \labthis \label{exp-3.6.5}
\end{align*}
where the second-to-last equality follows since $p_{j_aj_b} \le \delta_m$ for any of the $j_a,j_b \in \cD_{v,m}$ and the last inequality follows from the assumption \eqref{ass-df-gen}. 

Lastly, for $C_{v,m}^{(0)} =\{\{\jonev\} \in \cC_{v,m}: \text{all of the $j$'s with}\,\, p_j \le \delta_m\},$ we have
\begin{align*}
     0 &\le   \Ex{\frac{\sqrt{m}}{\binom{m}{e}} \sum_{i_1<\dots<i_r} \sum_{\cC_{v,m}^{(0)}} \sum_{K_\fR} \indc{ K_\fR \subseteq h_{\ioner}(\jonev)} \indc{D_{j_1} \ge d, \dots, D_{j_v} \ge d}} \\ 
     & \le \Ex{\frac{\sqrt{m}}{\binom{m}{e}} \sum_{i_1<\dots<i_r} \sum_{\cC_{v,m}^{(0)}} \sum_{K_\fR} \indc{ K_\fR \subseteq h_{\ioner}(\jonev)}}\\
 & = \sqrt{m}\sum_{ \cC_{v,m}^{(0)}} \sum_{K_\fR} \prod_{(a,b)\in \cE(K_\fR)} p_{j_aj_b}\\
    & = O\plr{\sqrt{m} |\cD_{v,m}|^v \delta_m^e}\\
    & = o(1) \labthis \label{exp-3.3.6}
\end{align*}
where the second-last inequality holds similar to \eqref{exp-3.3.5}.
Thus, combining equations \eqref{exp-3.3.3}, \eqref{exp-3.3.4} and \eqref{exp-3.3.5} we get $\sqrt{m} T_d^{(3)}(H_\fR) \larw{1} 0$ and $\sqrt{m} T^{(3)}(H_\fR) \larw{1} 0$. Thus, we have $ \sqrt{m} [ T(H_\fR) - T_d(H_\fR)] \parw 0$. 
Finally, we combine Proposition \ref{prop3.1-asymp-norm} and Slutsky's theorem to complete the proof.
\end{proof}

We now apply Lemma \ref{clt-df-gen} to prove the Theorem \ref{thm3.3-df-gen-poly} for the polynomial decay on $ \plr{p_{(j)}}_{j \ge 1}$ as follows. 

\begin{proof}[Proof of Theorem \ref{thm3.3-df-gen-poly}]


We set $\delta_m = O(m^{-(1-\nu)})$ and $\delta_m' = \Theta(m^{-1})$, where $ \nu < \beta$. Solving the equations ${x_1}^{-\alpha} = m^{-(1-\nu)}$ and ${x_2}^{-\alpha} = m^{-1}$ for $x_1$ and $x_2$ respectively, we get $|\cD_{v,m}| = x_2-x_1 = m^{1/\alpha}-m^{(1-\nu)/\alpha}=O(m^{1/\alpha})$. Also, note that $|\cA_{v,m}|^{1/v} \lesssim \delta_m \lesssim m^{1-\nu}$. For the remainder of the proof, we compute the asymptotic bounds in \eqref{ass-df-gen} for the choices of $\delta_m$ and $\delta_m'$:
\begin{align*}
    & \sqrt{m}|\cD_{v,m}| \delta_m^{\widebar{d}_{\widebar{H}}/2} = O\plr{m^{\frac{\widebar{d}_{\widebar{H}}}{2}\clr{ \nu - 1 + \frac{2}{\widebar{d}_{\widebar{H}}}\plr{\frac{1}{2} + \frac{1}{\alpha}}} }} = o(1),\\
    & \sqrt{m} |\cD_{v,m}|^{v-k} \delta_m^{N_k} = O\plr{m^{N_k \clr{\nu - 1 + \frac{1}{N_k}\plr{\frac{v-k}{\alpha} + \frac{1}{2}} }}} = o(1) \quad \text{for $ 2 \le k \le v-2$},\\
    & \sqrt{m} |\cA_{v,m}|^{1/v} |\cD_{v,m}|^{v-1}\delta_m^{N_1} = O \plr{m^{(N_1-1) \clr{\nu - 1 +\frac{1}{N_1-1}\plr{\frac{v-1}{\alpha}+ \frac{1}{2}} }}} = o(1),\\
    & \sqrt{m} |\cD_{v,m}|^v \delta_m^e = O\plr{m^{e \clr{ \nu - 1 + \frac{1}{e} \plr{\frac{v}{\alpha} + \frac{1}{2}}}}} =o(1),
\end{align*}
where the equality follows from the fact that $  d \ll m^{1-\beta}$ and hence the result follows.
\end{proof}

\begin{proof}[Proof of Proposition \ref{prop3.4-sharpdeg}]
To prove the first claim, that is $\sqrt{m} \Ex{T(H_\fR) - T_d(H_\fR)} \to 0 $ holds we need to find the range of $\alpha$ that satisfies the condition for $d$ and $\beta$ in \eqref{beta-poly-def}. 
Note that, for rainbow two-stars, we have $ v = 3$, $e = 2,$ $r = 2,$ $\widebar{d}_{\widebar H} = 2 $ and  $N_1 = 2$.

Therefore, $\beta$ in \eqref{beta-poly-def} simplifies to 
\begin{align}\label{sharpbeta}
    \beta &= \max \left\{\frac{1}{2}+\frac{1}{\alpha}, \frac{1}{4} + \frac{3}{2\alpha}\right\} = \frac{1}{2}+\frac{1}{\alpha}.
\end{align}
Using \eqref{sharpbeta} it is easy to verify that, 
\begin{align*}
    d \ll m^{1-\beta} = m^{1/2 - 2/\alpha}
\end{align*}
The lower bound for $(1/2 - 2/\alpha) > 0$ implies $\alpha > 4$.
This proves our first claim. 

Next, to prove the second claim, we consider the rainbow two-star model $H_\fR$, where $\pr{|h|=2} =1$ and

\begin{align*}
    \pr{h = \{1,j\}} = \begin{cases}  p_{1j}, \quad j \ge 2\\
                         0,    \quad \text{otherwise}.
    \end{cases} \labthis \label{two-star-model}
\end{align*}
Note that, under the model \eqref{two-star-model}, the appearance probabilities of the vertices
\begin{align*}
   p_j = \pr{h \cap \{j\} \ne \phi} = \begin{cases}  1 \quad j = 1\\
                        p_{1j}  \quad j \ge 2.
   \end{cases}
\end{align*}
Now, we aim to show that under the model \eqref{two-star-model} for rainbow two-stars $T_\fR$,
\begin{align*}
    \sqrt{m} \Ex{T(H_\fR) - T_d(H_\fR)} \to \infty
\end{align*}
Consider the set of tuples 
\begin{align*}
    \cC^{(d)}_1 = \clr{ (1, j_1, j_2) : \frac{1}{m} < p_{j_1}, p_{j_2} < \frac{1}{m^{1/2+2/\alpha}} }
\end{align*}

Recall that, 
\begin{align*}
    & \sqrt{m} \Ex{T(H_\fR) - T_d(H_\fR)} \\ 
    & = \sqrt{m}\Ex{{\frac{1}{\binom{m}{2}}} \sum_{i_1<i_2} \sum_{j_1<j_2} \indc{h_{i_1} \ni \{1, j_1\}} \indc{h_{i_2} \ni \{1, j_2\}} \indc{ D_1 < d, D_{j_1} < d, D_{j_2} < d}}.
\end{align*}
Note that for any $(1, j_1, j_2) \in \cC^{(d)}_1,$ there exists $\epsilon > 0$ such that
\begin{align*}
    \pr{ (D_1 < d) \cup (D_{j_1} < d) \cup (D_{j_2} < d)} \ge \pr{\widetilde{D}_{j_1,-\clr{1,2}} \ge d-2} > \epsilon
\end{align*}
Therefore,  
\begin{align*}
    \sqrt{m} \Ex{T(H_\fR) - T_d(H_\fR)} & \ge \sqrt{m} \sum_{i_1,i_2} p_{1j_1}p_{1j_2} \pr{\widetilde{D}_{j_1,-1,2} \ge d-2}\\
    & \ge \epsilon \sqrt{m} \sum_{i_1,i_2} p_{j_1}p_{j_2}\\
    & = \omega \plr{\dfrac{\epsilon \sqrt{m} \;|\cC^{(d)}_1|}{m^2}}\\
    & =\omega \plr{ m^{\frac{8 - \alpha}{2 \alpha}} }\to \infty
\end{align*}
   The last inequality holds since $|\cC^{(d)}_1| = \omega(m^{(1 + 4/\alpha)})$ and $\alpha \in (4,8)$. This concludes our claim.
\end{proof}

Before proving the Theorem \ref{thm3.5-df-tri-poly} for colored triangles for the polynomial decay rate, we state a lemma for colored triangles under any general tail decay rate for $\clr{p_{(j)}}_{j \ge 1}$.

\begin{lemma}[Central Limit Theorems for Degree-Filtered Colored Triangles] \label{clt-df-thm-tri}


Assume that $m\delta'_m \ll d \ll m \delta_m$ with $m\delta_m \to \infty$ and the condition \eqref{clt-assump-2} holds true for all types of colored triangles.
    \begin{enumerate}
    \item[(i)] $\sqrt{m}[T_{d}(\Delta_1)-\theta(\Delta_1)] \darw N(0,\sigma_{\Delta_1}^2)$ for Type 1 triangles holds provided $\pr{h \cap \cC_{3,m} \ne \phi }= o(1/m)$,
        
        \item[(ii)] $\sqrt{m}[T_{d}(\Delta_2)-\theta(\Delta_2)] \darw N(0,\sigma_{\Delta_2}^2)$ for Type 2 triangles holds if $\sum_{j=1}^\infty\sqrt{p_j} = O(1)$,
        \begin{equation}\label{ass-df-2-2}
            \begin{split}
                \sqrt{m} |\cD_{3,m}|^3 \delta_m^2 = o(1), \qquad &\qquad \sqrt{m} |\cA_{3,m}|^{1/3} |\cD_{3,m}|^2 \delta_m^2= o(1),\\
                \quad \sqrt{m} |\cD_{3,m}| \delta_m= o(1)\quad &{\rm{and}} \quad \sqrt{m} |\cD_{3,m}| \delta_m^{3/4}= o(1), 
            \end{split}
        \end{equation}
        
        \item[(iii)] $\sqrt{m}[T_{d}(\Delta_3)-\theta(\Delta_3)]\darw N(0,\sigma_{\Delta_3}^2)$ for Type 3 triangles holds if
       \begin{equation}\label{ass-df-3-2}
           \begin{split}
               &\sqrt{m} |\cD_{3,m}| \delta_m = o(1), \quad \sqrt{m} |\cA_{3,m}|^{1/3} |\cD_{3,m}|^2 \delta_m^3 = o(1) \\
               &\qquad\qquad\qquad {\rm{and}} \quad \sqrt{m} |\cD_{3,m}|^3 \delta_m^3 = o(1), 
           \end{split}
       \end{equation}
    \end{enumerate}
    where $\theta(\Delta_k)$ and $\sigma_{\Delta_k}^2$ are the expected density and the variance of Type $k$ triangles in Theorem \ref{prop3.1-asymp-norm} for all $ k=1,2,3$.
\end{lemma}

\begin{proof}[Proof of Lemma \ref{clt-df-thm-tri}]


We claim that the under the conditions of Lemma \ref{clt-df-thm-tri}, $\sqrt{m} [ T_d(\Delta_k) - \theta(\Delta_k)] \parw 0$ for all $ k=1,2,3$. Then the limiting distribution of $T_d(\Delta_k)$ should follow from the Proposition \ref{prop3.1-asymp-norm} and Slutsky's theorem.

To prove (i), we first decompose $T_d(\Delta_1)$ into three mutually exclusive sets as follows.
\begin{align*}
    T_d(\Delta_1) & = \frac{1}{m} \sum_{i=1}^m \left\{ \sum_{\cA_{3,m}} \indc{h_i \ni \{j_1,j_2,j_3\}} \indc{D_{j_1} \ge d, D_{j_2} \ge d, D_{j_1} \ge d} + \right. \\ 
     & \qquad\qquad \sum_{\cB_{3,m}} \indc{h_i \ni \{j_1,j_2,j_3\}} \indc{D_{j_1} \ge d, D_{j_2} \ge d, D_{j_1} \ge d} +\\
     & \left.\qquad\qquad \sum_{\cC_{3,m}} \indc{h_i \ni \{j_1,j_2,j_3\}} \indc{D_{j_1} \ge d, D_{j_2} \ge d, D_{j_1} \ge d} \right\}\\
    & = \frac{1}{m} \sum_{i=1}^m \clr{ C_{i;d}^{(1)}(\Delta_1) + C_{i;d}^{(2)} (\Delta_1)+ C_{i;d}^{(3)}(\Delta_1) }\\
    & =  \Tdoo +\Tdotw + \Tdoth.
\end{align*}
We similarly decompose $T(\Delta_1) = T^{(1)}(\Delta_1)+ T^{(2)} (\Delta_1)+ T^{(3)}(\Delta_1)$.
We now compute the following:

\begin{align*}
    0 & \le \Ex{\sqrt{m}\plr{T^{(1)} - T_d^{(1)}(\Delta_1)}} \\
    & = \dfrac{\sqrt{m}}{m} \sum_{i=1}^m \sum_{\cA_{3,m}} \Ex{ \indc{h_i \ni \{j_1,j_2,j_3\}}\indc{(D_{j_1} < d) \cup (D_{j_2} < d) \cup (D_{j_3} < d)}}\\
    & \le \dfrac{\sqrt{m}}{m} \sum_{i=1}^m \sum_{\cA_{3,m}} \pr{{(D_{j_1} < d)\cup (D_{j_2} < d) \cup (D_{j_3} < d)}}
\end{align*}

For any $j \in \{j_1, j_2, j_3\}$ such that $(j_1,j_2,j_3) \in \cA_{3,m}$, choose $\epsilon>0$ such that $d \le (1- \epsilon)mp_j$. Then using Chernoff's bound we have $\pr{D_j \le d} \le \exp\clr{- \frac{\epsilon^2}{2}mp_j}$. Therefore,
\begin{align*}
    0\le \Ex{\sqrt{m}\plr{T^{(1)}(\Delta_1) - T_d^{(1)}(\Delta_1)}}  &\lesssim |\cA_{3,m}| \sqrt{m}  \exp\clr{- \frac{\epsilon^2}{2} m\delta_m} \\
    &\lesssim \frac{\sqrt{m}}{\delta_m^3} \exp\clr{- \frac{\epsilon^2}{2} m \delta_m} = o(1), \labthis\label{exp-3.1.1}
\end{align*}
where the second last step follows from the fact that the number maximum elements in the set $\cA_{3,m}$ is $|\cA_{3,m}| \lesssim \delta_m^{-3}$. Then, equation \eqref{exp-3.1.1} provides $\sqrt{m}(T^{(1)}(\Delta_1) - \Tdoo) \larw{1} 0$. 

To prove $\sqrt{m}\Tdotw \parw 0$, we note that
\begin{align*}
    \Tdotw \le \dfrac{1}{m} \sum_{i=1}^m \sum_{\cB_{3,m}}\indc{h_i \ni \clr{j_1,j_2,j_3}}  \indc{\widetilde{D}_{j_1,-\clr{i}} \ge d-1, \widetilde{D}_{j_2,-\clr{i}} \ge d-1, \widetilde{D}_{j_3,-\clr{i}} \ge d-1}.
\end{align*}
Since $m\delta_m' \ll d$, there exists $\epsilon_1 \in (0,1)$ such that $m-d \le (m-1)(1-p_j)$ for any $j \in \{j_1,j_2,j_3\}$ such that $(j_1,j_2,j_3) \in \cB_{3,m}$. Then we apply again Chernoff's bound to obtain:
\begin{align*}
    & \max_{\cB_{3,m}} \pr{\widetilde{D}_{j_1,-\clr{i}} \ge d-1,\widetilde{D}_{j_2,-\clr{i}} \ge d-1,\widetilde{D}_{j_3,-\clr{i}} \ge d-1} \\
    & \le \max_{\cB_{3,m}} \min_{j = j_1, j_2, j_3} \pr{\widetilde D_{j,-\clr{i}} \ge  d-1}\\
    & = \max_{\cB_{3,m}} \min_{j = j_1, j_2, j_3} \pr{ m-1-\widetilde D_{j,-\clr{i}} \le m-d}\\
    & \le \max_{\cB_{3,m}} \min_{j = j_1, j_2, j_3} \pr{m-1-\widetilde D_{j,-\clr{i}} \le (1-\epsilon_1)(m-1)(1-p_j) )}\\
    & \le \exp\clr{-\frac{\epsilon_1^2}{2}(m-1)(1-\delta_m)}.
\end{align*}
Therefore, for any $\epsilon>0$ we also have
\begin{align*}
   &\pr{\sqrt{m}|T_d^{(2)}| \ge \epsilon}\\
   &\le \dfrac{\sqrt{m}}{\epsilon m} \sum_{i=1}^m \Ex{\sum_{\cB_{3,m}} \indc{h_i \ni \{j_1, j_2, j_3\} }\indc{\widetilde{D}_{j_1,-\clr{i}} \ge d-1,\widetilde{D}_{j_2,-\clr{i}} \ge d-1, \widetilde{D}_{j_3,-\clr{i}} \ge d-1}}\\
   & \le \dfrac{\sqrt{m}}{\epsilon} \Ex{|h|^3 \indc{h \cap \cB_{3,m} \ne \phi} \indc{\widetilde{D}_{j_1,-\clr{i}} \ge d-1,\widetilde{D}_{j_2,-\clr{i}} \ge d-1,\widetilde{D}_{j_3,-\clr{i}} \ge d-1} }\\
   &   \le \dfrac{\sqrt{m}}{\epsilon} \sqrt{\Ex{|h|^6} \pr{h \cap \cB_{3,m} \ne \phi}}\\
   &\qquad\qquad \max_{\cB_{3,m}} \pr{\widetilde{D}_{j_1,-\clr{i}} \ge d-1,\widetilde{D}_{j_2,-\clr{i}} \ge d-1,\widetilde{D}_{j_3,-\clr{i}} \ge d-1} \\
   & \le \dfrac{\sqrt{m}}{\epsilon} \exp\clr{-\frac{\epsilon_1^2}{2}(m-1)(1-p_j)}\\
   & =o(1) \labthis \label{exp-3.1.2} 
\end{align*}
where in the last step, we used the fact that $\Ex{|h|^6} = O(1)$ (cf. \eqref{clt-assump-2}). Thus, we have $\sqrt{m}\Tdotw \parw 0$ and similar computation also leads to $\sqrt{m}T^{(2)}(\Delta_1) \parw 0$. 

 Finally, we use Chebyshev's inequality to compute
 \begin{align*}
    0 \le \Ex{\sqrt{m} T_{1}^{(3)}} & \le \sqrt{m} \Ex{|h|^3 \indc{h \cap \cC_{3,m} \ne \phi}}\\
    & \le \sqrt{m} \sqrt{\Ex{|h|^6} \pr{h \cap \cC_{3,m} \ne \phi}} \qquad \\
   & = O\left(\sqrt{m \pr{h \cap \cC_{3,m} \ne \phi}}\right)\\
   & = o(1). \labthis \label{exp-3.1.3}
\end{align*} 
Hence, we get $\sqrt{m}T^{(3)}(\Delta_1) \larw{1} 0$. Since $0\le \Tdoth \le T^{(1)}(\Delta_1)$, we also have $\sqrt{m}\Tdoth\larw{1} 0$. Combining \eqref{exp-3.1.1}, \eqref{exp-3.1.2} and \eqref{exp-3.1.3} $\sqrt{m}\slr{T(\Delta_1) - T_d(\Delta_1)} \parw 0$.

To prove (ii), we begin with a decomposition of $T_d(\Delta_2)$ according to the sets $\cA_{3,m}$, $\cB_{3,m}$ and $\cC_{3,m}$ as follows:

\begin{align*}
    T_{d}(\Delta_2) & = \frac{1}{\binom{m}{2}} \sum_{1 \le i_1 < i_2 \le m } \clr{ C_{d}^{(1)}(i_1,i_2;\Delta_2) + C_{d}^{(2)}(i_1,i_2;\Delta_2) + C_{d}^{(3)}(i_1,i_2;\Delta_2) }\\
    &= \Tdtwo +\Tdtwtw + \Tdtwth,
\end{align*}
where 
\begin{align*}
    C_{d}^{(1)}(i_1,i_2;\Delta_2) = & \sum_{\cA_{3,m}} \indc{h_{i_1} \ni \{j_1,j_2\}, h_{i_2} \ni \{j_2,j_3\},h_{i_2} \ni \{j_3,j_1\}}\indc{D_{j_1} \ge d, D_{j_2} \ge d, D_{j_2} \ge d},\\
    C_{d}^{(2)}(i_1,i_2;\Delta_2) = & \sum_{\cB_{3,m}} \indc{h_{i_1} \ni \{j_1,j_2\},h_{i_2} \ni \{j_2,j_3\},h_{i_2} \ni \{j_3,j_1\}}\indc{D_{j_1} \ge d, D_{j_2} \ge d, D_{j_2} \ge d},\\
    C_{d}^{(3)}(i_1,i_2;\Delta_2) = & \sum_{\cC_{3,m}} \indc{h_{i_1} \ni \{j_1,j_2\},h_{i_2} \ni \{j_2,j_3\},h_{i_2} \ni \{j_3,j_1\}}\indc{D_{j_1} \ge d, D_{j_2} \ge d, D_{j_2} \ge d}.
\end{align*}
We similarly decompose $T(\Delta_2) = T^{(1)}(\Delta_2) + T^{(2)}(\Delta_2)+ T^{(3)}(\Delta_2).$
Note that $D_j \sim\text{Bin}(m,p_j)$ for all $j \in \cV$. For any $\epsilon>0$ such that $ d \le (1-\epsilon)mp_j$ we apply Chernoff's bound to obtain
\begin{align*}
    \pr{D_j \le d} \le  \exp\clr{-\frac{\epsilon^2}{2}mp_j}.
\end{align*}
Thus, we now have
\begin{align*}
   0 &\le \Ex{\sqrt{m}\plr{T^{(1)}(\Delta_2) - T_{d}^{(1)}(\Delta_2) }}\\
   & = \dfrac{\sqrt{m}}{\binom{m}{2}} \sum_{i_1<i_2}\sum_{\cA_{3,m}} \bbE( \indc{h_{i_1} \ni \{j_1,j_2\},h_{i_2} \ni \{j_2,j_3\},h_{i_2} \ni \{j_3,j_1\}}\\
   & \qquad\qquad\qquad\qquad\qquad\qquad\qquad\indc{D_{j_1} < d \cup D_{j_2} < d \cup D_{j_3} < d})\\
   & \le \sqrt{m} \sum_{\cA_{3,m}} \pr{(D_{j_1} < d) \cup (D_{j_2} < d) \cup (D_{j_3} < d)}\\
   & \lesssim \sqrt{m} |\cA_{3,m}| \exp\clr{- \frac{\epsilon^2}{2}m \delta_m}\\
   & = o(1), \labthis \label{exp-3.2.1}
\end{align*}
where the second last step follows from the fact that the maximum cardinality of the set $\cA_{3,m}$ is $|\cA_{3,m}| \lesssim \delta_m^{-3}$. This imples $\sqrt{m}\slr{T^{(1)}(\Delta_2)- \Tdtwo} \larw{1} 0$.

Next, we note that
\begin{align*}
    T_{d}^{(2)}(\Delta_2) &\le \dfrac{1}{\binom{m}{2}} \sum_{i_1<i_2} \sum_{\cB_{3,m}} \Bigg\{ \indc{h_{i_1} \ni \{j_1,j_2\},h_{i_2} \ni \{j_2,j_3\},h_{i_2} \ni \{j_3,j_1\}}\\
& \qquad\qquad\qquad    \indc{\widetilde{D}_{j_1,-\clr{i_1,i_2}} \ge d-2, \widetilde{D}_{j_2,-\clr{i_1,i_2}} \ge d-2, \widetilde{D}_{j_3,-\clr{i_1,i_2}} \ge d-2} \Bigg\}.
\end{align*}
For any $j \in \clr{j_1, j_2,j_3}$ such that $(j_1,j_2,j_3) \in \cB_{3,m}$, fix $\epsilon_2 \in (0,1)$ such that $m-d < (1-\epsilon_2)(m-2)(1-p_j)$ and apply Chernoff's bound to get,
\begin{align*}
    & \max_{\cB_{3,m}}\clr{\pr{\widetilde{D}_{j_1,-\clr{i_1,i_2}} \ge d-2, \widetilde{D}_{j_2,-\clr{i_1,i_2}} \ge d-2, \widetilde{D}_{j_3,-\clr{i_1,i_2}} \ge d-2}}\\
    & \le \max_{\cB_{3,m}} \min_{j=j_1,j_2,j_3} \pr{\widetilde{D}_{j,-\clr{i_1,i_2}} \ge d-2}\\
    & =  \max_{\cB_{3,m}} \min_{j=j_1,j_2,j_3} \pr{m-2-\widetilde{D}_{j,-\clr{i_1,i_2}} \le m-d}\\
    & \le \exp\clr{-\frac{\epsilon_2^2}{2} (m-2)(1-\delta_m')}.
\end{align*}
Now, for any fixed $\epsilon>0$ and we have

\begin{align*}
    & \pr{\sqrt{m}\md{T_{d}^{(2)} (\Delta_2)} \ge \epsilon} \\
    & \le \dfrac{\sqrt{m}}{\epsilon\binom{m}{2}}\sum_{\cB_{3,m}}  \bbE\Bigg( \indc{h_{i_1} \ni \{j_1,j_2\}} \indc{h_{i_2} \ni \{j_2,j_3\}} \indc{h_{i_2} \ni \{j_3,j_1\}} \\
    & \qquad\qquad\qquad\qquad\indc{\widetilde{D}_{j_1,-\clr{i_1,i_2}} \ge d-2, \widetilde{D}_{j_2,-\clr{i_1,i_2}} \ge d-2, \widetilde{D}_{j_3,-\clr{i_1,i_2}} \ge d-2}\Bigg)\\
    & \lesssim \frac{\sqrt{m}}{\epsilon} \sum_{\cB_{3,m}} p_{j_{1}j_{2}j_{3}}p_{j_{1}j_{2}}\pr{\widetilde{D}_{j_1,-\clr{i_1,i_2}} \ge d-2, \widetilde{D}_{j_2,-\clr{i_1,i_2}} \ge d-2, \widetilde{D}_{j_3,-\clr{i_1,i_2}} \ge d-2}\\
    & \le \frac{\sqrt{m}}{\epsilon} \max_{\cB_{3,m}}\clr{\pr{\widetilde{D}_{j_1,-\clr{i_1,i_2}} \ge d-2, \widetilde{D}_{j_2,-\clr{i_1,i_2}} \ge d-2, \widetilde{D}_{j_3,-\clr{i_1,i_2}} \ge d-2}}\\
    & \hspace{5cm}\sum_{\cB_{3,m}} \left(p_{j_{1}j_{2}j_{3}}p_{j_{1}j_{2}} + p_{j_{1}j_{2}j_{3}}p_{j_{2}j_{3}} + p_{j_{1}j_{2}j_{3}}p_{j_{1}j_{3}}\right)\\
    & \lesssim  \frac{\sqrt{m}}{\epsilon} \max_{\cB_{3,m}}\clr{\pr{\widetilde{D}_{j_1,-\clr{i_1,i_2}} \ge d-2, \widetilde{D}_{j_2,-\clr{i_1,i_2}} \ge d-2, \widetilde{D}_{j_3,-\clr{i_1,i_2}} \ge d-2}} \left(\sum_{j \in \cV} p_j \right)\left(\sum_{j \in \cV} \sqrt{p_j} \right)^2\\
    & \le \frac{\sqrt{m}}{\epsilon} \exp\clr{-\frac{\epsilon_2^2}{2} (m-2)(1-\delta_m')} \left(\sum_{j \in \cV} p_j \right)\left(\sum_{j \in \cV} \sqrt{p_j} \right)^2\\
    & = o(1). \labthis \label{exp-3.2.2}
\end{align*}
Hence, we have $\sqrt{m}\Tdtwtw \parw 0$. Similar calculation also yields $\sqrt{m}T^{(2)}(\Delta_2) \parw 0$.

We now proceed to prove that $\sqrt{m} T^{(3)}(\Delta_2) \parw 0$. Note that
\begin{align*}
    0& \le \Ex{\dfrac{\sqrt{m}}{\binom{m}{2}}\sum_{i_1<i_2}\sum_{\cC_{3,m}^{(0)}} \indc{h_{i_1} \ni \{j_1,j_2\}} \indc{h_{i_2} \ni \{j_2,j_3\}} \indc{h_{i_2} \ni \{j_3,j_1\}} }\\
    & = \sqrt{m} \sum_{\cC_{3,m}^{(0)}}\Ex{\indc{h_{i_1} \ni \{j_1,j_2\}} \indc{h_{i_2} \ni \{j_2,j_3\}} \indc{h_{i_2} \ni \{j_3,j_1\}}}\\
    & = \sqrt{m} \sum_{ \cC_{3,m}^{(0)}} \left\{p_{j_{1}j_{2}j_{3}}p_{j_{1}j_{2}}+ p_{j_{1}j_{2}j_{3}}p_{j_{1}j_{3}} + p_{j_{1}j_{2}j_{3}}p_{j_{2}j_{3}}\right\}\\
    & \lesssim \sqrt{m} \sum_{\cC_{3,m}^{(0)}} \left\{p_{j_{1}j_{2}j_{3}}p_{j_{1}j_{2}}\right\}\\
    & \lesssim \sqrt{m} |\cD_{3,m}|^3 \delta_m^2\\
    & = o(1), \labthis \label{exp-3.2.3}
\end{align*}
where the second last inequality follows from the fact that $p_{j_1j_2j_3} \le p_{j_1j_2}\le p_{j_1} \le \delta_m$ for ant $(j_1,j_2,j_3)\in \cC_{3,m}^{(0)}$ and the last inequality follows from the assumption \eqref{ass-df-2-2}.
We also compute
\begin{align*}
     0&\le \frac{\sqrt{m}}{\binom{m}{2}}\sum_{i_1 <i_2}\sum_{\cC_{3,m}^{(1)}} \Ex{\indc{h_{i_1} \ni \{j_1,j_2\}} \indc{h_{i_2} \ni \{j_2,j_3\}} \indc{h_{i_2} \ni \{j_3,j_1\}}}\\
    & = \sqrt{m}\sum_{\cC_{3,m}^{(1)}} \left\{p_{j_{1}j_{2}j_{3}}p_{j_{1}j_{2}} + p_{j_{1}j_{2}j_{3}}p_{j_{1}j_{3}} + p_{j_{1}j_{2}j_{3}}p_{j_{2}j_{3}} \right\}\\
    & \le \sqrt{m} |\cA_{3,m}|^{1/3} \left(\sum_{j: p_j \le \delta_m} p_j\right)^2\\
    & \lesssim \sqrt{m} |\cA_{3,m}|^{1/3} |\cD_{3,m}|^2 \delta_m^2\\
    & = o(1), \labthis \label{exp-3.2.4}
\end{align*}
where the last inequality follows from the assumption \eqref{ass-df-2-2}.
For any $j_1 \in \cD_{3,m}$, we define 
\begin{align*}
    T_{\Delta_2}(j_1;\cC_{3,m}^{(2)}) = \sum_{j_2,j_3: (j_1,j_2,j_3) \in \cC_{3,m}^{(2)}} \indc{h_{i_1} \ni \{j_1,j_2\}} \indc{h_{i_2} \ni \{j_2,j_3\}} \indc{h_{i_2} \ni \{j_3,j_1\}} ,
\end{align*}
as the local subgraph count for the vertex $j_1$ for the hyperedges $(i_1,i_2,i_3)$. Thus, we get
\begin{align*}
0 &\le  \frac{\sqrt{m}}{\binom{m}{2}}\sum_{i_1 <i_2} \Ex{\sum_{\in \cC_{3,m}^{(2)}}\indc{h_{i_1} \ni \{j_1,j_2\}} \indc{h_{i_2} \ni \{j_2,j_3\}} \indc{h_{i_2} \ni \{j_3,j_1\}} }\\ 
 & = \sqrt{m} \Ex{\sum_{\cD_{3,m}} T_{\Delta_2}(j_1;\cC_{3,m}^{(2)}) \clr{\indc{h_1 \ni \{j_1, j_2, j_3\}} \indc{h_2 \ni \{j_1, j_2\}}
  + \indc{h_1 \ni \{j_1, j_2, j_3\}} } }\\
 &= I + II \qquad\text{(say)}.
\end{align*}
Then,
\begin{align*}
I & \le \sqrt{m}\sum_{j_1 \in \cD_{3,m}} \sqrt{\Ex{T_{\Delta_2}^2(j_1;\cC_{3,m}^{(2)})} \Ex{ \indc{h_1 \in \{j_1,j_2,j_3\}}\indc{h_2 \in \{j_1,j_2\}}}}\\
& \lesssim \sqrt{m} |\cD_{3,m}| \sqrt{ \Ex{T_{\Delta_2}^2(j_1;\cC_{3,m}^{(2)})} p_{j_1,j_2,j_3} p_{j_1,j_2}}\\
& \le \sqrt{m} |\cD_{3,m}| \delta_m\\
& = o(1), \labthis \label{exp-3.2.5}
\end{align*}
where the last inequality follows from \eqref{ass-df-2-2}. We further use H\"{o}lder's inequality and assumption \eqref{ass-df-2-2} to obtain

\begin{align*}
II & \le \sqrt{m}\sum_{j_1 \in \cD_{3,m}} \Ex{T_{\Delta_2}^4(j_1;\cC_{3,m}^{(2)})}^{1/4}(\Ex{\indc{h_1 \ni \{j_1,j_2,j_3\}}})^{3/4}\\
& \le \sqrt{m} |\cD_{3,m}| \Ex{T_{j_{1},\cC_\delta}^4}^{1/4} (p_{j_1,j_2,j_3})^{3/4}\\
& \le \sqrt{m} |\cD_{3,m}| \delta_m^{3/4}\\
& = o(1). \labthis \label{exp-3.2.6}
\end{align*}
Finally, we combine \eqref{exp-3.2.3}, \eqref{exp-3.2.4}, \eqref{exp-3.2.5} and \eqref{exp-3.2.6} to get $\sqrt{m}T^{(3)} (\Delta_2) \larw{1} 0$. Since $0 \le T^{(3)}_d (\Delta_2) \le T^{(3)} (\Delta_2)$, we further have $\sqrt{m}T^{(3)}_d(\Delta_2) \larw{1} 0$. This proves the (ii).

We now proceed to prove (iii). 
To this end, We decompose the degree filtered Type 3 triangle frequency as
\begin{align*}
    T_{d}(\Delta_3) = T_{d}^{(1)}(\Delta_3) + T_{d}^{(2)}(\Delta_3) + T_{d}^{(3)}(\Delta_3),
\end{align*}
according to the sets $\cA_{3,m}$, $\cB_{3,m}$ and $\cC_{3,m}$ respectively. We also decompose $T(\Delta_3) = T^{(1)}(\Delta_3)+ T^{(2)}(\Delta_3)+T^{(3)}(\Delta_3)$.
For any $\epsilon>0$ such that $d \le (1-\epsilon)mp_{j}$, we apply Chernoff's bound to get
\begin{align*}
    P(D_{{j}} \le d) &\le \exp\clr{-\dfrac{\epsilon^2}{2}mp_{j}} \le \exp\clr{-\dfrac{\epsilon^2}{2}m\delta_m},
\end{align*}
for any $ j \in \clr{j_1,j_2,j_3}$ such that $(j_1,j_2,j_3) \in \cA_{3,m}$.
Also recall that $|\cA_{3,m}| \lesssim \delta_m^{-3}$ which implies
\begin{align*}
    0\le \Ex{\sqrt{m}\plr{T^{(1)}(\Delta_3) - T_{d}^{(1)}(\Delta_3)}} 
    & \lesssim \sqrt{m} |\cA| \exp\clr{-\dfrac{\epsilon^2}{2}m\delta_m}\\
    & \lesssim \dfrac{\sqrt{m}}{\delta_m^3}\exp\clr{-\dfrac{\epsilon^2}{2}m\delta_m}\\
    & = o(1) \labthis \label{exp-3.3.1}
\end{align*}
Thus we have $\sqrt{m}[T^{(1)}(\Delta_3) - T_{d}^{(1)}(\Delta_3)] \larw{1} 0$.

Note that
\begin{align*}
    T_{d}^{(2)}(\Delta_2) \le \dfrac{1}{\binom{m}{3}} \sum_{i_1<i_2<i_3} \sum_{\cB_{3,m}} \Bigg\{ \indc{h_{i_1} \ni \{j_1,j_2\}} \indc{h_{i_2} \ni \{j_2,j_3\}} \indc{h_{i_3} \ni \{j_3,j_1\}} \\
    \indc{\widetilde{D}_{j_1,-\clr{i_1,i_2,i_3}} \ge d-3,\widetilde{D}_{j_2,-\clr{i_1,i_2,i_3}} \ge d-3,\widetilde{D}_{j_3,-\clr{i_1,i_2,i_3}} \ge d-3} \Bigg\}.
\end{align*}
Therefore, for any $\epsilon_3 \in (0,1)$ such that $ m-d < (1-\epsilon_3) (m-3)(1-p_j)$, we apply Chernoff's bound to obtain the following,
\begin{align*}
    & \max_{\cB_{3,m}} \clr{\pr{\widetilde{D}_{j_1, -\clr{i_1,i_2,i_3}} \ge d-3, \widetilde{D}_{j_2, -\clr{i_1,i_2,i_3}} \ge d-3, \widetilde{D}_{j_3, -\clr{i_1,i_2,i_3}} \ge d-3}}\\
    & \le \max_{\cB_{3,m}} \min_{j = j_1,j_2,j_3} \clr{\pr{\widetilde{D}_{j, -\clr{i_1,i_2,i_3}} \ge d-3}}\\
    & \le \max_{\cB_{3,m}} \min_{j = j_1,j_2,j_3} \exp\clr{-\frac{\epsilon_3^2}{2}(m-3)(1-p_j)},\\
   & \le \exp\clr{-\frac{\epsilon_3^2}{2}(m-3)(1-\delta_m)},
\end{align*}
for any $1 \le i_1<i_2<i_3\le m$. 

Moreover, applying Markov's inequality we have
\begin{align*}
   &\pr{\sqrt{m}|T_{d}^{(2)}(\Delta_3)| \ge \epsilon} \\
   & \le \dfrac{\sqrt{m}}{\epsilon\binom{m}{3}} \sum_{\cB_{3,m}}\sum_{i_1<i_2<i_3} \Ex{\indc{h_{i_1} \ni \{j_1,j_2\}} \indc{h_{i_2} \ni \{j_2,j_3\}} \indc{h_{i_3} \ni \{j_3,j_1\}}} \\
   & \qquad\qquad\Ex{\indc{\widetilde{D}_{j_1, -\clr{i_1,i_2,i_3}} \ge d-3, \widetilde{D}_{j_2, -\clr{i_1,i_2,i_3}} \ge d-3, \widetilde{D}_{j_3, -\clr{i_1,i_2,i_3}} \ge d-3}}\\
   &  = \dfrac{\sqrt{m}}{\epsilon\binom{m}{3}} \sum_{\cB_{3,m}}\sum_{i_1<i_2<i_3} p_{j_{1}j_{2}}p_{j_{2}j_{3}}p_{j_{1}j_{3}}\\
   & \qquad\qquad\qquad\pr{\widetilde{D}_{j_1, -\clr{i_1,i_2,i_3}} \ge d-3, \widetilde{D}_{j_2, -\clr{i_1,i_2,i_3}} \ge d-3, \widetilde{D}_{j_3, -\clr{i_1,i_2,i_3}} \ge d-3}\\
   & \le \dfrac{\sqrt{m}}{\epsilon} \max_{\cB_{3,m}} \clr{\pr{\widetilde{D}_{j_1, -\clr{i_1,i_2,i_3}} \ge d-3, \widetilde{D}_{j_2, -\clr{i_1,i_2,i_3}} \ge d-3, \widetilde{D}_{j_3, -\clr{i_1,i_2,i_3}} \ge d-3}} \left(\sum_{j \in \cV} p_j\right)^3\\
   & = o(1), \labthis \label{exp-3.3.2}
\end{align*}
where the second last inequality follows from the fact the $ p_{j_1j_2} \le \sqrt{p_{j_1}p_{j_2}}$ for any $j_1,j_2 \in \cV$.

For any $j \in \cD_{3,m}$, we define  
\begin{align*}
    T_{\Delta_3}(j_{1};\cC_{3,m}^{(2)}) = \sum_{j_2,j_3: (j_1,j_2,j_3) \in C_{3,m}^{(2)}} \indc{h_{i_1} \ni \{j_1,j_2\}} \indc{h_{i_2} \ni \{j_2,j_3\}} \indc{h_{i_3} \ni \{j_3,j_1\}} 
\end{align*}
as the triangle local triangle count for vertex $j_1$ formed by hyperedges $h_{i_1}, h_{i_2},h_{i_3}$.
Then,
\begin{align*}
    0 & \le  \Ex{\frac{\sqrt{m}}{\binom{m}{3}} \sum_{i_1<i_2<i_3} \sum_{\cC_{3,m}^{(2)}} \indc{h_{i_1} \ni \{j_1,j_2\}} \indc{h_{i_2} \ni \{j_2,j_3\}} \indc{h_{i_3} \ni \{j_3,j_1\}}} \\
    & = \sqrt{m} \Ex{\sum_{j_1 \in \cD_{3,m}} T_{\Delta_3}(j_{1};\cC_{3,m}^{(2)}) \indc{h_1 \ni \{j_1,j_2\}, h_2 \ni \{j_1, j_3\}}}\\
    & \le \sqrt{m}\sum_{j_1 \in \cD_{3,m}} \sqrt{\Ex{T_{\Delta_3}^2(j_{1};C_{3,m}^{(2)})} \Ex{ \indc{h_1 \ni \{j_1,j_2\}, h_2 \ni \{j_1, j_3\}}}}\\
    & \le \sqrt{m} \sum_{j_1 \in \cD_{3,m}} \sqrt{ \Ex{T_{\Delta_3}^2(j_{1};\cC_{3,m}^{(2)})}}\,\, p_{j_1}\\
    & \lesssim \sqrt{m} |\cD_{3,m}| \delta_m\\
    & = o(1), \labthis \label{exp-3.3.3}
\end{align*}
where the last inequality follows from \eqref{ass-df-3-2}.
For the set $\cC_{3,m}^{(1)}$, we have
\begin{align*}
    0& \le \Ex{ \frac{\sqrt{m}}{\binom{m}{3}}\sum_{i_1<i_2<i_3} \sum_{\cC_{3,m}^{(1)}} \indc{h_{i_1} \ni \{j_1,j_2\}} \indc{h_{i_2} \ni \{j_2,j_3\}} \indc{h_{i_3} \ni \{j_3,j_1\}}} \\
    & = \sqrt{m}\sum_{\cC_{3,m}^{(1)}} p_{j_1j_2}p_{j_2j_3}p_{j_1j_3}\\
    & \le \sqrt{m} |\cA_{3,m}|^{1/3} \left(\sum_{j: p_j < \delta_m} p_{j}^2\right)\left(\sum_{j: p_j < \delta_m} p_{j}\right)\\
    & \le  \sqrt{m} |\cA_{3,m}|^{1/3} |\cD_{3,m}|^2 \delta_m^3\\
    & = o(1) \labthis \label{exp-3.3.4}
\end{align*}
where the second last equality follows from the assumption \eqref{ass-df-3-2}.

Finally, we focus on the set $\cC_{3,m}^{(0)}$. We compute
\begin{align*}
    0 & \le \Ex{\frac{\sqrt{m}}{{\binom{m}{3}}} \sum_{i_1<i_2<i_3} \sum_{\cC_{3,m}^{(0)}} \indc{h_{i_1} \ni \{j_1,j_2\}} \indc{h_{i_2} \ni \{j_2,j_3\}} \indc{h_{i_3} \ni \{j_3,j_1\}}}\\
    & = \sqrt{m}\sum_{\cC_{3,m}^{(0)}} p_{j_1j_2}p_{j_2j_3}p_{j_1j_3}\\
    & \le \sqrt{m} \sum_{\cC_{3,m}^{(0)}} p_{j_1}p_{j_2}p_{j_3}\\
    & \le \sqrt{m} \left(\sum_{j \in \cD_{3,m}} p_{j}\right)^3\\
    & \lesssim \sqrt{m}  |\cD_{3,m}|^3 \delta_m^3\\
    & = o(1), \labthis \label{exp-3.3.5}
\end{align*}
where the first inequality holds due to the fact that $p_{j_1j_2} \le \sqrt{p_{j_1}p_{j_2}}$ for any $j_1,j_2 \in \cV$ and the last inequality follows from the assumption \eqref{ass-df-3-2}. Thus, combining equations \eqref{exp-3.3.3}, \eqref{exp-3.3.4} and \eqref{exp-3.3.5} we get $\sqrt{m} T^{(3)}(\Delta_3) \larw{1} 0$. Since $0 \le T_d^{(3)}(\Delta_3) \le  T^{(3)}(\Delta_3) $, we further have $\sqrt{m}T_d^{(3)}(\Delta_3) \larw{1} 0$. Finally, we get $\sqrt{m}[T(\Delta_3) -T_{d}(\Delta_3)] \parw 0$.
   
    This completes the proof.
\end{proof}

We now apply Lemma \ref{clt-df-thm-tri} to prove the Theorem \ref{thm3.5-df-tri-poly} for the polynomial decay on $ \clr{p_{(j)}}_{j \ge 1}$ as follows.

\begin{proof}[Proof of Theorem \ref{thm3.5-df-tri-poly}]

We set $\delta_m = O(1/m^{1-\nu})$ and $\delta_m' = \Theta(1/m)$. Solving the equations ${x_1}^{-\alpha} = m^{-(1-\nu)}$ and ${x_2}^{-\alpha} = m^{-1}$ for $x_1$ and $x_2$ respectively, we get $|\cD_{3,m}| = x_2-x_1 = m^{1/\alpha}-m^{(1-\nu)/\alpha}=O(m^{1/\alpha})$. Also, note that $|\cA_{3,m}|^{1/3} \lesssim \delta_m \lesssim m^{1-\nu}$. For Type 2 triangles, we set $0 < \nu<\min\{1/3-1/\alpha,1/2-2/\alpha\}$. Then, the asymptotic bounds in \eqref{ass-df-2-2} reduces to the following
\begin{align*}
    & \sqrt{m} |\cD_{3,m}|^3 \delta_m^2 = O\plr{m^{2(\alpha-3/4+1/\alpha)}} = o(1),\\
    & \sqrt{m} |\cA_{3,m}|^{1/3} |\cD_{3,m}|^2 \delta_m^2= O\plr{m^{\nu - 1/2+2/a}}= o(1),\\
    & \sqrt{m} |\cD_{3,m}| \delta_m=O\plr{m^{\alpha-1/2+1/\alpha}} = o(1), \\
    & \sqrt{m} |\cD_{3,m}| \delta_m^{3/4}= O\plr{ m^{3(\nu-1/3 +1/\alpha)/4}} = o(1),
\end{align*}
where the inequalities follows from the fact that $d\ll m^{\min\{1/3-1/\alpha,1/2-2/\alpha\}}$.

For Type 3 triangles, set $\nu < 1/2 +1/a$. Thus, the asymptotic bounds for \eqref{ass-df-3-2} reduces to the following,
\begin{align*}
    &\sqrt{m} |\cD_{3,m}| \delta_m =O\plr{ m^{\nu-1/2 +1/\alpha}} =o(1),\\
    & \sqrt{m}  |\cD_{3,m}|^2 \delta_m^2 = O\plr{ m^{2(\nu -3/4+1/\alpha)}} = o(1), \\
    &  \sqrt{m}  |\cD_{3,m}|^3 \delta_m^3 = O\plr{m^{3(\nu-5/6+1/\alpha)} } =o(1),
\end{align*}
where the inequalities follows from the fact that $d\ll m^{1/2-1/\alpha}$. This completes the proof.
\end{proof}

\section{Proof of Section \ref{sec-without-mult}}\label{sec: without mult}

We now present the proof of Proposition \ref{prop3.6-as-wm-fin}.

\begin{proof}[Proof of Proposition \ref{prop3.6-as-wm-fin}]


Suppose that $H$ has $e$ edges. Define the collection:
\begin{align*}
\cF = \{ F \iso H, \ \mathcal{V}(F) \subseteq \mathcal{V}  \}.
\end{align*}
For each $F \in \cF$, define the probability:
\begin{align*}
q_F = \max_{  \text{simple } F_\fC  : \widebar{F}  = F  }  \pr{F_\fC \subseteq G_{\fC}(\mathcal{H}_e)}.
\end{align*}
Since $|\cF| < \infty$, 
\begin{align*}
    \inf_{F \in \cF} \{ q_F  : q_F >0  \} = \epsilon > 0.
\end{align*}
Let $ \cF_* =  \{ F :  q_F > 0  \} $ and  $\cS = |\cF_*|$.  Thus,
\begin{align*}
&\pr{\exists M :  \forall m \geq M, \widetilde{T}^{(m)}(H) = \cS}\\
&= \pr{\plr{\bigcup_{m=1}^\infty \bigcap_{F \in \cF_*} \clr{F \subseteq \widebar{G}_m}} \bigcap \plr{\bigcap_{m=1}^\infty \bigcup_{F \in \cF_*^c} \{F \not \subseteq \widebar{G}_m\} }}.
\end{align*}
Clearly, for $F \in \cF_*^c$, $\pr{ F \subseteq \widebar{G}_m} = 0$. Therefore, 
\begin{align*}
    \pr{ \bigcap_{m=1}^\infty \bigcup_{F \in \cF_*^c} \clr{F \not \subseteq \widebar{G}_m}} = \pr{\bigcup_{m=1}^\infty \bigcap_{F \in \cF_*^c} \clr{F \subseteq \widebar{G}_m}^c}= 1,
\end{align*}
and consequently,
\begin{align*}
\pr{ \exists M :  \forall m \geq M, \widetilde{T}^{(m)}(H) = \cS}
= \pr{\bigcup_{m=1}^\infty \bigcap_{F \in \cF_*} \clr{F \subseteq \widebar{G}_m}  }.
\end{align*}
Now, for $q  \in \mathbb{N}$, let  $\mathcal{H}(q) = (h_{(q-1)e}, \ldots h_{qe})$ and $\widebar{G}(q)$ denote the colorless restriction of $G_{\fC}(\mathcal{H}(q))$.  For an arbitrary ordering of elements in $\cF_*$,  define the event:
\begin{align*}
A_l = \bigcap_{s=1}^{\cS} \clr{ F_s \subseteq   \widebar{G}((\cS(l-1)+s )   }.
\end{align*}
Observe that, for any $l \in \mathbb{N}$,
\begin{align*}
\pr{A_l} \geq \prod_{F \in \cF_*} q_F \geq \epsilon^{\cS} > 0.
\end{align*}
Therefore it follows that,
\begin{align*}
\pr{ \exists M :  \forall m \geq M, \widetilde{T}^{(m)}(H) = \cS} \geq \pr{\bigcup_{l=1}^\infty A_l} \geq \pr{\limsup_{l \rightarrow \infty} A_l} = 1,
\end{align*}
where the last line follows from the Second Borel-Cantelli Lemma. The claim follows.    

\end{proof}

We now present some lemmas which turn out to be crucial for our proof of Theorem \ref{th3.7-wo-mul-thm}. 
For any $j \ge 1$, we can write $ E_j = \sum_{i=1}^m \indc{h_i = \scA_j}$ as the degree of the $k$-order interaction $\scA_j$. Thus, the without-multiplicity statistic reduces to 
\begin{align}
    \widetilde{T}_k^{(m)} = \sum_{j=1}^\infty X_j = \sum_{j=1}^\infty \indc{E_j > 0} = \sum_{j=1}^\infty \indc{\sum_{i=1}^m \indc{h_i = \scA_j} > 0}.  \label{wo-decomp-1}
\end{align}
Let, $\Delta_{k,m}^2= \Var{Z_{k,\delta}}$ for $ k=2,3$. For brevity, we suppress the subscript $m$ in $\Delta_{k,m}$ henceforth.
Moreover, we denote $W_i = \sum_{j=d_2}^{\infty} 1(h_i = \scA_j)$ and write
\begin{align}
 \sum_{j=d_2}^{\infty}E_j = \sum_{i =1}^{m}\sum_{j=d_2}^{\infty} {\indc{h_i = \scA_j}} = \sum_{i =1}^{m} {W_i}. \label{wo-decomp-2}
\end{align}
Later in the proof of Theorem \ref{th3.7-wo-mul-thm}, we show that the variances can be approximated by 
\begin{align}
    \Delta_2^2 = \sum_{j=d_1}^{d_2-1}\Var{X_j}\qquad\text{and} \qquad \Delta_3^2 = \Var{\sum_{i=1}^m W_i }, \label{var-apx-wo}
\end{align}
see \eqref{exp-4.11} and \eqref{exp-4.8} respectively for more details.

The proof of the theorem is based on the negative dependence of $\clr{X_j}_{j \ge 1}$ which is presented in Lemma \ref{lem-1}. Lemma \ref{lem-2} discusses the approximation of the tail sum of $\clr{p_j}_{j \ge 1}$. Lemma \ref{lem-3} establishes the relation between the variances of $Z_{2, \delta}$ and $Z_{3,\delta}$. Finally, Lemma \ref{lem-4} establishes a bound on the absolute value of the covariances between $X_j$ and $X_k$.

\begin{lemma}\label{lem-1}
    $\{X_j\}_{j=1}^\infty$ are negatively associated.
\end{lemma}

\begin{lemma} \label{lem-2}
For any $\alpha>2$, consider $p_j \propto 1/j^\alpha$ for all $j=1,2,\dots,\infty$, where $\sum_{j=1}^\infty p_j=1$. Then for $d\in \bbN$, there exists constants $0 < c_\alpha < C_\alpha$ depending on $\alpha$ such that 
\begin{align*}
    \frac{c_\alpha}{d^{\alpha-1}} \le \sum_{j=d}^\infty p_j\le \frac{C_\alpha}{d^{\alpha-1}}.
\end{align*}
\end{lemma}

\begin{lemma} \label{lem-3}
    Consider the variance terms $\Delta_{2}^2$ and $\Delta_{3}^2$. Then, there exists constants such that $ 0 < c_{\alpha, \delta} < C_{\alpha, \delta}$ depending on $\alpha$ and $\delta$ such that $c_{\alpha, \delta} \le \Delta_2^2/\Delta_{3}^2 \le C_{\alpha, \delta}$.
\end{lemma}

\begin{lemma} \label{lem-4}
    For any $ j \ne k$,
    \begin{align*}
        \md{\Cov{X_j, X_k}}\le m(p_j + p_k)^2 e^{-m(p_j + p_k)}.
    \end{align*}
\end{lemma}

\begin{proof}[Proof of Lemma \eqref{lem-1}]
We start by showing that an infinite multinomial distribution exists.  To this end, we define a projection of this multinomial distribution to a distribution defined on $\mathbb{R}^n$. Let $p_0 = \pr{|h| \neq k}$,  and for $j > 0$, let $p_j = \pr{h = \mathscr{A}_j}$.  For any $n > 0$, define the probability measure:
\begin{align*}
\mu_n(\cdot) =  \mathbb{P}_n((E_0, \ldots, E_{n-1}) \in \cdot),
\end{align*}
where $\mathbb{P}_n$ is the following multinomial distribution:
\begin{align*}
(E_0, \ldots, E_{n-1}, F_n)  \sim \text{Multinomial}(m;p_0, \ldots, p_{n-1}, q_{n}),
\end{align*}
and $q_n = 1- \sum_{i=0}^{n-1} p_i$.
    It is clear that $\mu_n$ is a probability measure on $\plr{\bbR^n, \mathscr{B}(\bbR^n)}$, where $\mathscr{B}(\bbR^n)$ is the Borel sigma algebra on $\bbR^n$. Now, we take any $B \in \mathscr{B}(\bbR^n)$ and compute,
    \begin{align*}
        \mu_{n+1}(B \times \bbR) & = \pr{(E_0,\dots,E_{n-1})\in B, E_{n} \in \bbR}\\
        & = \pr{(E_0,\dots,E_{n-1})\in B}\\
        & = \mu_n(B).
    \end{align*}
   Now, the Kolmogorov consistency theorem implies that  a distribution on $\{E_j\}_{j=0}^\infty$ is well-defined. We apply \cite{joag1983negative} to get that $\{E_j\}_{j=0}^\infty$ are negatively associated. Since $X_j = \indc{E_j>0}$ is a monotone transformation, we finally have $\{X_j\}_{j=1}^\infty$ are negatively associated.
\end{proof}

\begin{proof}[Proof of Lemma \eqref{lem-2}]
    We use the integral bound to compute
    \begin{align*}
        \int_{d+1}^\infty \frac{1}{x^\alpha} \,dx & \le \sum_{j=d}^\infty \frac{1}{j^\alpha} \le \int_{d-1}^\infty \frac{1}{x^\alpha} \,dx\\
         \implies \frac{1}{(\alpha-1)(d+1)^{\alpha-1}} & \le \sum_{j=d}^\infty \frac{1}{j^\alpha} \le \frac{1}{(\alpha-1)(d-1)^{\alpha-1}}.
    \end{align*}
    This completes the proof.
\end{proof}

\begin{proof}[Proof of Lemma \eqref{lem-3}]
    From \eqref{var-apx-wo}, we note that
   \begin{align*}
        \Delta_2^2 &= \Theta\plr{\sum_{j=d_1}^{d_2-1} \plr{1-p_j}^m \plr{1-\plr{1-p_j}^m}}, \\
        \Delta_3^2 &= \Theta\plr{m \sum_{j = d_2}^\infty p_j}.
    \end{align*}
    Then, we have
    \begin{align*}
        \frac{\Delta_{2}^2}{\Delta_{3}^2} & \gtrsim \frac{(d_2-d_1)\min_{d_1 \le j \le d_2 } \clr{(1-p_j)^m\plr{1-(1-p_j)^m}}}{m\sum_{j=d_2}^\infty p_j}\\
        & = \frac{(d_2-d_1) \clr{(1-p_{d_1})^m\plr{1-(1-p_{d_2})^m}}}{m\sum_{j=d_2}^\infty p_j} \\
        & \gtrsim \frac{d_2^{\alpha-1} (C_\delta - c_\delta)m^{1/\alpha}\;e^{-mp_{d_1}}\plr{1-e^{-mp_{d_2}}}}{m}\\
        & \gtrsim (C_\delta - c_\delta) C_\delta^{\alpha-1}e^{-1/c_\delta^\alpha}\plr{1-e^{-1/C_\delta^\alpha}}, \labthis \label{lem-exp-3.1}
    \end{align*}
    for large enough $m$, where the third inequality follows from Lemma \eqref{lem-2}.
    Similarly, we have
    \begin{align}
        \frac{\Delta_{2}^2}{\Delta_{3}^2} \lesssim (C_\delta - c_\delta) C_\delta^{\alpha-1}e^{-1/C_\delta^\alpha}\plr{1-e^{-1/c_\delta^\alpha}} \label{lem-exp-3.2}
    \end{align}
    for large enough $m$. Combining \eqref{lem-exp-3.1} and \eqref{lem-exp-3.2} we have the result.
\end{proof}

\begin{proof}[Proof of Lemma \eqref{lem-4}]
    Note that for any integer $m>0$ we have
    \begin{align*}
        \plr{1+\frac{x}{m}}^m \ge e^x \plr{1-\frac{x^2}{m}}\quad \text{for any $|x|\le m$}.
    \end{align*}
    Using this fact, we have 
    \begin{align}
        (1-p_j-p_k)^m \ge e^{-m(p_j+p_k)} \plr{1-m(p_j+p_k)} \label{lem-exp-4.1}.
    \end{align}
    Moreover, for any integer $m>0$ and $x\in \bbR$, we have $(1-x)^m \le e^{-mx}$. So, we also get 
    \begin{align}
        e^{-mp_k} \ge (1-p_k)^m. \label{lem-exp-4.2}
    \end{align}
    Combining \eqref{lem-exp-4.1} and \eqref{lem-exp-4.2} we have the following inequality,
    \begin{align*}
         0 \ge \Cov{X_j,X_k} & \ge e^{-m(p_j+p_k)} \plr{1-m(p_j+p_k)^2} - e^{-m(p_j+p_k)} \\
         &= -m (p_j+p_k)^2e^{-m(p_j+p_k)}\\
         \implies |\Cov{X_j,X_k}| & \le m (p_j+p_k)^2e^{-m(p_j+p_k)}.
    \end{align*}
\end{proof}

We now present the proof of Theorem \ref{th3.7-wo-mul-thm}.

\begin{proof}[Proof of Theorem \ref{th3.7-wo-mul-thm}]


Our main aim is to show that
\begin{align*}
    \frac{\widetilde{T}_k^{(m)} - \Ex{\widetilde{T}_k^{(m)}}}{\Delta_m} \darw N(0,1).
\end{align*}
Therefore, it is enough to show for fixed $\epsilon > 0,$
\begin{align}
    \supx{\pr{\frac{\widetilde{T}_k^{(m)} - \Ex{\widetilde{T}_k^{(m)}}}{\Delta_m} \le x} - \Phix} < \epsilon, \label{exp-4.1}
\end{align}
where $\Phi(\cdot)$ is the CDF of standard normal distribution.

Combining \eqref{decomp-wo} and \eqref{exp-4.1}, we obtain
    \begin{align*}
        &\supx {\pr{\frac{\widetilde{T}_k^{(m)} - \Ex{\widetilde{T}_k^{(m)}}}{\Delta_m} \le x} - \Phix} \\
        &\le  \supx{\pr{\frac{ Z_{2,\delta} +  Z_{3,\delta} - \Ex{Z_{2,\delta} +  Z_{3,\delta}}}{\Delta_m} \le x} - \Phix} + \pr{\md{\frac{ Z_{1,\delta}-\Ex{Z_{1,\delta}}}{\Delta_m} } > \epsilon} + O(\epsilon) \labthis \label{exp-4.2}
    \end{align*}

To prove (i), it is enough to show that for any $\epsilon >0$, 
 there exists $\delta >0$ such that $\pr{\md{{ Z_{1,\delta} -\Ex{Z_{1,\delta}}}}/{\Delta_m} > \epsilon} < \delta$ for sufficiently large $m$.
    Using Chebyshev's inequality we further have
    \begin{align*}
        \pr{\md{\frac{ Z_{1,\delta} - \Ex{Z_{1,\delta}}}{\Delta_m}} > \epsilon} \le \frac{\Var{Z_{1,\delta}}}{\Delta_m^2\epsilon^2}.
    \end{align*}
    Now, 
    \begin{align*}
       \frac{\Var{Z_{1,\delta}}}{\Delta_m^2} & = \frac{1}{\Delta_m^2} \Var{\sum_{j=1}^{d_1-1}X_j}\\
       & \lesssim \frac{1}{\Delta_m^2} \sum_{j=1}^{d_1-1} \Var{X_j} \\ 
       & = \frac{1}{\Delta_m^2} \sum_{j=1}^{d_1-1} (1-p_j)^m \clr{1-(1-p_j)^m}\\
       & \le \clr{\max_{1 \le j \le d_1-1}\pr{E_j = 0}}\frac{d_1}{m\sum_{j=d_2}^\infty p_j}\frac{\Delta_{3}^2}{\Delta_m^2}\\
        & \le \clr{\max_{1 \le j \le d_1-1}\pr{E_j = 0}} \frac{m^{1/\alpha}(d_2-1)^{\alpha-1}}{m} \\
        & \le \clr{\max_{1 \le j \le d_1-1}\pr{E_j = 0}} \plr{1-\frac{1}{m^{1/\alpha}}}^{\alpha-1}\\
        & = O\plr{\max_{1 \le j \le d_1-1}\pr{E_j = 0}},
    \end{align*}
    where the second inequality follows from negative association of $X_j$'s (c.f. Lemma \ref{lem-1}) and the fourth equality follows from Lemma \ref{lem-2} and Lemma \ref{lem-3}.
    Hence, we choose $d_1$ such that 
    \begin{align}
        \max_{1 \le j \le d_1-1}\pr{E_j = 0} < C \delta, \label{exp-4.3}
    \end{align}
    for some constant $0 < C< \infty$.
    This completes the proof of (i).

Next, we proceed to prove (ii). We would separately show that 
\begin{align}
    \frac{Z_{2,\delta} - \Ex{Z_{2,\delta}}}{\Delta_{2}} \darw N(0,1), \quad \text{and} \quad\frac{Z_{3,\delta} - \Ex{Z_{3,\delta}}}{\Delta_3} \darw N(0,1). \label{exp-4.4}
\end{align}
First, we note that 
\begin{align*}
    &\supx {\pr{\frac{Z_{3,\delta}-\Ex{Z_{3,\delta}}}{\Delta_3} \le x} - \Phi(x)} \\
    & = \supx{\pr{\sum_{j=d_2}^{\infty} \frac{X_j - \Ex{X_j}}{\Delta_3} \le x} - \Phi(x)} \\
    & = \supx {\pr{\sum_{j=d_2}^{\infty}\plr{\frac{E_j - \Ex{E_j}}{\Delta_3}} + \sum_{j=d_2}^{\infty}\plr{\frac{X_j-E_j - \Ex{X_j - E_j}}{\Delta_3}} \le x} - \Phi(x)} \\
    & \le \supx{\pr{\sum_{j=d_2}^{\infty}\plr{\frac{E_j - \Ex{E_j}}{\Delta_3}} \le x} - \Phi(x)} + \\
    & \qquad\qquad\qquad \pr{\md{\sum_{j=d_2}^{\infty}\plr{\frac{X_j-E_j - \Ex{X_j - E_j}}{\Delta_3}}} > \epsilon} + O(\epsilon). \labthis \label{exp-4.5}
\end{align*}
Using Chebyshev's inequality we have the following:
\begin{align*}
    \pr{\md{\sum_{j=d_2}^{\infty}\plr{\frac{X_j-E_j - \Ex{X_j - E_j}}{\Delta_3}}} > \epsilon} \le \frac{\Var{\sum_{j=d_2}^{\infty}({X_j-E_j})}}{\epsilon^2\Delta_{3}^2}. \labthis \label{exp-4.6}
\end{align*}
Using the relation $E_j - X_j = (E_j-1)\indc{E_j > 1}$ for any $j \ge 1$, we can write
\begin{align*}
    \Delta_3^{-1} \sum_{j=d_2}^{\infty} \clr{E_j - \indc{E_j > 0}} =  \Delta_3^{-1} \sum_{j=d_2}^{\infty} \clr{(E_j - 1) \indc{E_j > 1}}.
\end{align*}
Then, \eqref{exp-4.6} reduces to as follows,
\begin{align*}
    &\Var{\Delta_3^{-1} \sum_{j=d_2}^{\infty} (E_j - 1) \indc{E_j > 1}} \\
    & \le \Delta_3^{-2} \sum_{j=d_2}^{\infty} \Ex{(E_j - 1)^2 \indc{E_j > 1}}\\
    & \le \Delta_3^{-2} \sum_{j=d_2}^{\infty} \sqrt{\Ex{E_j^4} \pr{E_j > 1}}\\
        & \le \sqrt{\max_{j \ge d_2} \pr{E_j > 1}} \sum_{j =d_2}^{\infty} \frac{\sqrt{\Ex{E_j^4}}}{\Delta_{3}^2}\\
    & \lesssim \sqrt{\max_{j \ge d_2} \pr{E_j > 1}} \frac{\sum_{j =d_2}^{\infty}\sqrt{m^4p_j^4 + m^3p_j^3 + m^2p_j^2 + mp_j}}{\Delta_{3}^2}\\
    & \lesssim \sqrt{\max_{j \ge d_2} \pr{E_j > 1}}\\ 
    & \quad \clr{\frac{\sum_{j =d_2}^{\infty}(mp_j)^2 + \sum_{j =d_2}^{\infty}(mp_j)^{3/2} + \sum_{j =d_2}^{\infty}mp_j + \sum_{j =d_2}^{\infty}\sqrt{mp_j}}{m\sum_{j \ge d_2} p_j}}, \labthis \label{exp-4.7}
\end{align*}
where the first inequality uses the negative association of $\clr{E_j}_{j \ge 1}$.
Note that for $\alpha > 2$
\begin{align*}
    & \frac{\sum_{j =d_2}^{\infty}(mp_j)^2 + \sum_{j =d_2}^{\infty}(mp_j)^{3/2} + \sum_{j =d_2}^{\infty}mp_j + \sum_{j =d_2}^{\infty}\sqrt{mp_j}}{m\sum_{j \ge d_2} p_j}\\
    = & \frac{\frac{m^2}{(d_2-1)^{2\alpha-1}} + \frac{m^{3/2}}{(d_2-1)^{\frac{3\alpha}{2}-1}} + \frac{m}{(d_2-1)^{\alpha-1}} + \frac{{\sqrt{m}}}{(d_2-1)^{\frac{\alpha}{2}-1}}}{\frac{m}{(d_2+1)^{\alpha-1}}}\\
    \lesssim & \plr{\frac{d_2+1}{d_2-1}}^{\alpha -1 } \clr{\frac{m}{(d_2-1)^{\alpha}} + \frac{\sqrt{m}}{(d_2-1)^{\frac{\alpha}{2}}} + 1 + \frac{(d_2-1)^{\frac{\alpha}{2}}}{{\sqrt{m}}}}\\
    \lesssim & \plr{1 + \frac{2}{d_2-1}}^{\alpha -1 }\clr{\frac{1}{(1-1/d_2)^\alpha} + \frac{1}{(1-1/d_2)^{\alpha/2}} + 1 + \plr{1-\frac{1}{d_2}}^{\alpha/2}}\\
    = & O(1).
\end{align*}
Then combining \eqref{exp-4.6} and \eqref{exp-4.7} choose $d_2$ such that 
\begin{align}
    \pr{\md{\sum_{j=d_2}^{\infty} \plr{\frac{X_j - E_j - \Ex{X_j - E_j}}{\Delta_3}}} > \epsilon} \lesssim \sqrt{\max_{j \ge d_2}\pr{E_j > 1}} < \delta. \label{exp-4.8}
\end{align}
Next, we want to show that
\begin{align*}
    \supx{\pr{\Delta_3^{-1}\sum_{j=d_2}^{\infty}\plr{E_j - \Ex{E_j} } \le x} - \Phi(x)} < \delta,
\end{align*}
for which it is enough to prove that 
\begin{align*}
    {\Delta_3}^{-1}\sum_{j=d_2}^{\infty}\plr{E_j - \Ex{E_j} }\darw N(0,1)
\end{align*}
This can be obtained through Lyapounov condition of CLT using the i.i.d. structure of $W_i$ (c.f. \eqref{wo-decomp-2}). Note that:
\begin{align*}
    \frac{\Delta_3^{-4} \sum_{i=1}^m \Ex{W_i - \Ex{W_i}}^4 }{\Delta_3^{-4}\plr{\sum_{i=1}^m \Ex{W_i - \Ex{W_i}}^2}^2}
      =  \frac{m \Ex{W_1 - \Ex{W_1}}^4 }{\plr{m^2 \Var{W_1}}^2}
    &\le \frac{m \sum_{j=d_2}^{\infty} p_j}{\plr{m^2 \sum_{j=d_2}^{\infty} p_j}^2} \\
    &\le \dfrac{(d_2 +1)^{\alpha -1 }}{m} \\
    & = O\plr{\dfrac{1}{d_2}}.
\end{align*}
Thus, we have 
\begin{align}
    \frac{Z_{3,\delta} - \Ex{Z_{3,\delta}}}{\Delta_3} \darw N(0,1). \label{z3-clt-wo}
\end{align}
This completes the second claim of \eqref{exp-4.4}.

Next, we show that 
\begin{align*}
    \frac{Z_{2,\delta} -\Ex{Z_{2,\delta}} }{\Delta_{2}} \darw N(0,1).
\end{align*}
Using Theorem (10.3.1) of \cite{al2006}, it is enough to show that for any $K>0$, 
\begin{align*}
&\sup_{|t| \le K}  \md{\Ex{e^{\iota t\plr{Z_{2,\delta}-\Ex{Z_{2,\delta}}}/\Delta_{2}}} - e^{-t^2/2}} \\ 
&\le \sup_{|t| \le K} \md{\Ex{e^{\iota t\plr{Z_{2,\delta}-\Ex{Z_{2,\delta}}}/\Delta_{2}}} - \prod_{j=d_1}^{d_2-1} \Ex{e^{\iota t\plr{X_j - \Ex{X_j}}/\Delta_{2}}}} \\
& \qquad\qquad + \sup_{|t| \le K} \md{ \prod_{j=d_1}^{d_2-1} E(e^{\iota t\plr{X_j - \Ex{X_j}}/\Delta_{2}}) - e^{-t^2/2}} = o(1), \labthis \label{exp-4.9}
\end{align*}
where $\ci = \sqrt{-1}$ .
Now, for the first quantity of \eqref{exp-4.9}, we use Theorem 10 of \cite{newman1984}. We also define that $\widetilde X_j = \indc{E_j = 0} = 1- X_j$. Then, $ \widetilde X_j \sim \Ber{(1-p_j)^m}$ for any $j \ge 1$. We further aim to show that 
\begin{align*}
\Delta_{2}^{-2} \sum_{d_1 \le j \le k \le d_2} \Cov{\widetilde X_j, \widetilde X_k} = \Delta_{2}^{-2} \sum_{d_1 \le j \le k \le d_2} \Cov{ X_j,  X_k} = o(1) 
\end{align*}
So we compute:
\begin{align}
    \Cov{\widetilde X_j, \widetilde X_k} = (1-p_j-p_k)^m - (1-p_j)^m(1-p_k)^m = A_{jk}. \label{exp-4.10}
\end{align}
Since $\clr{X_j}_{j \ge 1}$ are negatively associated (c.f. Lemma \ref{lem-1}), $\clr{\widetilde X_j}_{j \ge 1}$ are also negatively associated. Thus, we have that $A_{jk} \le 0$ and using Lemma \ref{lem-4} we further get the following:
\begin{align*}
\Delta_{2}^{-2} \md{\sum_{d_1 \le j < k \le d_2-1} \Cov{\widetilde X_j, \widetilde X_k}}
& \le \Delta_{2}^{-2} \sum_{d_1 \le j \le k \le d_2-1}m(p_j + p_k)^2 \exp\{-m(p_j + p_k)\}\\
& = \Delta_{3}^2 \Delta_{2}^{-2} \dfrac{(d_2-d_1)^2 \max\{(p_j + p_k)^2 \exp\{-m(p_j + p_k)\}\}}{ \sum_{j = d_1}^{\infty}p_j}\\
& \lesssim \dfrac{(C_\delta -c_\delta)^2 m^{2/\alpha}p_{d_1}^2 \exp\{-mp_{d_2}\}}{\sum_{j = d_1}^{\infty}p_j}\\
& \lesssim \dfrac{\frac{d_1^2}{(d_1 -1)^{2\alpha}} \exp\{- \frac{m}{(d_2+1)^\alpha}\}}{(d_1 + 1)^{1-\alpha}}\\
& \lesssim \plr{\frac{d_1+1}{d_1-1}}^{\alpha -1 }\\
& = O\plr{(d_1 -1)^{1-\alpha}}\\
& = o(1),
\end{align*}
where the last inequality uses the fact that $\alpha>2$.
Therefore, using Theorem 10 of \cite{newman1984}, we have
\begin{align*}
&\sup_{|t| \le K} \md{\Ex{e^{\iota t\plr{Z_{2,\delta}- \Ex{Z_{2,\delta}}}/\Delta_{2}}} - \prod_{j=d_1}^{d_2-1} \Ex{e^{\iota t\plr{X_j - \Ex{X_j}}/\Delta_{2}}}} \\
&\qquad\qquad\le t^2 \Delta_{2}^{-2} \md{\sum_{d_1 \le j < k \le d_2-1} \Cov{\widetilde X_j, \widetilde X_k}} < \delta. \labthis\label{exp-4.11}
\end{align*}
The above equation \eqref{exp-4.11} proves that $\clr{X_j}_{d_1 \le j \le d_2-1}$ are asymptotically independent. Hence, the Lyapunov CLT condition is computed as follows.
\begin{align*}
    \dfrac{\Ex{\Delta_{2}^{-4} \sum_{j=d_1}^{d_2-1}  \plr{X_j - \Ex{X_j}}^4}}{\plr{\sum_{j=d_1}^{d_2-1}\Var{\Delta_{2}^{-1}X_j}}^2}
    & \lesssim \dfrac{(d_2-d_1) \max_{d_1 \le j \le d_2-1} \clr{\Ex{ X_j - \Ex{X_j}}^4}}{\plr{\sum_{j=d_1}^{d_2-1}\Var{ X_j}}^2}\\
    & \lesssim \dfrac{(d_2-d_1) \max_{d_1 \le j \le d_2-1} \clr{\Ex{X_j^4} + (\Ex{X_j})^4}}{\slr{(d_2-d_1) \min_{d_1 \le j \le d_2-1} \clr{\Var{X_j}}}^2}\\
    & \lesssim \dfrac{(d_2-d_1)^{-1} \max \clr{(1-p_j)^m + (1-p_j)^{4m}}}{\min \clr{(1-p_j)^m (1-(1-p_j)^m)}}\\
    & \lesssim  \dfrac{(d_2-d_1)^{-1} \clr{(1-p_{d_1})^m + (1-p_{d_1})^{4m}}}{ (1-p_{d_2})^m (1-(1-p_{d_1})^m)}\\
    & = \dfrac{(C_\delta - c_\delta)^{-d} \exp(-{m}/{d_1^\alpha})\exp(-{4m}/{d_1^\alpha})}{ \exp(-{m}/{d_2^\alpha})(1-\exp(-{m}/{d_1^\alpha}))}\\
    & = O\plr{\frac{1}{(C_\delta-c_\delta)^d}}\\
    & =o(1),
\end{align*}
where the last inequality follows from the fact that $ 0 < c_\delta < C_\delta$. Therefore, we have 
\begin{align}
    \frac{Z_{2,\delta} - \Ex{Z_{2,\delta}}}{\Delta_2} \darw N(0,1). \label{z2-clt-wo}
\end{align}
This proves first claim of \eqref{exp-4.4} and completes the proof of (ii).

Next, we show that 
\begin{align*}
    \supx {\pr{\frac{Z_{2,\delta} + Z_{3,\delta} - \Ex{Z_{2,\delta} + Z_{3,\delta}}}{\Delta_m} \le x} - \Phix} \rightarrow 0,
\end{align*}
by obtaining the limiting distribution 
\begin{align*}
    \frac{Z_2 + Z_3-\Ex{Z_2 + Z_3}}{\Delta_m} \darw N(0,1).
\end{align*}
To this end, 
for any fixed $K > 0,$ we have
\begin{align*}
    &\sup_{|t| \le K} \md{ \Ex{e^{\iota t\plr{Z_{2,\delta}+Z_{3,\delta} - \Ex{Z_{2,\delta}+Z_{3,\delta}}}/\Delta_m}} - e^{-t^2/2}} \\
   & \le \sup_{|t| \le K} \md{ E(e^{\iota t\plr{Z_{2,\delta}+Z_{3,\delta} - \Ex{Z_{2,\delta}+Z_{3,\delta}}}/\Delta_m}) - \Ex{e^{\iota t\plr{Z_{2,\delta}-\Ex{Z_{2,\delta}}}/\Delta_m}}\Ex{e^{\iota t\plr{Z_{3,\delta} - \Ex{Z_{3,\delta}}}/\Delta_m}}} \\
   & \qquad\qquad + \sup_{|t| \le K} \md{ \Ex{e^{\iota t\plr{Z_{2,\delta}-\Ex{Z_{2,\delta}}}/\Delta_m}}\Ex{e^{\iota t\plr{Z_{3,\delta} - \Ex{Z_{3,\delta}}}/\Delta_m}} - e^{-t^2/2}}  \\
   & = \sup_{|t| \le K} |A_{1,\delta}(t)| + \sup_{|t| \le K} |A_{2,\delta}(t)|, \quad\text{say} . \labthis \label{exp-4.12}
\end{align*}
Note that, 
\begin{align*}
\Delta_m^{-2} \Cov{Z_{2,\delta}, Z_{3,\delta}} = \Delta_m^{-2} \clr{\plr{1- \sum_{j = d_1}^{\infty}p_j}^m - \plr{1- \sum_{j = d_1}^{d_2}p_j}^m \plr{1- \sum_{j = d_2}^{\infty}p_j}^m }
\end{align*}
Using Theorem 10 of \cite{newman1984}, we get
\begin{align*} 
    \suptK {A_{1,\delta}(t)} 
   & \le \sup_{|t| \le K} \md{ t^2 \Delta_m^{-2} \Cov{Z_{2,\delta}, Z_{3,\delta}}}  \\
   & \le \Delta_m^{-2} K^2 m \plr{\sum_{j = d_1}^{\infty}p_j}^2 e^{-\sum_{j = d_1}^{\infty}mp_j}\\
 & = \dfrac{K^2 m \plr{\sum_{j = d_1}^{\infty}p_j}^2 e^{-\sum_{j = d_1}^{\infty}mp_j}}{\sum_{j = d_2}^{\infty}mp_j} \frac{\Delta_{3}^2}{\Delta_m^2}\\
 & \lesssim \dfrac{(d_1-1)^{2-2\alpha} e^{-\frac{m}{(d_2+1)^{\alpha -1 }}}}{(d_1+1)^{1-\alpha}}\\
 & = O\plr{e^{-d_2 \plr{\frac{d_2}{d_1+1}}^{\alpha+1}}}< \epsilon. \labthis \label{exp-4.13}
\end{align*}
Note that
\begin{align*}
 |A_{2,\delta}(t)| & \le \Big|\Ex{e^{\iota t\plr{Z_{3,\delta} - \Ex{Z_{3,\delta}}}/\Delta_m}} \clr{ \Ex{e^{\iota t\plr{Z_{2,\delta} - \Ex{Z_{2,\delta}}}/\Delta_m}} - e^{-\Delta_{2}^2t^2/(2\Delta_m^2)}} \\
 &\qquad\qquad + e^{-\Delta_{2}^2 t^2/(2\Delta_m)^2}\clr{ \Ex{e^{\iota t\plr{Z_{3,\delta}-\Ex{Z_{3,\delta}}}/\Delta_m}}- e^{-\Delta_{3}^2 t^2/(2\Delta_m^2)}}\Big|\\
 & \le \md{\Ex{e^{\iota t\plr{Z_{2,\delta} - \Ex{Z_{2,\delta}}}/\Delta_m}} - e^{-\Delta_{2}^2 t^2/(2\Delta_m^2)}} \\
 & \qquad\qquad + \md{\Ex{e^{\iota t\plr{Z_{3,\delta}-\Ex{Z_{3,\delta}}}/\Delta_m}}- e^{-\Delta_{3}^2t^2/(2\Delta_m^2)} }
\end{align*}
Set $ t_2 = {t\Delta_{2}}/{\Delta_m}$ and $t_3 = {t\Delta_3}/{\Delta_m},$ then we have
\begin{align}
 |A_{2,\delta}(t)| \le \md{\Ex{e^{-\iota t_2 \plr{Z_{2,\delta}-\Ex{Z_{2,\delta}}}/ \Delta_{2}}} - e^{-t_2^2/2}} + \md{\Ex{e^{-\iota t_3 \plr{Z_{3,\delta}-\Ex{Z_{3,\delta}}}/ \Delta_3}} - e^{-t_3^2/2}}. \label{exp-4.14}
\end{align} 
Combining \eqref{z2-clt-wo} and Theorem 10.3.1 of \cite{al2006}, for a fixed $ 0 < K < \infty$ and $ \epsilon > 0$, there exits $ M_2(\epsilon) > 0$, such that
\begin{align*}
   \sup_{|t| \le K}\md{\Ex{e^{\iota t\plr{Z_{2,\delta} - \Ex{Z_{2,\delta}}}/\Delta_{2}}} - e^{-t^2/2}} < \epsilon, 
\end{align*}
for all $m \ge M_2(\epsilon)$.
Moreover, $|t_2| < |t|$ implies
\begin{align}
   \sup_{|t_2| \le K}\md{\Ex{e^{\iota t_2\plr{Z_{2,\delta} - \Ex{Z_{2,\delta}}}/\Delta_{2}}} - e^{-t^2_2/2}} \le  \sup_{|t| \le K}\md{\Ex{e^{\iota t\plr{Z_{2,\delta} - \Ex{Z_{2,\delta}}}/\Delta_{2}}} - e^{-t^2/2}} < \epsilon. \label{exp-4.15}
\end{align}
Similarly, following the above calculations, there exists $M_3(\epsilon) >0$ such that
\begin{align}
   \sup_{|t_3| \le K}\md{\Ex{e^{\iota t\plr{Z_{3,\delta}-\Ex{Z_{3,\delta}}}/\Delta_3}} - e^{-t^2_3/2}} \le  \sup_{|t| \le K} \md{\Ex{e^{\iota t\plr{Z_{3,\delta}-\Ex{Z_{3,\delta}}}/\Delta_3}} - e^{-t^2/2}} < \epsilon, 
   \label{exp-4.16}
\end{align}
for all $m \ge M_3(\epsilon)$.
Now, combining \eqref{exp-4.14}, \eqref{exp-4.15} and \eqref{exp-4.16} we have 
\begin{align*}
  \sup_{|t| \le K}  |A_{2,\delta}(t)| 
  &\le \sup_{|t_2| \le K}\md{\Ex{e^{\iota t_2\plr{Z_{2,\delta}-\Ex{Z_{2,\delta}}}/\Delta_{2}}} - e^{-t^2_2/2}} \\
  &\qquad\qquad+ \sup_{|t_3| \le K}\md{\Ex{e^{\iota t\plr{Z_{3,\delta}-\Ex{Z_{3,\delta}}}/\Delta_3}} - e^{-t^2_3/2}} 
  \le 2\epsilon.\labthis \label{exp-4.17}
\end{align*}
for all $ m \ge \max\{M_2(\epsilon), M_3(\epsilon)\},$. 
Therefore, combining \eqref{exp-4.13} and \eqref{exp-4.16} we have 
\begin{align*}
    \sup_{|t| \le K} \md{ \Ex{e^{\iota t\plr{Z_{2,\delta}+Z_{3,\delta}-\Ex{Z_{2,\delta}+Z_{3,\delta}}}/\Delta_m}} - e^{-t^2/2}} \le 3\epsilon.
\end{align*}
Now using Levy-Cramer continuity theorem (c.f. Theorem 10.3.4 of \cite{al2006}, we have 
\begin{align*}
\frac{Z_{2,\delta}+Z_{3,\delta}- \Ex{Z_{2,\delta}+Z_{3,\delta}}}{\Delta_m} \darw N(0,1),    
\end{align*}
which concludes the proof.
\end{proof}

\section{Proof of Section \ref{sec-subsampling}}\label{sec: res subsamp}

Before beginning the proof of Proposition \ref{prop3.8-subsampling}, we define a few notations.
We denote the U-statistic as 
\begin{align*}
   U_{0,k}^{(m)} = \frac{1}{\binom{m}{r_k}} \sum_{1 \le i_1 \le \cdots < i_{r_k} \le m} \gk(h_{i_1},\dots, h_{i_{r_k}}),
\end{align*}
for all $k = 1, \ldots, p$,
and 
\begin{align*}
 U_{j,k}^{(b)} = \frac{1}{\binom{b}{r_k}}  \sum_{ (l_1, \dots , l_{r_k}) \in \euA_{j,m,b}} \gk(h_{l_1},\dots, h_{l_{r_k}}),
\end{align*}
as the U-statistic defined on the set $\euA_{j,m,b}$ for all $j = 1, \ldots, N$ and $k = 1, \ldots, p$.
We further denote
\begin{align*}
    \widebar{U}^{(b)}_{k} = \frac{1}{N} \sum_{j=1}^N U_{j,k}^{(b)},
\end{align*}
for all $ k = 1, \ldots, p$ and $\widebar{U}^{(b)} = \plr{\widebar{U}^{(b)}_{1}, \ldots, \widebar{U}^{(b)}_{p}}$. Finally, we define the subsample based variance estimator as
\begin{align}
    \Xihsub = \frac{b}{N}\sum_{j=1}^N \plr{U_{j}^{(b)} -  U_{0}^{(m)}}\plr{U_{j}^{(b)} -  U_{0}^{(m)}}^\transpose.\label{xihsub}
\end{align}
Now, we proceed to prove Proposition \ref{prop3.8-subsampling}.

\begin{proof}[Proof of Proposition \ref{prop3.8-subsampling}]


We first observe that
\begin{align*}
    U^{(m)}_0 = \frac{1}{\binom{m}{b}} \sum_{j = 1}^{\binom{m}{b}}  U_{j}^{(b)}.
\end{align*}
Therefore, we decompose $\vhsub$ as follows
\begin{align*} 
    \Xihsub & = \frac{b}{N} \sum_{j=1}^N \plr{U_{j}^{(b)} -  U_{0}^{(m)}}\plr{U_{j}^{(b)} -  U_{0}^{(m)}}^\transpose\\
    & = \Xisub - 2b\plr{\widebar{U}^{(b)}- \mu}\plr{U_0^{(m)} - \mu}^\transpose + b\plr{U_0^{(m)} - \mu}\plr{U_0^{(m)} - \mu}^\transpose, \labthis \label{exp-sub-1}
\end{align*}
where 
\begin{align}
    \Xisub = \frac{b}{N} \sum_{j=1}^N \plr{U_{j}^{(b)} -  \mu}\plr{U_{j}^{(b)} -  \mu}^\transpose, \label{xisub}
\end{align}
and $\mu = \Ex{U_{0}^{(m)}}$. We first aim to show that $\Xihsub \parw \Lambda$.

Since the elements of $\Lambda$ are finite and the diagonal entries are bounded away from 0, we further have (cf. Theorem 4.5.2 of \cite{korolyuk2013theory})
\begin{align}
    \sqrt{m} \plr{{U_{0}^{(m)} -\mu}} \darw N(0,\Lambda). \label{exp-sub-2}
\end{align}
The assumption $b=o(m)$ and \eqref{exp-sub-2} further implies $b({U_{0}^{(m)} -\mu})({U_{0}^{(m)} -\mu})^\transpose \parw 0$.
Next, using equation (4.4) of \cite{blom1976some}, we have $\Var{\widebar{U}^{(b)}} = O(1/m)$. Therefore, we have $\sqrt{b} (\widebar{U}^{(b)} - \mu) = o_P(1)$, which further implies that $b(\widebar{U}^{(b)}- \mu)(U_0^{(m)} - \mu)^\transpose = o_P(1)$. Furthermore, from the definition of $\Lambda$, we have
\begin{align}
    \Ex{\Xisub} = \Ex{b\plr{U_{0}^{(b)} -\mu}\plr{U_{0}^{(b)} -\mu}^\transpose} \to \Lambda. \label{exp-sub-3}
\end{align}
 Therefore, from equation \eqref{exp-sub-1}, it is enough to show that $\Xisub - \Ex{\Xisub} \parw 0$.

Towards this, note that each element of $\Xisub$ in \eqref{xisub} can be represented as 
\begin{align*}
   \Xisub^{(l, l')} =  \frac{b}{N} \sum_{j=1}^N \plr{U_{j,l}^{(b)} -  \mu}\plr{U_{j,l'}^{(b)} -  \mu}^\transpose,
\end{align*}
for each $ 1 \le l \le l' \le p$. Therefore, applying equation (4.4) of \cite{blom1976some}, we have $\Var{\Xisub^{(l, l')}} = O(1/b)$. This implies $\Xisub^{(l, l')} \parw \Lambda^{(l,l')}$, for all $1 \le l \le l' \le p$. Finally, from \eqref{exp-sub-1} we have 
\begin{align}
    \Xihsub \parw \Lambda. \label{exp-sub-4}
\end{align}

Next, we use \eqref{exp-sub-4} to prove $\Lambdahsub \parw \Lambda$. We decompose $\Lambdahsub$ as follows
\begin{align*}
    \Lambdahsub &= \frac{b}{N} \sum_{j=1}^N \plr{S_j^{(b)} - \widebar{S}^{(b)}}\plr{S_j^{(b)} - \widebar{S}^{(b)}}^\transpose\\
    & = \frac{b}{N} \sum_{j=1}^N \plr{S_j^{(b)} - U_{j}^{(b)}}\plr{S_j^{(b)} - U_{j}^{(b)}}^\transpose + \Xihsub + b \plr{\widebar{U}^{(b)} - \widebar{S}^{(b)}}\plr{\widebar{U}^{(b)} - \widebar{S}^{(b)}}^\transpose + R_{b,m} \labthis\label{exp-sub-5}
\end{align*}
From the assumption $\Ex{S_{0,k}^{(b)}-U_{0,k}^{(b)}}^2 = o(1/b)$, it follows that 
\begin{align*}
    \frac{b}{N} \sum_{j=1}^N \plr{S_{j,k}^{(b)} - U_{j,k}^{(b)}}^2 \larw{1} 0,
\end{align*}
for all $ k=1, \ldots,p$, which further implies that the first term in \eqref{exp-sub-5} is $o_P(1)$. From \eqref{exp-sub-4}, we have $\Xihsub = O_P(1)$. Using Cauchy-Schwartz inequality we have,
\begin{align*}
    \Ex{\widebar{U}_{k}^{(b)} -\widebar{S}_{k}^{(b)} }^2 
    = \Ex{\frac{1}{N}\sum_{j=1}^N\plr{{U}_{j,k}^{(b)} -{S}_{j,k}^{(b)}} }^2 
    & \le \frac{1}{N}\Ex{\sum_{j=1}^N\plr{{U}_{j,k}^{(b)} -{S}_{j,k}^{(b)}}^2 }  \\
    & =\Ex{{U}_{0,k}^{(b)} -{S}_{0,k}^{(b)}}^2 \\
    & = o\plr{\frac{1}{b}},
\end{align*}
for all $ k = 1, \ldots, p$. Hence, the third term of \eqref{exp-sub-5} is $o_P(1)$. Lastly, using the fact $\Ex{\gk^4(h_1, \ldots,h_{r_k})} < \infty$ for all $k = 1, \ldots, p$ and Cauchy-Schwartz inequality, we further have $R_{b,m} = o_P(1)$. Hence, the result follows.

\end{proof}
\section{Additional Technical Results on Degree Filtering}\label{sec: addres deg-fil}

In this section, we present results for the degree-filtering subgraph densities under exponential tail decay $\clr{p_{(j)}}_{j \ge 1}$.

\begin{corollary} [Deletion Stability for Rainbow Subgraphs Under Exponential Decay]\label{cor-df-gen-exp}

Suppose that $H_{\fR}$ is a rainbow subgraph with $v$ vertices and $e$ edges.  Moreover, suppose that $p_{(j)} \ll \alpha^{-j}$, where $\alpha > 1$, and  (\ref{clt-assump-2}) holds. Then, 
    \begin{align*}
    \sqrt{m}[T(H_\fR) - T_d(H_\fR)] \stackrel{P}{\rightarrow} 0
    \end{align*}
    as $\mti$, provided $ d \ll m^{1-\varepsilon}$, where
\begin{align}
    \varepsilon = \max \clr{\frac{1}{\widebar{d}_{\widebar{H}}},\frac{1}{2(N_1 - 1)}, \frac{1}{2e}, \clr{\frac{1}{2N_{k}}}_{k=2}^{v-2} }, \label{beta-exp-def}
\end{align}
and $N_k$ is defined in (\ref{eq-n-k}). Consequently, 
\begin{align}
\sqrt{m}[T_d(H_\fR) - \theta(H_\fR)] \stackrel{\mathcal{L}}{\rightarrow} N(0, \sigma_{H_\fR}^2),
\end{align} where $\sigma_{H_\fR}^2$ is the asymptotic variance of $\sqrt{m}[T(H_\fR)-\theta(H_\fR)]$.
\end{corollary}

\begin{proof}[Proof of Corollary \ref{cor-df-gen-exp}]


We set $\delta_m = O(m^{-(1-\nu)})$ and $\delta_m' = \Theta(m^{-1})$, where $ \nu < \varepsilon$. 
Solving the equations $\alpha^{-x_1} = m^{-(1-\nu)}$ and $\alpha^{-x_2} = m^{-1}$, we get $|\cD_{v,m}| = x_2-x_1 = \nu\log(m)/\log(\alpha)$.
Also, note that $|\cA_{v,m}|^{1/v} \lesssim \delta_m \lesssim m^{1-\nu}$.
The remainder of the proof follows from Lemma \ref{clt-df-gen}.    
\end{proof}

\begin{corollary}[Central Limit Theorems for Degree-Filtered Colored Triangles Under Exponential Decay] \label{cor-exp-df-tri}
Suppose that $p_{(j)} \ll \alpha^{-j}$, where $\alpha > 1$. Then the following statement holds: 
\begin{enumerate}
    \item [(i)] $\sqrt{m}[T_{d}(\Delta_1)-\theta(\Delta_1)] \darw N(0,\sigma_{\Delta_1}^2)$ for Type 1 triangles holds provided $\pr{h \cap \cC^{(d)} \ne \phi }= o(1/m)$,
    \item [(ii)] $\sqrt{m}[T_{d}(\Delta_2)-\theta(\Delta_2)] \darw N(0,\sigma_{\Delta_2}^2)$ for Type 2 triangles holds if $ d \ll m^{1/3}$,
    \item [(iii)] $\sqrt{m}[T_{d}(\Delta_3)-\theta(\Delta_3)]\darw N(0,\sigma_{\Delta_3}^2)$ for Type 3 triangles holds if $d \ll m^{1/2}$,
\end{enumerate}
where $\theta(\Delta_k)$ and $\sigma_{\Delta_k}^2$ are the expected density and the variance of Type $k$ triangles in Theorem \ref{prop3.1-asymp-norm} for all $ k=1,2,3$.
\end{corollary}

\begin{proof}[Proof of Proposition \ref{cor-exp-df-tri}]

The proof follows from Lemma \ref{clt-df-thm-tri} with proper choice of $\delta_m$ and $\delta_m'$. 
We set $\delta_m = O(m^{-(1-\nu)})$ and $\delta_m' = \Theta(m^{-1})$. Solving the equations $1/\alpha^{x_1} = m^{-(1-\nu)}$ and $\alpha^{-x_2} = m^{-1}$, we get $|\cD_{3,m}| = x_2-x_1 = \nu\log(m)/\log(\alpha)$. Also, note that $|\cA_{3,m}|^{1/3} \lesssim \delta_m \lesssim m^{1-\nu}$. 
We set $0 < \nu<1/3$ for Type 2 triangles. Then the asymptotic bounds for \eqref{ass-df-2-2} reduces to the following:
\begin{align*}
    &\sqrt{m} |\cD_{3,m}|^3 \delta_m^2 = O\plr{ \frac{(\nu\log(m))^2}{m^{3/2-2\nu}}} = o(1), \\
     & \sqrt{m} |\cA_{3,m}|^{1/3} |\cD_{3,m}|^2 \delta_m^2= O\plr{\frac{(\nu\log(m))^2}{m^{1/2-\nu}}}= o(1),\\
     & \sqrt{m} |\cD_{3,m}| \delta_m=O\plr{\frac{\sqrt{m}\nu\log(m)}{m^{1/2-\nu}}} = o(1), \\
    & \sqrt{m} |\cD_{3,m}| \delta_m^{3/4}= O\plr{ \frac{(m\log(m))^{3/4}}{m^{1/4-3\nu/4}}} = o(1).
\end{align*}
where the inequalities follows from the fact that $d\ll m^{1/3}$.

We set $0 < \nu < 1/2$ for Type 3 triangles. Thus, the asymptotic bounds for \eqref{ass-df-3-2} reduces to the following:

\begin{align*}
    &\sqrt{m} |\cD_{3,m}| \delta_m =O\plr{ \frac{\nu\log(m)}{m^{1/2-\nu}}} =o(1),\\
    & \sqrt{m}  |\cD_{3,m}|^2 \delta_m^2 = O\plr{ \frac{(\nu\log(m))^2}{m^{3/2-2\nu}}  } = o(1), \\
    &  \sqrt{m}  |\cD_{3,m}|^3 \delta_m^3 = O\plr{\frac{(\nu\log(m))^3}{m^{5/2-3\nu}} } =o(1),
\end{align*}
where the inequalities follows from the fact that $d\ll m^{1/2}$.

\end{proof}

\section{Additional Simulation Results}\label{sec: add simu}

Figure \ref{fig-hist-two-stars-plot} and Figure \ref{fig-hist-tri-plot} presents the sampling and subsampling distributions of the scaled Type 2 two-stars and colorless triangle frequencies. The red curve represents the sampling distribution, where the Monte Carlo variance is treated as an estimate of the true variance. The subsampling distribution is shown by the blue curve. For reference, the standard normal distribution is overlaid with a black curve. The density plots indicate that the subsampling distribution closely approximates both the sampling distribution and the limiting standard normal distribution, providing empirical support for the validity of the subsampling approach in this setting.

\begin{figure}[H]
    \centering
    \begin{subfigure}{0.24\textwidth}
    \centering
        \includegraphics[width=\textwidth, page =1]{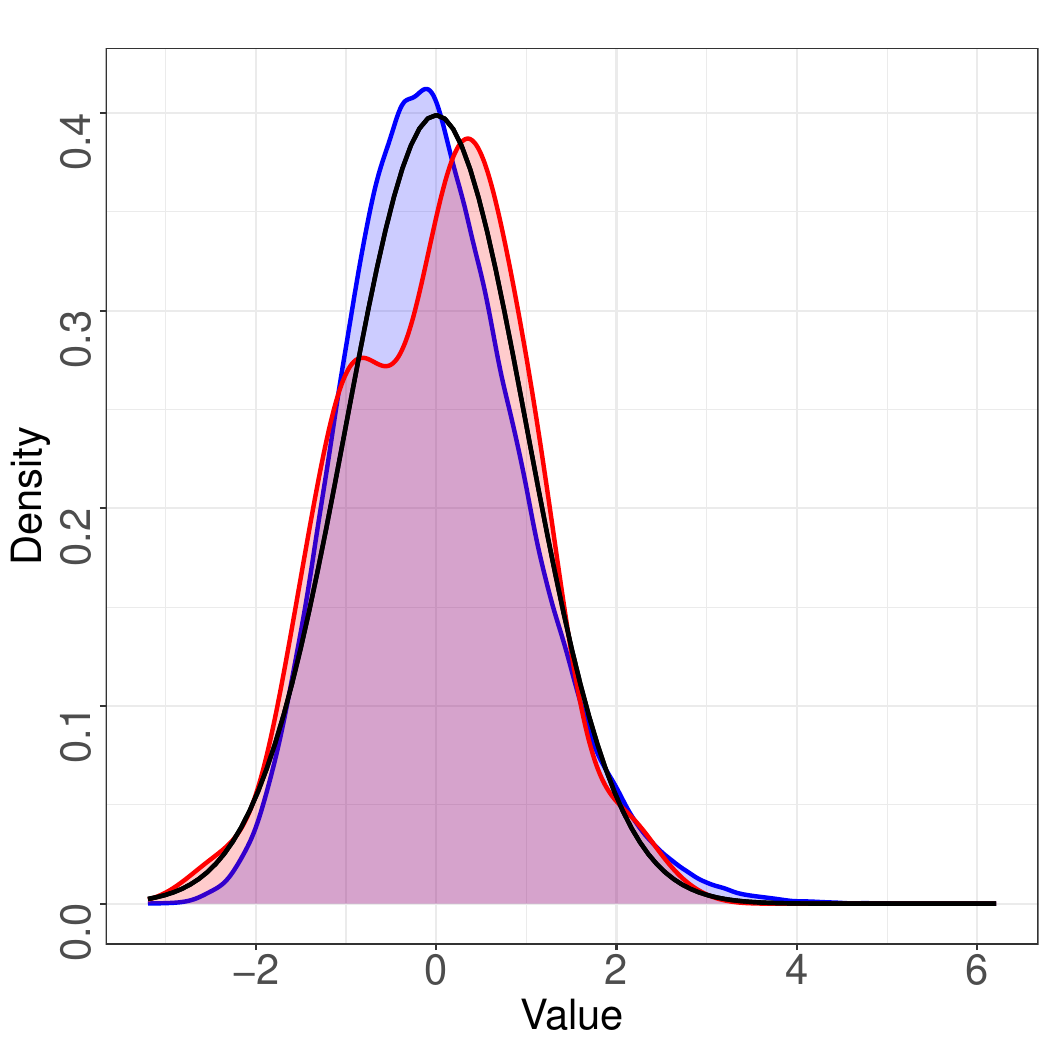}
        \subcaption{$m=500,C=1$}
    \end{subfigure}
    \begin{subfigure}{0.24\textwidth}
    \centering
        \includegraphics[width=\textwidth, page =2]{twostar_density.pdf}
        \subcaption{$m=500,C=1.2$}
    \end{subfigure}
    \begin{subfigure}{0.24\textwidth}
    \centering
        \includegraphics[width=\textwidth, page =3]{twostar_density.pdf}
        \subcaption{$m=500,C=1.4$}
    \end{subfigure}
    \begin{subfigure}{0.24\textwidth}
    \centering
        \includegraphics[width=\textwidth, page =4]{twostar_density.pdf}
        \subcaption{$m=500,C=1.6$}
    \end{subfigure}
    \begin{subfigure}{0.24\textwidth}
    \centering
        \includegraphics[width=\textwidth, page =5]{twostar_density.pdf}
        \subcaption{$m=500,C=1.8$}
    \end{subfigure}
    \begin{subfigure}{0.24\textwidth}
    \centering
        \includegraphics[width=\textwidth, page =6]{twostar_density.pdf}
        \subcaption{$m=500,C=2$}
    \end{subfigure}
    \begin{subfigure}{0.24\textwidth}
    \centering
        \includegraphics[width=\textwidth, page =7]{twostar_density.pdf}
        \subcaption{$m=500,C=2.2$}
    \end{subfigure}
        \begin{subfigure}{0.24\textwidth}
    \centering
        \includegraphics[width=\textwidth, page =8]{twostar_density.pdf}
        \subcaption{$m=500,C=2.4$}
    \end{subfigure}
    \begin{subfigure}{0.24\textwidth}
    \centering
        \includegraphics[width=\textwidth, page =9]{twostar_density.pdf}
        \subcaption{$m=1000,C=1$}
    \end{subfigure}
    \begin{subfigure}{0.24\textwidth}
    \centering
        \includegraphics[width=\textwidth, page =10]{twostar_density.pdf}
        \subcaption{$m=1000,C=1.2$}
    \end{subfigure}
    \begin{subfigure}{0.24\textwidth}
    \centering
        \includegraphics[width=\textwidth, page =11]{twostar_density.pdf}
        \subcaption{$m=1000,C=1.4$}
    \end{subfigure}
    \begin{subfigure}{0.24\textwidth}
    \centering
        \includegraphics[width=\textwidth, page =12]{twostar_density.pdf}
        \subcaption{$m=1000,C=1.6$}
    \end{subfigure}
    \begin{subfigure}{0.24\textwidth}
    \centering
        \includegraphics[width=\textwidth, page =13]{twostar_density.pdf}
        \subcaption{$m=1000,C=1.8$}
    \end{subfigure}
    \begin{subfigure}{0.24\textwidth}
    \centering
        \includegraphics[width=\textwidth, page =14]{twostar_density.pdf}
        \subcaption{$m=1000,C=2$}
    \end{subfigure}
    \begin{subfigure}{0.24\textwidth}
    \centering
        \includegraphics[width=\textwidth, page =15]{twostar_density.pdf}
        \subcaption{$m=1000,C=2.2$}
    \end{subfigure}
        \begin{subfigure}{0.24\textwidth}
    \centering
        \includegraphics[width=\textwidth, page =16]{twostar_density.pdf}
        \subcaption{$m=1000,C=2.4$}
    \end{subfigure}
   
    \caption{Sampling and subsampling distribution of the scaled Type 2 twos-stars density. The red and the blue curve corresponds to the sampling and a subsampling distributions. The black curve represents the density of standard normal distribution.}
    \label{fig-hist-two-stars-plot}
\end{figure}

\begin{figure}[H]
    \centering
    \begin{subfigure}{0.24\textwidth}
    \centering
        \includegraphics[width=\textwidth, page =1]{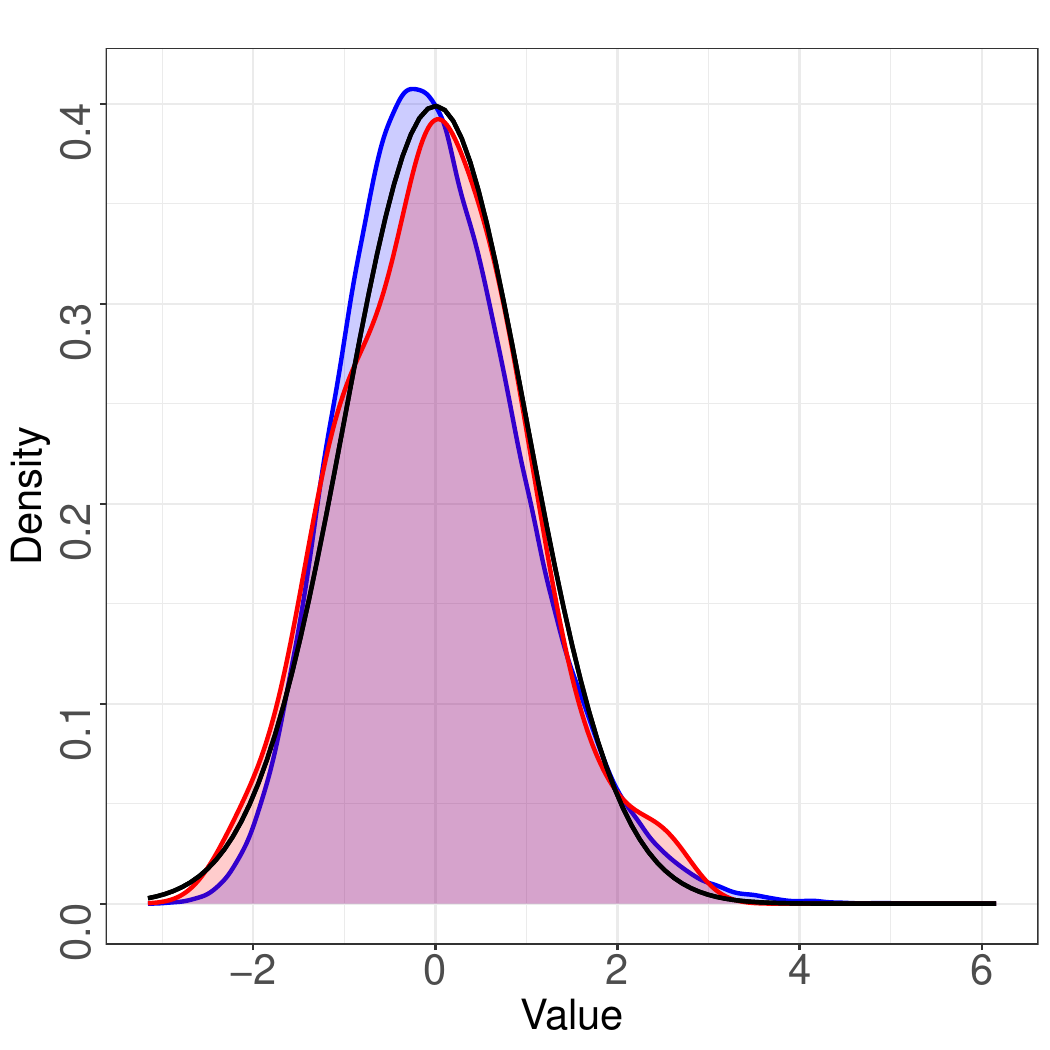}
        \subcaption{$m=500,C=1$}
    \end{subfigure}
    \begin{subfigure}{0.24\textwidth}
    \centering
        \includegraphics[width=\textwidth, page =2]{colorless_triangles_hist.pdf}
        \subcaption{$m=500,C=1.2$}
    \end{subfigure}
    \begin{subfigure}{0.24\textwidth}
    \centering
        \includegraphics[width=\textwidth, page =3]{colorless_triangles_hist.pdf}
        \subcaption{$m=500,C=1.4$}
    \end{subfigure}
    \begin{subfigure}{0.24\textwidth}
    \centering
        \includegraphics[width=\textwidth, page =4]{colorless_triangles_hist.pdf}
        \subcaption{$m=500,C=1.6$}
    \end{subfigure}
    \begin{subfigure}{0.24\textwidth}
    \centering
        \includegraphics[width=\textwidth, page =5]{colorless_triangles_hist.pdf}
        \subcaption{$m=500,C=1.8$}
    \end{subfigure}
    \begin{subfigure}{0.24\textwidth}
    \centering
        \includegraphics[width=\textwidth, page =6]{colorless_triangles_hist.pdf}
        \subcaption{$m=500,C=2$}
    \end{subfigure}
    \begin{subfigure}{0.24\textwidth}
    \centering
        \includegraphics[width=\textwidth, page =7]{colorless_triangles_hist.pdf}
        \subcaption{$m=500,C=2.2$}
    \end{subfigure}
        \begin{subfigure}{0.24\textwidth}
    \centering
        \includegraphics[width=\textwidth, page =8]{colorless_triangles_hist.pdf}
        \subcaption{$m=500,C=2.4$}
    \end{subfigure}
    \begin{subfigure}{0.24\textwidth}
    \centering
        \includegraphics[width=\textwidth, page =9]{colorless_triangles_hist.pdf}
        \subcaption{$m=1000,C=1$}
    \end{subfigure}
    \begin{subfigure}{0.24\textwidth}
    \centering
        \includegraphics[width=\textwidth, page =10]{colorless_triangles_hist.pdf}
        \subcaption{$m=1000,C=1.2$}
    \end{subfigure}
    \begin{subfigure}{0.24\textwidth}
    \centering
        \includegraphics[width=\textwidth, page =11]{colorless_triangles_hist.pdf}
        \subcaption{$m=1000,C=1.4$}
    \end{subfigure}
    \begin{subfigure}{0.24\textwidth}
    \centering
        \includegraphics[width=\textwidth, page =12]{colorless_triangles_hist.pdf}
        \subcaption{$m=1000,C=1.6$}
    \end{subfigure}
    \begin{subfigure}{0.24\textwidth}
    \centering
        \includegraphics[width=\textwidth, page =13]{colorless_triangles_hist.pdf}
        \subcaption{$m=1000,C=1.8$}
    \end{subfigure}
    \begin{subfigure}{0.24\textwidth}
    \centering
        \includegraphics[width=\textwidth, page =14]{colorless_triangles_hist.pdf}
        \subcaption{$m=1000,C=2$}
    \end{subfigure}
    \begin{subfigure}{0.24\textwidth}
    \centering
        \includegraphics[width=\textwidth, page =15]{colorless_triangles_hist.pdf}
        \subcaption{$m=1000,C=2.2$}
    \end{subfigure}
        \begin{subfigure}{0.24\textwidth}
    \centering
        \includegraphics[width=\textwidth, page =16]{colorless_triangles_hist.pdf}
        \subcaption{$m=1000,C=2.4$}
    \end{subfigure}
   
    \caption{Sampling and subsampling distribution of the scaled colorless triangle density. The red and the blue curve corresponds to the sampling and a subsampling distributions. The black curve represents the density of standard normal distribution.}
    \label{fig-hist-tri-plot}
\end{figure}

\section{Additional Real-Data Analysis Results}\label{sec: add realdat}

Table \ref{real-dat-net-tab1} summarizes the subgraph statistics for hypergraphs and binarized networks computed on training datasets.

\begin{table}[!htb]
\center
\caption{Network features in real-life train dataset}
\label{real-dat-net-tab1}
\begin{tabular}{@{}lclcccc@{}}
\toprule
 &  &                                  & JASA   & JCGS   & EMI      & Movies \\ \midrule
\multicolumn{7}{l}{\textbf{Hypergraph features:}}                                 \\
 &  & Type 2 clustering coefficient    & 0.2471  & 0.1463  & 0.4148    & 0.0235       \\
 &  & Type 2 two-stars frequency    & 0.0270  & 0.0157  & 1.7705    & 0.4451       \\
 \midrule
\multicolumn{7}{l}{\textbf{Binarized network features:}}                          \\
 &  & Scaled clustering coefficient           & 111.6898  & 104.8811  & 128.9333    & 137.7187       \\
 &  & Scaled two-star density           & 2.4393  & 2.3139  & 4.0028    & 3.9288       \\
 \bottomrule
\end{tabular}
\end{table}



\end{document}